\documentclass[12pt,aps,prd,nofootinbib,superscript address]{revtex4-2}
\usepackage[utf8]{inputenc}
\usepackage[T1]{fontenc}
\usepackage[dvipsnames,table,xcdraw]{xcolor}
\usepackage{epsf,amssymb,amsmath,amsthm,bbold,amsfonts}
\usepackage{mathrsfs}
\usepackage{mathtools}
\usepackage{tensor}
\usepackage{upgreek}
\usepackage{soul}
\usepackage{sourcecodepro}
\usepackage{accents}
\usepackage{mathtools}
\usepackage{ytableau}
\usepackage{tabu}
\usepackage{tabularray}
\usepackage{multirow}
\usepackage{calc}
\usepackage{scalerel,stackengine}
\usepackage{natbib} 
\usepackage{tocbasic}
\usepackage{notoccite}
\usepackage{appendix}
\usepackage{etoolbox}
\usepackage{hhline}
\usepackage{booktabs}
\usepackage{siunitx}
\usepackage{physics}
\usepackage{multirow,tabularx}
\usepackage{graphicx}

\usepackage[colorlinks,hyperindex,breaklinks]{hyperref}
\hypersetup{pageanchor=false,citecolor=red,urlcolor=red}

\def\a{\alpha}

\def\b{\beta}

\def\f{\frac}
\def\g{\gamma}

\def\m{\mu}
\def\n{\nu}

\def\p{\partial}

\def\s{\sigma}

\def\tr{\mathrm{tr}}

\def\be{\begin{equation}}
\def\ee{\end{equation}}

\def\bea{\begin{eqnarray}}
\def\eea{\end{eqnarray}}

\def\ba{\begin{array}}
\def\ea{\end{array}}

\def\bc{\begin{center}}
\def\ec{\end{center}}

\def\bl{\begin{flushleft}}
\def\el{\end{flushleft}}

\def\br{\begin{flushright}}
\def\er{\end{flushright}}

\def\bi{\begin{itemize}}
\def\ei{\end{itemize}}

\def\bt{\begin{tabular}}
\def\et{\end{tabular}}

\makeatletter

\newsavebox\myboxA
\newsavebox\myboxB
\newlength\mylenA

\newcommand*\xoverline[2][0.75]{%
    \sbox{\myboxA}{$\m@th#2$}%
    \setbox\myboxB\null% Phantom box
    \ht\myboxB=\ht\myboxA%
    \dp\myboxB=\dp\myboxA%
    \wd\myboxB=#1\wd\myboxA% Scale phantom
    \sbox\myboxB{$\m@th\overline{\copy\myboxB}$}%  Overlined phantom
    \setlength\mylenA{\the\wd\myboxA}%   calc width diff
    \addtolength\mylenA{-\the\wd\myboxB}%
    \ifdim\wd\myboxB<\wd\myboxA%
       \rlap{\hskip 0.5\mylenA\usebox\myboxB}{\usebox\myboxA}%
    \else
        \hskip -0.5\mylenA\rlap{\usebox\myboxA}{\hskip 0.5\mylenA\usebox\myboxB}%
    \fi}
\makeatother

\def\be{\begin{equation}}
\def\ee{\end{equation}}

\def\bea{\begin{eqnarray}}
\def\eea{\end{eqnarray}}

\def\f{\frac}

\def\p{\partial}

\usepackage{xspace}

\usepackage{xifthen}% Provides \isempty test
\newcommand*\diff{\mathrm{d}} % Straight differential
\newcommand*\ldiff[2][]{ \ifthenelse{\isempty{#1}}{ \diff
#2}{\diff^#1#2} \,} % Differential with measure; the mandatory argument
\let\limitint\int % Only when I provide explicit limits for the
\renewcommand{\int}{\limitint \!} % The standard integral should have
\usepackage{xspace}

\newrobustcmd{\PSALTer}{\textit{PSALTer}\xspace}
\newrobustcmd{\xAct}{\textit{xAct}\xspace}
\newrobustcmd{\xTensor}{\textit{xTensor}\xspace}
\newrobustcmd{\xCoba}{\textit{xCoba}\xspace}
\newrobustcmd{\xPerm}{\textit{xPerm}\xspace}
\newrobustcmd{\xCore}{\textit{xCore}\xspace}
\newrobustcmd{\xTras}{\textit{xTras}\xspace}
\newrobustcmd{\SymManipulator}{\textit{SymManipulator}\xspace}
\newrobustcmd{\RectanglePacking}{\textit{RectanglePacking}\xspace}
\newrobustcmd{\Inkscape}{\textit{Inkscape}\xspace}
\newrobustcmd{\Mathematica}{\textit{Mathematica}\xspace}
\newrobustcmd{\xPert}{\textit{xPert}\xspace}
\newrobustcmd{\MathGR}{\textit{MathGR}\xspace}
\newrobustcmd{\HiGGS}{\textit{HiGGS}\xspace}
\newrobustcmd{\Windows}{\textit{Microsoft Windows}\xspace}
\newrobustcmd{\Mac}{\textit{macOS}\xspace}
\newrobustcmd{\Linux}{\textit{Linux}\xspace}
\newrobustcmd{\GitHub}{\textit{GitHub}\xspace}
\newrobustcmd{\Bash}{\textit{bash}\xspace}
\newrobustcmd{\WolframLanguage}{\textit{Wolfram Language}\xspace}
\newrobustcmd{\CPP}{\textit{C++}\xspace}

\usepackage{xspace}
\usepackage{xstring}
\usepackage{titlesec}
\titleformat{\section}[block]{\normalfont\bfseries\centering}{\MakeUppercase\thesection}{1em}{\MakeUppercase}

\newrobustcmd{\pea}[1]{%
	\emph{#1}\textbf{.\ \ \ ---}
}
\titleformat{\paragraph}[runin]{\vspace{5pt}\normalfont\normalsize\bfseries}{\emph\theparagraph}{1em}{\pea}

\DeclareTOCStyleEntry[numwidth=20pt,linefill=\bfseries\TOCLineLeaderFill]{tocline}{section}
\DeclareTOCStyleEntry[entryformat=\textit,numwidth=10pt,linefill=\TOCLineLeaderFill]{tocline}{subsection}
\DeclareTOCStyleEntry[entryformat=\textit,numwidth=10pt,linefill=\TOCLineLeaderFill]{tocline}{subsubsection}

\newtheorem{theorem}{Theorem}[section]
\newtheorem{corollary}[theorem]{Corollary}
\theoremstyle{definition}
\newtheorem{definition}[theorem]{Definition}
\newtheorem*{remark}{Remark}

\usepackage[capitalize,nameinlink]{cleveref}
\crefname{paragraph}{paragraph}{paragraphs}
\Crefname{paragraph}{Paragraph}{Paragraphs}
\AddToHook{cmd/appendix/before}{\crefalias{section}{appendix}}

\allowdisplaybreaks

\usepackage{listings}
\usepackage{upquote}
\newrobustcmd{\dollar}{\mbox{\color{gray}\textbf\textdollar}}

 \lstdefinestyle{ascii-tree}{ 
    literate=
    {├}{{%
	\hphantom{\raisebox{0.5ex}{\rule{0.5ex}{1pt}}}%
	\smash{\raisebox{-1ex}{\rule{1pt}{1.1\baselineskip}}}%
	\raisebox{0.5ex}{\rule{0.5ex}{1pt}}%
	}}1 
    {│}{{%
	\hphantom{\raisebox{0.5ex}{\rule{0.5ex}{1pt}}}%
	\smash{\raisebox{-1ex}{\rule{1pt}{1.1\baselineskip}}}%
	\hphantom{\raisebox{0.5ex}{\rule{0.5ex}{1pt}}}%
	}}1 
    {─}{{%
	\smash{\raisebox{0.5ex}{\rule{1ex}{1pt}}}%
	}}1 
    {└}{{%
	\hphantom{\raisebox{0.5ex}{\rule{0.5ex}{1pt}}}%
	\smash{\raisebox{0.5ex}{\rule{1pt}{0.5\baselineskip}}}%
	\raisebox{0.5ex}{\rule{0.5ex}{1pt}}%
	}}1 
    { }{ }1 
 }

\lstdefinelanguage[PSALTer]{Mathematica}[]{Mathematica}{
	morekeywords=[2]{brown,PacletInstall,DrazinInverse,mx,wl,Association,KernelID,OptionsPattern,OptionValue,Private,Head,DistributeDefinitions,Portrait},
	morekeywords=[3]{blue,SymmetryOf,xAct,xAct`,xTensor,xCore,xPerm,xTras,SymManipulator,AllowUpperDerivatives,Antisymmetric,ChangeCovD,Christoffel,ConstantSymbolQ,ContractMetric,ContractMetrics,DefConstantSymbol,DefScalarFunction,DefTensor,delta,DependenciesOfTensor,IndicesOfVBundle,Labels,LI,LieDToCovD,MakeRule,NoScalar,OverDerivatives,ParamD,PD,PrintAs,Projected,ScalarFunctionQ,ScreenDollarIndices,SeparateMetric,SlotsOfTensor,Symmetric,SymmetryGroupOfTensor,ToCanonical,CommuteCovDs,xTensorQ,Zero,Tensors,ConstantSymbols,DefManifold,IndexRange,DefMetric,FlatMetric,SymCovDQ,DefCovD},
	morekeywords=[4]{green,a,Action,b,c,cartesian,CD,ChristoffelCD,ChristoffelPDcartesian,d,Def,DefField,DetG,e,EinsteinCCCD,EinsteinCCPDcartesian,EinsteinCD,En,Eps,epsilonG,etaDowncartesian,etaUpcartesian,f,g,G,h,i,j,k,KretschmannCD,l,m,M4,MaxLaurentDepth,Mo,n,o,p,P,ParticleSpectrum,PDcartesian,PoleResidue,PrintSourceAs,q,r,RicciCD,RicciPDcartesian,RicciScalarCD,RicciScalarPDcartesian,RiemannCD,RiemannPDcartesian,s,SchoutenCCCD,SchoutenCCPDcartesian,SchoutenCD,SchoutenPDcartesian,SymRiemannCD,SymRiemannPDcartesian,t,TangentM4,TetraG,TFRicciCD,TheoryName,ShowPropagator,TorsionCD,TorsionPDcartesian,u,v,V,w,WeylCD,x,y,z},
	morekeywords=[5]{red,Coupling1,Coupling2,Coupling3,Coupling4,Coupling5,Coupling6,ScalarField,VectorField,TwoFormField,MetricPerturbation,TetradPerturbation,SpinConnection,HigherSpinField,Connection,Alp0,Alp1,Alp2,Alp3,Alp4,Alp5,Alp6,Bet1,Bet2,Bet3,kT1,A0,A1,A2,A3,A4,A5,A6,A7,A8,A9,A10,A11,C0,C1,C2,C3,C4,C5,C6,C7,C8,C9,C10,C11,C12,C13,C14,C15,C16,D1,U0,U1,U2,U3,U4,V0,V1,V2,V3,V4,W0,W1,W2,W3,W4,MetricPerturbation,SquareMass1,SquareMass2,Mass1,kLambda,kR1,kR2,kR3,kR4,kR5,kR6,kR7,kR8,kT1,kT2,kT3,kT4,kT5,ThreeFormField}}

\lstdefinelanguage{Special}{%
morekeywords=[1]{%
ParticleSpectrographMassiveGravity.m,ValidateSymmetryField.m,ValidateSymmetryMode.m,ValidateTraceless.m,ValidateInverseField.m,ValidateInverseMode.m,ValidateSpatial.m,DefAllComponentValues.m,DefSummary.m,ValidateSO3Irreps.m,DefSymbol.m,MakeAutomaticallyTraceless.m,MakeUniquePartialDual.m,MakeUniqueQuadratic.m,RemoveContraction.m,MakeAutomaticallyNotAntisymmetric.m,RegisterFieldRank0.m,RegisterFieldRank1.m,RegisterFieldRank2.m,RegisterFieldRank2Antisymmetric.m,RegisterFieldRank2Symmetric.m,RegisterFieldRank3.m,RegisterFieldRank3Antisymmetric13.m,RegisterFieldRank3Symmetric12.m,RegisterFieldRank3Symmetric13.m,RegisterFieldRank3Symmetric23.m,RegisterFieldRank3TotallySymmetric.m,DecompositionTable.m,ExpansionTable.m,FieldMosaic.m,AllIndexConfigurations.m,AllocateTensorValues.m,RegisterFieldRank3Antisymmetric12.m,RegisterFieldRank3Antisymmetric23.m,RegisterFieldRank3TotallyAntisymmetric.m,AppendToField.m,DefFiducialField.m,DefSO3Irrep.m,PreComputeComponents.m,CombineRules.m,SummariseField.m,IsNegativeParitySpinTwo.m,CatalogueInvariant.m,CacheContexts.m,DefPlaceholderSpins.m,GenerateAnsatz.m,NormaliseRescalings.m,PoleToSquareMass.m,IsolatePoles.m,IrrationalQ.m,StripPoly.m,PartitionDeterminant.m,GaugeArtifactQ.m,MassiveGhost.m,MassiveAnalysisOfSector.m,SimplifyMasses.m,SquareMassQ.m,UnresolvedPoleQ.m,IndependentComponents.m,DefFreeSourceVariables.m,AllIndependentComponents.m,ConstraintComponentToLightcone.m,MakeConstraintComponentList.m,MakeFreeSourceVariables.m,RescaleNullVector.m,ExtractPart.m,ExtractReparameterisationMatrix.m,ExtractDenominator.m,DenominatorOfElement.m,ExtractSecularEquation.m,MasslessAnalysisOfTotal.m,AssistedResidue.m,NumericalLowestOrder.m,LowestOrder.m,ObtainTermReplacements.m,ProcessPart.m,ResidueFormula.m,NullResidue.m,PerformReduction.m,PrepareReduction.m,BasicNullResidue.m,TestDependency.m,ReparameteriseSources.m,ConstrainInLightcone.m,ExpressInLightcone.m,FullyCanonicalise.m,FullyExpandSources.m,MakeSaturatedMatrix.m,MatrixFromSymbolic.m,MatrixToSymbolic.m,Repartition.m,SourcesToComponents.m,ExaminePoleOrder.m,ConstructLightcone.m,ConvertLightcone.m,CarefullyOrthogonalise.m,ParameterisedNullVectorQ.m,GrabExpression.m,SimplifyIfSmall.m,GradualExpand.m,BatchExpanded.m,ConsolidateUnmakeSymbolic.m,GradualExpandSubTask.m,InitialExpand.m,ConsolidateFinalElement.m,IntermediateRules.m,MakeSymbolic.m,ManualPseudoInverse.m,DistributeConjugate.m,UnmakeSymbolic.m,ConjectureInverse.m,CreateList.m,IsNullVectorOfSpace.m,MinimalExampleCaseNullSpace.m,CommonNullVector.m,CleanNullVector.m,EnsureLinearInCouplings.m,RemoveReferencesToMomentum.m,SymbolicNullSpace.m,ConjectureNullSpace.m,NonTrivialDot.m,ToCovariantForm.m,GetHermitianPart.m,ConstructOperator.m,FourierLagrangian.m,TimedReduce.m,CombineAssociations.m,ConstructLinearAction.m,ConstructWaveOperator.m,ValidateMaxLaurentDepth.m,ValidateTheoryName.m,ConstructUnitarityConditions.m,ConstructSourceConstraints.m,GetDiagram.m,NRoot.m,ProtractList.m,StripFactors.m,ParticleRow.m,ParticleRows.m,CLICallStack.m,Status.m,GUICallStack.m,ShowIfSmall.m,TheoryRows.m,UnresolvedPoleRows.m,UnresolvedPoleRow.m,AbbreviationRow.m,SourceConstraintRows.m,UnitarityRow.m,Abbreviate.m,ReplaceRadicals.m,CarefulReplace.m,GetAbbreviationRules.m,ShortEnoughQ.m,StripRadicals.m,WignerGrid.m,AllParticleRows.m,ValidateLagrangian.m,ConstructMassiveAnalysis.m,ConstructMasslessAnalysis.m,ConstructSaturatedPropagator.m,ConstructSpectrograph.m,UpdateTheoryAssociation.m,Diagnostic.m,MonitorParallel.m,ParallelGrid.m,RaggedBlock.m,ReMagnify.m,NewFramed.m,NewParallelSubmit.m,Vectorize.m,CallStackEnd.m,NameOfFunction.m,StackStrip.m,ToNewCanonical.m,CallStackBegin.m,Colours.m,StackSetDelayed.m,DefGeometry.m,ReloadPackage.m,DefField.m,ParticleSpectrum.m,PSALTer.m},morekeywords=[2]{%
init.wl},morekeywords=[3]{%
},morekeywords=[4]{%
FieldKinematics.pdf,GitHubLogo.pdf,GitLabLogo.pdf,ParticleSpectrograph.pdf,FeynmanDiagramHexic.pdf,FeynmanDiagramQuadratic.pdf,FeynmanDiagramQuartic.pdf,FeynmanDiagramSpin0ParityEvenResolved.pdf,FeynmanDiagramSpin0ParityEvenUnresolvedDouble.pdf,FeynmanDiagramSpin0ParityEvenUnresolvedQuadruple.pdf,FeynmanDiagramSpin0ParityEvenUnresolvedQuintuple.pdf,FeynmanDiagramSpin0ParityEvenUnresolvedTriple.pdf,FeynmanDiagramSpin0ParityNoneResolved.pdf,FeynmanDiagramSpin0ParityNoneUnresolvedDouble.pdf,FeynmanDiagramSpin0ParityNoneUnresolvedQuadruple.pdf,FeynmanDiagramSpin0ParityNoneUnresolvedQuintuple.pdf,FeynmanDiagramSpin0ParityNoneUnresolvedTriple.pdf,FeynmanDiagramSpin0ParityOddResolved.pdf,FeynmanDiagramSpin0ParityOddUnresolvedDouble.pdf,FeynmanDiagramSpin0ParityOddUnresolvedQuadruple.pdf,FeynmanDiagramSpin0ParityOddUnresolvedQuintuple.pdf,FeynmanDiagramSpin0ParityOddUnresolvedTriple.pdf,FeynmanDiagramSpin1ParityEvenResolved.pdf,FeynmanDiagramSpin1ParityEvenUnresolvedDouble.pdf,FeynmanDiagramSpin1ParityEvenUnresolvedQuadruple.pdf,FeynmanDiagramSpin1ParityEvenUnresolvedQuintuple.pdf,FeynmanDiagramSpin1ParityEvenUnresolvedTriple.pdf,FeynmanDiagramSpin1ParityNoneResolved.pdf,FeynmanDiagramSpin1ParityNoneUnresolvedDouble.pdf,FeynmanDiagramSpin1ParityNoneUnresolvedQuadruple.pdf,FeynmanDiagramSpin1ParityNoneUnresolvedQuintuple.pdf,FeynmanDiagramSpin1ParityNoneUnresolvedTriple.pdf,FeynmanDiagramSpin1ParityOddResolved.pdf,FeynmanDiagramSpin1ParityOddUnresolvedDouble.pdf,FeynmanDiagramSpin1ParityOddUnresolvedQuadruple.pdf,FeynmanDiagramSpin1ParityOddUnresolvedQuintuple.pdf,FeynmanDiagramSpin1ParityOddUnresolvedTriple.pdf,FeynmanDiagramSpin2ParityEvenResolved.pdf,FeynmanDiagramSpin2ParityEvenUnresolvedDouble.pdf,FeynmanDiagramSpin2ParityEvenUnresolvedQuadruple.pdf,FeynmanDiagramSpin2ParityEvenUnresolvedQuintuple.pdf,FeynmanDiagramSpin2ParityEvenUnresolvedTriple.pdf,FeynmanDiagramSpin2ParityNoneResolved.pdf,FeynmanDiagramSpin2ParityNoneUnresolvedDouble.pdf,FeynmanDiagramSpin2ParityNoneUnresolvedQuadruple.pdf,FeynmanDiagramSpin2ParityNoneUnresolvedQuintuple.pdf,FeynmanDiagramSpin2ParityNoneUnresolvedTriple.pdf,FeynmanDiagramSpin2ParityOddResolved.pdf,FeynmanDiagramSpin2ParityOddUnresolvedDouble.pdf,FeynmanDiagramSpin2ParityOddUnresolvedQuadruple.pdf,FeynmanDiagramSpin2ParityOddUnresolvedQuintuple.pdf,FeynmanDiagramSpin2ParityOddUnresolvedTriple.pdf,FeynmanDiagramSpin3ParityEvenResolved.pdf,FeynmanDiagramSpin3ParityEvenUnresolvedDouble.pdf,FeynmanDiagramSpin3ParityEvenUnresolvedQuadruple.pdf,FeynmanDiagramSpin3ParityEvenUnresolvedQuintuple.pdf,FeynmanDiagramSpin3ParityEvenUnresolvedTriple.pdf,FeynmanDiagramSpin3ParityNoneResolved.pdf,FeynmanDiagramSpin3ParityNoneUnresolvedDouble.pdf,FeynmanDiagramSpin3ParityNoneUnresolvedQuadruple.pdf,FeynmanDiagramSpin3ParityNoneUnresolvedQuintuple.pdf,FeynmanDiagramSpin3ParityNoneUnresolvedTriple.pdf,FeynmanDiagramSpin3ParityOddResolved.pdf,FeynmanDiagramSpin3ParityOddUnresolvedDouble.pdf,FeynmanDiagramSpin3ParityOddUnresolvedQuadruple.pdf,FeynmanDiagramSpin3ParityOddUnresolvedQuintuple.pdf,FeynmanDiagramSpin3ParityOddUnresolvedTriple.pdf},morekeywords=[5]{%
FieldKinematics.png,GitHubLogo.png,GitLabLogo.png,ParticleSpectrograph.png},morekeywords=[6]{%
FeynmanDiagramHexic.tex,FeynmanDiagramQuadratic.tex,FeynmanDiagramQuartic.tex,Preamble.tex,FeynmanDiagramSpin0ParityEvenResolved.tex,FeynmanDiagramSpin0ParityEvenUnresolvedDouble.tex,FeynmanDiagramSpin0ParityEvenUnresolvedQuadruple.tex,FeynmanDiagramSpin0ParityEvenUnresolvedQuintuple.tex,FeynmanDiagramSpin0ParityEvenUnresolvedTriple.tex,FeynmanDiagramSpin0ParityNoneResolved.tex,FeynmanDiagramSpin0ParityNoneUnresolvedDouble.tex,FeynmanDiagramSpin0ParityNoneUnresolvedQuadruple.tex,FeynmanDiagramSpin0ParityNoneUnresolvedQuintuple.tex,FeynmanDiagramSpin0ParityNoneUnresolvedTriple.tex,FeynmanDiagramSpin0ParityOddResolved.tex,FeynmanDiagramSpin0ParityOddUnresolvedDouble.tex,FeynmanDiagramSpin0ParityOddUnresolvedQuadruple.tex,FeynmanDiagramSpin0ParityOddUnresolvedQuintuple.tex,FeynmanDiagramSpin0ParityOddUnresolvedTriple.tex,FeynmanDiagramSpin1ParityEvenResolved.tex,FeynmanDiagramSpin1ParityEvenUnresolvedDouble.tex,FeynmanDiagramSpin1ParityEvenUnresolvedQuadruple.tex,FeynmanDiagramSpin1ParityEvenUnresolvedQuintuple.tex,FeynmanDiagramSpin1ParityEvenUnresolvedTriple.tex,FeynmanDiagramSpin1ParityNoneResolved.tex,FeynmanDiagramSpin1ParityNoneUnresolvedDouble.tex,FeynmanDiagramSpin1ParityNoneUnresolvedQuadruple.tex,FeynmanDiagramSpin1ParityNoneUnresolvedQuintuple.tex,FeynmanDiagramSpin1ParityNoneUnresolvedTriple.tex,FeynmanDiagramSpin1ParityOddResolved.tex,FeynmanDiagramSpin1ParityOddUnresolvedDouble.tex,FeynmanDiagramSpin1ParityOddUnresolvedQuadruple.tex,FeynmanDiagramSpin1ParityOddUnresolvedQuintuple.tex,FeynmanDiagramSpin1ParityOddUnresolvedTriple.tex,FeynmanDiagramSpin2ParityEvenResolved.tex,FeynmanDiagramSpin2ParityEvenUnresolvedDouble.tex,FeynmanDiagramSpin2ParityEvenUnresolvedQuadruple.tex,FeynmanDiagramSpin2ParityEvenUnresolvedQuintuple.tex,FeynmanDiagramSpin2ParityEvenUnresolvedTriple.tex,FeynmanDiagramSpin2ParityNoneResolved.tex,FeynmanDiagramSpin2ParityNoneUnresolvedDouble.tex,FeynmanDiagramSpin2ParityNoneUnresolvedQuadruple.tex,FeynmanDiagramSpin2ParityNoneUnresolvedQuintuple.tex,FeynmanDiagramSpin2ParityNoneUnresolvedTriple.tex,FeynmanDiagramSpin2ParityOddResolved.tex,FeynmanDiagramSpin2ParityOddUnresolvedDouble.tex,FeynmanDiagramSpin2ParityOddUnresolvedQuadruple.tex,FeynmanDiagramSpin2ParityOddUnresolvedQuintuple.tex,FeynmanDiagramSpin2ParityOddUnresolvedTriple.tex,FeynmanDiagramSpin3ParityEvenResolved.tex,FeynmanDiagramSpin3ParityEvenUnresolvedDouble.tex,FeynmanDiagramSpin3ParityEvenUnresolvedQuadruple.tex,FeynmanDiagramSpin3ParityEvenUnresolvedQuintuple.tex,FeynmanDiagramSpin3ParityEvenUnresolvedTriple.tex,FeynmanDiagramSpin3ParityNoneResolved.tex,FeynmanDiagramSpin3ParityNoneUnresolvedDouble.tex,FeynmanDiagramSpin3ParityNoneUnresolvedQuadruple.tex,FeynmanDiagramSpin3ParityNoneUnresolvedQuintuple.tex,FeynmanDiagramSpin3ParityNoneUnresolvedTriple.tex,FeynmanDiagramSpin3ParityOddResolved.tex,FeynmanDiagramSpin3ParityOddUnresolvedDouble.tex,FeynmanDiagramSpin3ParityOddUnresolvedQuadruple.tex,FeynmanDiagramSpin3ParityOddUnresolvedQuintuple.tex,FeynmanDiagramSpin3ParityOddUnresolvedTriple.tex},morekeywords=[7]{%
convert_logos.sh,compile_diagrams.sh,generate_sources.sh},morekeywords=[8]{%
ASCIILogo.txt},morekeywords=[9]{%
GitHubLogo.svgz,GitLabLogo.svgz},morekeywords=[10]{%
LICENSE.md,README.md},morekeywords=[11]{cd,git,clone,r,cp,tree,mathematica,sudo,pacman,S},
keywordstyle=[1]\color{Green}\textbf,keywordstyle=[2]\color{Brown}\textbf,keywordstyle=[3]\color{RoyalBlue}\textbf,keywordstyle=[4]\color{Cyan}\textbf,keywordstyle=[5]\color{SkyBlue}\textbf,keywordstyle=[6]\color{Blue}\textbf,keywordstyle=[7]\color{Purple}\textbf,keywordstyle=[8]\color{Black}\textbf,keywordstyle=[9]\color{teal}\textbf,keywordstyle=[10]\color{teal}\textbf,keywordstyle=[11]\color{black}\textbf,alsoletter={./},basicstyle=\color{gray}\ttfamily,moredelim=[s][\color{gray}\textbf]{[user@system }{]}}

\definecolor{backing}{rgb}{0.88,1,0.88}

\input{MacrosMath.tex}
\usepackage{enumitem}
\setlist[description]{leftmargin=\parindent,labelindent=\parindent}

\begin{document}

\title{\Large The particle spectra of parity-violating theories:\\
      A less radical approach and an upgrade of~\PSALTer{}}

\author{\large W. Barker}
\email{barker@fzu.cz}
\affiliation{Central European Institute for Cosmology and Fundamental Physics, Institute of Physics of the Czech Academy of Sciences, Na Slovance 1999/2, 182 00 Prague 8, Czechia}
\affiliation{Astrophysics Group, Cavendish Laboratory, JJ Thomson Avenue, Cambridge CB3 0HE, UK}
\affiliation{Kavli Institute for Cosmology, Madingley Road, Cambridge CB3 0HA, UK}

\author{\large G.~K.~Karananas}
\email{georgios.karananas@physik.uni-muenchen.de}
\affiliation{Arnold Sommerfeld Center for Theoretical Physics, Ludwig-Maximilians-Universit\"at M\"unchen, Theresienstra\ss e 37, 80333 M\"unchen, Germany}

\author{\large H. Tu}
\email{ht486@cam.ac.uk}
\email{ht486@cantab.ac.uk}
\affiliation{Astrophysics Group, Cavendish Laboratory, JJ Thomson Avenue, Cambridge CB3 0HE, UK}
\affiliation{Kavli Institute for Cosmology, Madingley Road, Cambridge CB3 0HA, UK}

\begin{abstract}
\leavevmode
\begin{center}
\textbf{Abstract}
\end{center}
\par  Due to computational barriers, the effects of parity violation have so far been grossly neglected in gravitational model-building, leading to a serious gap in the space of prior models. We present a new algorithm for efficiently computing the particle spectrum for any parity-violating tensorial field theory. It allows to extract conditions for the absence of massive ghosts without resorting to any manipulation of radicals in cases where the particle masses are irrational functions of the Lagrangian coupling coefficients. We test it against several examples, among which is the most general parity-indefinite Einstein--Cartan/Poincar\'e gravity that propagates two healthy massive scalars (in addition to the massless graviton). Importantly, we upgrade the~\PSALTer{} software in the \WolframLanguage{} to accommodate parity-violating theories.~\PSALTer{} is a contribution to the \xAct{} project.
\end{abstract}

\maketitle

\newpage

\tableofcontents

\newpage

\section{Introduction}\label{Introduction}

\paragraph*{Motivation} 
In the metric-based approach to gravitation, at lowest order in the derivative expansion, there is only one diffeomorphism invariant operator of the metric field:\footnote{To be precise, there is also the cosmological constant, but it is irrelevant for the discussions here.} the Ricci scalar. This is parity invariant, as are its quantum corrections~\cite{tHooft:1974toh,tHooft:1973bhk,Deser:1974cy,Deser:1974nb,Goroff:1985th}. This does not mean, however, that gravity necessarily preserves parity. Nor is its apparent parity invariance motivated by any physical principle or phenomenological observation: rather, it is an artefact of the metrical formulation leading to a highly symmetric Riemann curvature tensor. Parity violation comes hand-in-hand with formulations of gravity where the metric and connection are independent from each other, such as in Einstein--Cartan (EC) theory or metric-affine gravity (MAG). This means that in principle, starting at linear order in the derivative expansion, there exist in these theories pseudoscalar operators whose~\emph{a priori} exclusion is not justified. Actually, it is exactly the inclusion of such operators that can lead to good phenomenology: for instance in EC theory and MAG, parity-breaking terms enable inflation of geometrical origin. Specifically, the inclusion of the pseudoscalar curvature is absolutely essential for the inflaton potential to have a plateau --- see~\cite{Salvio:2022suk,DiMarco:2023ncs,Karananas:2025xcv} for details. Moreover,  the way gravity is formulated --- and consequently whether or not it  violates parity --- also infiltrates and modifies the Standard Model of particle physics, as shown for instance in~\cite{Freidel:2005sn,Alexandrov:2008iy,Magueijo:2012ug,Shaposhnikov:2020gts,Shaposhnikov:2020frq,Shaposhnikov:2020aen,Karananas:2021zkl,Karananas:2024xja}. In spite of the nontrivial implications (e.g.~\cite{Chatzistavrakidis:2020wum,Chatzistavrakidis:2021oyp}) for phenomenology, the field dynamics of parity-nonpreserving gravitational theories remains poorly studied, with the one exception being EC/Poincar\'e gravity, cf.~\cite{Kuhfuss:1986rb,Diakonov:2011fs,Baekler:2010fr,Baekler:2011jt,Obukhov:2012je,Karananas:2014pxa,Obukhov:2017pxa,Blagojevic:2018dpz} for a non-exhaustive list of references. 

\paragraph*{In this paper} We present two advances in the study of the particle content of parity-violating tensor field theories in the weak-field limit, with particular relevance to theories of gravity:
\begin{description}
	 \item[New software] We implement the parity-violating extension (for EC/Poincar\'e gravity, a first attempt was made in the 80s~\cite{Kuhfuss:1986rb}, but systematized much later in~\cite{Karananas:2014pxa}) of the popular spin-projection operator (SPO) method~\cite{barnes1965lagrangian,Rivers:1964nfl,VanNieuwenhuizen:1973fi,Percacci:2020ddy} as an upgrade of the pre-existing \WolframLanguage{} framework for such calculations: Particle Spectrum for Any Tensor Lagrangian~(\PSALTer{}), a package for the \Mathematica{} software system, first presented in~\cite{Barker:2024juc}. It is a contribution to the open-source \xAct{} tensor computer algebra project~\cite{Martin-Garcia:2007bqa,Martin-Garcia:2008ysv,Martin-Garcia:2008yei,Nutma:2013zea,xCore:2018,xCoba:2020,SymManipulator:2021}. Earlier versions of~\PSALTer{} have been used already in~\cite{Barker:2023bmr,Barker:2024ydb,Barker:2024dhb,Barker:2024goa,Dyer:2024kvo,Barker:2025xzd}. The~\PSALTer{} software can be obtained from the public \GitHub{} repository \href{https://github.com/wevbarker/PSALTer}{\texttt{github.com/wevbarker/PSALTer}}. The majority of the upgraded software was written by generative pretrained transformers, highlighting the growing application of artificial intelligence in theoretical physics.
\item[New algorithm] We point out~\emph{a novel way to efficiently} derive conditions for the absence of ghosts, that does not involve any inversion of the wave operator nor the computation of residues of the propagator at massive poles. This is particularly useful for, but not limited to,  parity-violating theories which propagate two or more massive modes within one spin sector, since it completely bypasses dealing with radicals. It will be implemented in \PSALTer{} hopefully in the near future. 
\end{description}

\paragraph*{Organization of the paper} This paper is organized as follows. In~\cref{TheoreticalDevelopment}, we briefly outline some basics of the SPO formalism and discuss the standard condition(s) for absence of tachyons, as well as our novel, simplified, way to obtain  constraints for ghost-freedom. In~\cref{SymbolicImplementation}, the new features in~\PSALTer{} are illustrated, by obtaining the particle spectrographs of various parity-violating theories. This section is divided in two parts: in the first, we consider somewhat trivial toy-models involving~$p$-forms, which provide a pedagogical introduction. In the second part, we use \PSALTer{} to study the spectrum of the most general parity-indefinite EC theory that propagates the massless graviton and two scalars of gravitational origin --- we explicitly show that this model is free from ghosts and tachyons. We conclude in~\cref{Conclusions}, and there follow several technical appendices.

\paragraph*{Conventions} The conventions and notation are aligned as closely as possible with~\cite{Barker:2024juc}; any departures will be explicitly noted. We work exclusively in four spacetime dimensions. We use the `particle physics' or `mostly minus' signature for the Minkowski metric~$\eta_{\m\n} = {\rm diag}(1,-1,-1,-1)$ and we take~$\tensor{\epsilon}{_{0123}} = 1$ for the totally antisymmetric tensor. The various different kinds of indices and labels introduced are summarized in~\cref{IndicesTable}. 

\begin{table}[!h]
\caption{\label{IndicesTable} 
Indices and labels. Summation is assumed only over repeated spacetime indices.}
\begin{center}
\begin{tabularx}{\linewidth}{c|c|l}
\hline\hline
	Symbols & Values & Meaning\\
\hline
	$\mu,\ \nu$ &~$\{0,1,2,3\}$ & Minkowski spacetime indices\\
	$\SmashAcute\mu,\ \SmashAcute\nu$ &~$\{0,1,2,3\}$ & curved spacetime indices\\
	$\ovl\mu,\ \ovl\nu$ &~$\{0,1,2,3\}$ & Minkowski spacetime indices orthogonal to the momentum  \\
	$X,\ Y$ &~(symbolic) & Distinct Lorentz-covariant tensor fields\\
~$\mu_X,\ \nu_X$ &~$\{0,1,2,3\}^{\mathbb{Z}^{\geq}}$ & Collections of symmetrised spacetime indices associated with field~$X$\\
	$J,\ \Jp{}$ &~$\mathbb{Z}^{\geq}$ & Spin\\
	$P,\ \Pp{}$ &~$\{1,-1\}$ & Parity\\
~$\States{J}{P}{X}{i},\ \States{J}{P}{X}{j}$ &~$\mathbb{Z}^{>}$ & Multiple independent copies of a~$J^P$ state associated with field~$X$\\
	$s_{J},\ \Sp{}_{J}$ &~$\mathbb{Z}^{>}$ & Slots for masses associated with~$J$ if~$P$ is not a quantum number\\
	$s_{J^P},\ \Sp{}_{J^P}$ &~$\mathbb{Z}^{>}$ & Slots for masses associated with a given~$J^P$ state\\
	$\ovl\mu_{J^P},\ \ovl\nu_{J^P}$ &~$\{0,1,2,3\}^{\mathbb{Z}^{\geq}}$ & Collections of momentum-orthogonal indices associated with~$J^P$ \\ 
	$a_{J},\ \Ap{}_{J}$ &~$\mathbb{Z}^{>}$ & Slots for the null eigenvectors~(if any) of the~$J$ wave operator block~$\WaveOperatorJ{J}$\\
\hline\hline
\end{tabularx}
\end{center}
\end{table}

\section{Theoretical development}\label{TheoreticalDevelopment}

\paragraph*{Introducing the algorithm} Any tensorial field theory which has been expanded quadratically around Minkowski spacetime is described by the free action
\begin{align}\label{BasicPositionAction}
S =\int\mathrm{d}^4x\ \sum_X\FieldDown{X}{\mu}\left[\sum_Y\WaveOperatorTensorUpDown{X}{Y}{\mu}{\nu}\FieldUp{Y}{\nu}-\SourceUp{X}{\mu}\right] \ ,
\end{align}
where~\cref{BasicPositionAction} has the following components:
\begin{enumerate}
	\item The quantities~$\FieldDown{X}{\mu}$ are real tensor fields~(they are never pseudotensors).\footnote{Note that the use of pseudotensor fields can always be exchanged for tensorial~$\FieldDown{X}{\mu}$ by altering the parity of various contributing terms in the wave operator.} Different fields are distinguished by an index~$X$, and have a collection of spacetime indices~$\FieldIndices{X}{\mu}$, with or without some symmetry. The symmetries implemented in~\PSALTer{} so far for any~$X$ are: scalar~$\tensor{\zeta}{}$; vector~$\tensor{\zeta}{_{\mu}}$; the general tensor of second rank~$\tensor{\zeta}{_{\mu\nu}}$ and the special cases of antisymmetric~$\tensor{\zeta}{_{\mu\nu}}\equiv\tensor{\zeta}{_{[\mu\nu]}}$ or symmetric fields~$\tensor{\zeta}{_{(\mu\nu)}}$; the general tensor of third rank~$\tensor{\zeta}{_{\mu\nu\sigma}}$ and the special cases~$\tensor{\zeta}{_{\mu\nu\sigma}}\equiv\tensor{\zeta}{_{[\mu\nu\sigma]}}$,~$\tensor{\zeta}{_{(\mu\nu\sigma)}}$,~$\tensor{\zeta}{_{[\mu\nu]\sigma}}$,~$\tensor{\zeta}{_{[\mu|\nu|\sigma]}}$,~$\tensor{\zeta}{_{\mu[\nu\sigma]}}$,~$\tensor{\zeta}{_{(\mu\nu)\sigma}}$,~$\tensor{\zeta}{_{(\mu|\nu|\sigma)}}$ and~$\tensor{\zeta}{_{\mu(\nu\sigma)}}$.
	\item The quantity~$\WaveOperatorTensorUpDown{X}{Y}{\mu}{\nu}$ is the wave operator, which contains all the kinematical data. It is a real differential operator constructed from~$\tensor{\eta}{_{\mu\nu}}$,~$\tensor{\partial}{_{\mu}}$ and the totally antisymmetric tensor~$\tensor{\epsilon}{_{\mu\nu\sigma\lambda}}$. The presence of the latter indicates that the model is parity-violating, and marks a departure from the assumptions of~\cite{Barker:2024juc}. The wave operator in \PSALTer{} must have a homogeneous, linear parametrization in terms of some collection of coupling coefficients, i.e., its every term must be linear in these couplings.
	\item The quantities~$\SourceUp{X}{\mu}$ are real sources for the fields~$\FieldDown{X}{\mu}$, which inherit their index symmetries. In the presence of gauge redundancies, the sources must of course obey the corresponding constraints.
\end{enumerate}
So far,~\cref{BasicPositionAction} has been written in coordinate space. Accordingly, we introduce the momentum~$\Momenta{^\mu}$ and define (using~$x$ and~$k$ as shorthand) the Fourier transform and its inverse
\begin{equation}\label{FourierTransform}
	\FieldDown{X}{\mu}(k)\equiv\int\diff^4 k \, \exp\left(-i \Momenta{_\mu}\Coordinates{^\mu}\right) \FieldDown{X}{\mu}(x) \ ,\quad
	\FieldDown{X}{\mu}(x)\equiv\f{1}{(2\pi)^4}\int\diff^4 k \, \exp\left(i \Momenta{_\mu}\Coordinates{^\mu}\right) \FieldDown{X}{\mu}(k) \ .
\end{equation}
By using~\cref{FourierTransform} and the convolution theorem, it is possible to express~\cref{BasicPositionAction} in momentum space as
\begin{equation}\label{MomentumRepresentation}
	\begin{aligned}
	S 
	=
	\frac{1}{(2\pi)^4}\int\mathrm{d}^4k
	\sum_X
		\Bigg[&
		\FieldDown{X}{\mu}\left(-k\right)
		\sum_Y
		\WaveOperatorTensorUpDown{X}{Y}{\mu}{\nu}\left(k\right)
		\FieldUp{Y}{\nu}\left(k\right)
		\\
		&
		-\frac{1}{2}\Big(
		\FieldDown{X}{\mu}\left(-k\right)
		\SourceUp{X}{\mu}\left(k\right)
		+
		\FieldDown{X}{\mu}\left(k\right)
		\SourceUp{X}{\mu}\left(-k\right)
		\Big)
	\Bigg].
	\end{aligned}
\end{equation}
For real fields and sources it follows that~$\FieldDown{X}{\mu}(-k)\equiv\FieldDownConj{X}{\mu}(k)$ and~$\SourceDown{X}{\mu}(-k)\equiv\SourceDownConj{X}{\mu}(k)$. Accordingly, we will henceforth suppress the~$k$-dependence, so that~\cref{MomentumRepresentation} becomes
\begin{equation}\label{MomentumRepresentation_suppressed}
	S 
	=
	\frac{1}{(2\pi)^4}\int\mathrm{d}^4k
	\sum_X
	\Bigg[
		\FieldDownConj{X}{\mu}
		\sum_Y
		\WaveOperatorTensorUpDown{X}{Y}{\mu}{\nu}
		\FieldUp{Y}{\nu}
		-\frac{1}{2}\left(
		\FieldDownConj{X}{\mu}
		\SourceUp{X}{\mu}
		+
		\FieldUp{X}{\mu}\SourceDownConj{X}{\m}
		\right)
	\Bigg] \ ,
\end{equation}
and from the equations of motion which follow from~\cref{MomentumRepresentation_suppressed}, we can immediately read off a formal definition for the (scalar-valued) saturated propagator
\begin{equation}
\label{eq:saturated_propagator}
\Pi \equiv \SourceDownConj{X}{\mu}\WaveOperatorInverseTensorUpDown{X}{Y}{\mu}{\nu}\SourceUp{Y}{\nu} \ ,
\end{equation}
where in~\cref{eq:saturated_propagator} the quantity~$\WaveOperatorInverseTensorUpDown{X}{Y}{\mu}{\nu}$ is the `inverse' of the wave operator. This definition will need further careful treatment in~\cref{SpinProjection}, because gauge symmetries actually render the wave operator \emph{non-}invertible. For the moment,~\cref{eq:saturated_propagator} is safe heuristically because any singular parts of the `inverse' will be sandwiched between parts of the~$\SourceUp{Y}{\nu}$ which those same symmetries force to vanish. Thus, the only remaining parts of~$\Pi$ refer exclusively to (i.e., they are \emph{saturated} by) the unconstrained parts of~$\SourceUp{Y}{\nu}$, which are also the physical parts. Working in~$k$-space, the pole structure of the propagator encodes all important information about the particle content. The squares of the particle masses can be read off from the positions of the poles; if these are real and positive, a theory is tachyon-free. Meanwhile, the positivity of the pole residues guarantees freedom from ghosts. Note, however, that for parity-violating theories propagating more than one massive particle of spin-$J$, it is somewhat involved to extract the constraints for absence of ghosts from~\cref{eq:saturated_propagator} --- we discuss how this difficulty can be fully bypassed in~\cref{MassiveSpectrum}.

\subsection{Spin-projection operators}\label{SpinProjection}

\paragraph*{Fully covariant approach} The simplest way to proceed when it comes to tensorial field theories is to work in a fully covariant manner by using spin-projection operators (SPOs). We will develop the general theory and conventions for SPOs in detail in~\cref{ConstructionOfSpinProjectionOperators}, and in~\cref{ExplicitSpinProjectorOperators} we provide explicit formulae for those SPOs which are relevant for the examples in~\cref{SymbolicImplementation}. Here, we will only briefly outline the main ideas and notation. As their name suggests, SPOs break tensorial fields down into their irreducible parts with respect to the three-dimensional rotation group~$\othree$, i.e. into constituent parts of definite spin~$J$ and parity~$P$, which we denote by~$J^P$. The action of the SPOs, however, goes beyond mere decomposition: they constitute a complete basis for the possible ways in which the various~$J^P$ states from across all the fields interface with each other. Dealing first with simple decomposition, we use the labels~$\States{J}{P}{X}{i}$ to indicate the multiple independent states with a common~$J^P$ which are contained within the single field~$X$. Then, the `diagonal' SPOs
\be
\SPODownUp{J}{P}{X}{X}{i}{i}{\mu}{\nu} \ ,
\ee
form a complete basis
\begin{align}\label{Completeness}
	\FieldDown{X}{\mu}\equiv\sum_{J,P}\sum_{\States{J}{P}{X}{i}}\SPODownUp{J}{P}{X}{X}{i}{i}{\mu}{\nu}\FieldDown{X}{\nu}\implies
\sum_{J,P}\sum_{\States{J}{P}{X}{i}}\SPODownUp{J}{P}{X}{X}{i}{i}{\mu}{\nu} = \tensor*{\Delta}{^{\nu_X}_{\mu_X}} \ ,
\end{align}
where~$\tensor*{\Delta}{^{\nu_X}_{\mu_X}}$ in~\cref{Completeness} is the product of Kronecker symbols~$\tensor*{\delta}{^\nu_\mu}$, with one Kronecker factor for each index, in order. The diagonal SPOs are also positive-definite; recalling that parity is associated with having even or odd free spatial indices, and that spatial indices pick up a minus sign under contraction, this leads to the positivity condition
\begin{equation}\label{Positivity}
P\FieldDown{X}{\mu}^*
		\SPOUpDown{J}{P}{X}{X}{i}{i}{\mu}{\nu}
		\FieldUp{X}{\nu}\geq 0\  . 
\end{equation}
Going beyond decomposition, the general notation
\be
\PVSPODownUp{J}{P}{\Pp{}}{X}{Y}{i}{j}{\mu}{\nu} \ ,
\ee
encompasses `off-diagonal' SPOs in which different states~$\States{J}{P}{X}{i}$ and~$\States{J}{\Pp{}}{Y}{j}$ are interfaced, or mixed.\footnote{In this general notation, the first and second SPO arguments always share field labels~$X$ and~$Y$ with the first and second collections of Lorentz indices, respectively.} Note that for parity-violating tensorial field theories, the two states need only share a common~$J$~\cite{Kuhfuss:1986rb,Karananas:2014pxa}. This requirement is more broad than that considered in~\cite{Barker:2024juc}, where the SPOs were additionally required to have the same parity~$P = \Pp{}$. The parity-violating SPOs, i.e. those off-diagonal SPOs for which~$P\neq \Pp{}$, necessarily contain an odd power of the totally antisymmetric tensor. Together with the diagonal SPOs as a special case, all SPOs in addition to completeness~\cref{Completeness}, satisfy the following properties 
\begin{subequations}
	\begin{gather}
		\label{NewHermicity}
	\PVSPODownUp{J}{P}{\Pp{}}{X}{Y}{i}{j}{\mu}{\nu} = \PVSPOUpDown{J}{\Pp{}}{P}{Y}{X}{j}{i}{\nu}{\mu} \ ,
		\\
	\PVSPOUpDown{J}{P}{\Pp{}}{X}{Y}{i}{j}{\mu}{\nu}^* = 
	P\Pp{}
	\PVSPODownUp{J}{\Pp{}}{P}{Y}{X}{j}{i}{\nu}{\mu} \ , \label{ChequerHermitianIndices}
		\\
	\PVSPODownUp{J}{P}{\Pp{}}{X}{Y}{i}{j}{\mu}{\nu}
	\PVSPODownUp{\Jp{}}{\Ppp{}}{\Pppp{}}{Y}{Z}{k}{l}{\nu}{\sigma}
         = 	
	\tensor*{\delta}{_{jk}}
	\tensor*{\delta}{_{J\Jp{}}}
	\tensor*{\delta}{_{\Pp{}\Ppp{}}}
	\PVSPODownUp{J}{P}{\Pppp{}}{X}{Z}{i}{l}{\mu}{\sigma} \ , \label{NewOrthonormality}
	\end{gather}
\end{subequations}
where~\cref{NewHermicity,NewOrthonormality} encode symmetry and orthonormality, respectively. The condition in~\cref{ChequerHermitianIndices} implies Hermicity for parity-preserving SPOs, and skew-Hermicity for parity-violating SPOs. Within each~$J$ sector, our convention is to collect all the~$P=1$ states together, followed by the~$P=-1$ states. When the SPOs are arranged in this~$2\times 2$ block form, their (skew-)Hermicity is confined to (off-)diagonal blocks. We refer to this property as \emph{chequer-Hermicity}, and discuss its consequences for algebraic manipulations in~\cref{ChequerHermicity}.

\paragraph*{Wave operator} As a consequence of the properties of the SPO basis in~\cref{Completeness,NewHermicity,ChequerHermitianIndices,NewOrthonormality}, the spectral analysis can be performed in a far more convenient matrix representation. For example, the wave operator introduced in~\cref{BasicPositionAction} can be expressed as 
\begin{equation}\label{ToMatrices1}
	\WaveOperatorTensorUpDown{X}{Y}{\mu}{\nu}  =  \sum_{J,P,\Pp{}}\sum_{\States{J}{P}{X}{i},\States{J}{\Pp{}}{Y}{j}}\WaveOperatorJComponents{J}{P}{\Pp{}}{X}{Y}{i}{j}\PVSPOUpDown{J}{P}{\Pp{}}{X}{Y}{i}{j}{\mu}{\nu},
\end{equation}
where, for each~$J$, the chequer-Hermitian wave operator coefficient matrix~$\WaveOperatorJ{J}$ is indexed by all the~$\States{J}{P}{X}{i}$ and~$\States{J}{\Pp{}}{Y}{j}$. As explained above, the different spins do not mix at the level of the free action, and so the total wave operator coefficient matrix~$\WaveOperator{}$ assumes a block-diagonal form in~$J$-space
\begin{equation}\label{BuildWaveOperator}
\WaveOperator{} \equiv \bigoplus_{J}\WaveOperatorJ{J},
\quad
	\WaveOperatorJComponents{J}{P}{\Pp{}}{X}{Y}{i}{j}
	\PVSPOUpDown{J}{P}{\Pp{}}{X}{Y}{i}{j}{\mu}{\nu}
	   = 
	\SPOUpDown{J}{P}{X}{X}{i}{i}{\mu}{\sigma}
	\WaveOperatorTensorUpDown{X}{Y}{\sigma}{\lambda}
	\SPOUpDown{J}{\Pp{}}{Y}{Y}{j}{j}{\lambda}{\nu},
\end{equation}
where~\cref{ToMatrices1,BuildWaveOperator} are mutually consistent due to the properties in~\cref{Completeness,NewHermicity,ChequerHermitianIndices,NewOrthonormality}. Within each~$\WaveOperatorJ{J}$ block, however, various massive and massless particles can perfectly well co-exist: this complicates the calculations, as we will see in~\cref{MassiveSpectrum}.

\paragraph*{Saturated propagator} Having determined the structure of the general wave operator in terms of SPOs, the next objective is to obtain the saturated propagator~$\Pi$ given in~\cref{eq:saturated_propagator}. Comparison of~\cref{eq:saturated_propagator,BuildWaveOperator} suggests that~$\WaveOperatorInverse{}$ is the relevant quantity to compute. As anticipated, however, difficulties arise when one or more of the~$\WaveOperatorJ{J}$ are degenerate, so that~$\WaveOperatorInverse{}$ may not exist. As explained in~\cite{Barker:2024juc}, the dimensionality of the kernel of each~$\WaveOperatorJ{J}$, multiplied by the multiplicity~$2J+1$, and summed over all~$J$, yields the total number of gauge generators for the free theory in~\cref{BasicPositionAction}. It is also explained in~\cite{Barker:2024juc} that the most elegant approach to inverting such singular matrix is via the \emph{Moore--Penrose pseudoinverse}~\cite{Moore:1920,Penrose:1955} --- see~\cref{MoorePenrose}.\footnote{Pseudoiversion may be understood as a systematic approach to inverting the non-singular parts of a matrix. In analyses prior to~\cite{Barker:2024juc}, a less controlled but physically valid procedure --- effectively equivalent to pseudoinversion --- was to simply invert the largest non-singular submatrices of each~$\WaveOperatorJ{J}$.} Denoting this by~$\MoorePenroseJ{J}$ we can replace~\cref{eq:saturated_propagator} with
\begin{equation}
\label{eq:saturated_propagator_SPOs}
\Pi  =  \sum_{X,Y}\sum_{J,P,\Pp{}}\sum_{\States{J}{P}{X}{i},\States{J}{\Pp{}}{Y}{j}}\tensor{\left[\MoorePenroseJ{J}\right]}{_{\States{J}{P}{X}{i}\States{J}{\Pp{}}{Y}{j}}}
	\SourceDownConj{X}{\mu}
	\PVSPOUpDown{J}{P}{\Pp{}}{X}{Y}{i}{j}{\mu}{\nu}\SourceUp{Y}{\nu} \ ,
	\quad
	\MoorePenroseJ{J}
	 = 
	\Chequer{\MoorePenroseJPP{J}}{\MoorePenroseJPM{J}}{\MoorePenroseJMP{J}}{\MoorePenroseJMM{J}},
\end{equation}
Note that the Moore--Penrose pseudoinverse of a chequer-Hermitian matrix is also chequer-Hermitian. This result is developed across~\cref{ChequerHermicity,MoorePenrose,OperatorCoefficientMatrices}, which uses the shading scheme in~\cref{eq:saturated_propagator_SPOs} to indicate a chequer-Hermitian block structure. 

\subsection{Massless spectrum}\label{MasslessSpectrum}

\paragraph*{No-ghost condtion} When massless particles are involved, SPOs --- irrespectively of whether these are parity-preserving or parity-violating --- should be employed with certain care. This is because they are constructed out of the usual transverse and longitudinal projectors, which are not well-defined in the massless limit. Specifically, all SPOs incorporate powers of~$\Momenta{^\mu}$ into their tensor structures; these are necessarily accompanied by negative powers of~$k$ for reasons of normalisation. These negative powers can lead to spurious massless poles in~$\Pi$. The safest approach is to explicitly compute the full saturated propagator in~\cref{eq:saturated_propagator_SPOs} by expanding both the SPOs and the sources in terms of their tensorial components in some fiducial frame. Being artefacts of a `poor' choice of basis, spurious singularities cancel out and the limit on the lightcone can be carefully taken. Absence of massless ghosts requires that the corresponding residue be positive
\be
\NewRes{k^2}{0} \Pi > 0 \ .
\ee
This brute-force procedure (see e.g.~\cite{Sezgin:1979zf,Karananas:2014pxa,Barker:2024juc}) is not especially elegant, and indeed the massless spectrum typically accounts for more wallclock time in~\PSALTer{}~(which uses this method) than its massive counterpart, which we discuss next in~\cref{MassiveSpectrum}. Moreover, the version of the algorithm presented in~\cite{Barker:2024juc} does not require any particular modification when parity-violating operators are introduced.

\subsection{Massive spectrum}\label{MassiveSpectrum}

\paragraph*{No-tachyon condition} 

Consider the spin-sector~$J$ for which the coefficient matrix~$\WaveOperatorJ{J}$ is not block-diagonal, i.e. there is parity violation. This sector may have various simple zeroes\footnote{As shown in~\cite{Barker:2024juc}, non-simple zeroes always signal the presence of ghost modes.} in~$k^2\equiv\tensor{k}{^\mu}\tensor{k}{_\mu}$, at the positions of the square masses~$\PVSquareMass{J}{s}\neq 0$, where~$\PVMasses{J}{s}$ is a label for the various masses associated with~$J$. These correspond to the zeroes of the determinant of the largest non-degenerate submatrix of~$\WaveOperatorJ{J}$ or, equivalently, the zeroes of its pseudodeterminant. In general the formulae for these zeroes in terms of the Lagrangian coupling coefficients may not be expressible using rational functions, requiring either radicals or transcendental functions. We are interested in stable and non-tachyonic states, i.e. the masses ~$\PVSquareMass{J}{s}$ must be real, and positive~\cite{Karananas:2014pxa,Blagojevic:2018dpz}
\begin{equation}\label{eq:positive_masses}
\PVSquareMass{J}{s}>0\quad \forall J , \PVMasses{J}{s}\ .
\end{equation}
Note that~\cref{eq:positive_masses} is a generalisation of the parity-preserving case in~\cite{Sezgin:1979zf,Barker:2024juc}, for which masses are confined to a given~$J^P$ sector and labelled by~$\Masses{J}{P}{s}$.

\paragraph*{No-ghost condition} The no-tachyon condition in~\cref{eq:positive_masses} must be complemented with a no-ghost criterion as well. There exist two fully equivalent ways, explained shortly, to determine if the massive modes have healthy kinetic terms. Which one to use crucially depends on how many such modes are present in each spin sector. The well-known approach is to demand that the residues of the propagator in~\cref{eq:saturated_propagator_SPOs}, evaluated at all the unique massive poles~$\PVSquareMass{J}{s}$, be positive. We show in~\cref{NoGhostCriterionAppendix} how the chequer-Hermitian structure in~\cref{eq:saturated_propagator_SPOs} reduces this criterion to
\begin{equation}
\label{NoGhostCriterionMixed} 
\NewRes{k^2}{\PVSquareMass{J}
 {s}}\left(\tr \MoorePenroseJP{J}{+}-\tr \MoorePenroseJP{J}
 {-}\right)>0\quad \forall J, \PVMasses{J}{s} \ .
\end{equation} 
The formula in~\cref{NoGhostCriterionMixed} was not known previously, though it is consistent with the known formula --- see~\cref{NoGhostCriterion2} --- in the parity-preserving case, and is fully equivalent to the condition used in~\cite{Karananas:2014pxa}. Due to it being a straightforward extension of the parity-preserving criterion, it is precisely~\cref{NoGhostCriterionMixed} which is implemented in the upgrade to~\PSALTer{}. In practice, however, this criterion becomes extremely difficult to apply in those cases where the massive spectrum comprises more than one fields, and~$\PVSquareMass{J}{s}$ are expressible in terms of radicals or transcendental functions of the Lagrangian coupling coefficients. We will see in~\cref{ECSpectroscopy}, for example, that this is exactly what happens in EC theory. In fact, the complications are due to the presence of massless modes and -- more specifically -- because of their kinetic mixings with the massive ones.\footnote{It can also happen that, due to gauge redundancies, a sector does not contain massless propagating particles; this is the situation with the spin-one fields of the general quadratic EC gravity~\cite{Karananas:2014pxa,Blagojevic:2018dpz}, see~\cref{app:PGT_comparison}. Also in this case, it is simpler to not compute the residues of the propagator at the massive poles.} We propose now a novel way for determining the no-ghost condition in such cases. By re-ordering rows and columns in~$\WaveOperatorJ{J}$ one can obtain a matrix with the canonical block structure\,\footnote{Note that~\cref{eq:canonical_block_structure} is not shaded, because the re-ordered coefficient matrix will not, in general, be chequer-Hermitian. However, we can always make the diagonal blocks~$\WaveOperatorJVV{J}$ and~$\WaveOperatorJLL{J}$ chequer-Hermitian, as we will see in~\cref{ECSpectroscopy}.}
\begin{equation}\label{eq:canonical_block_structure}
\ReducedWaveOperatorJ{J}=\NormalMatrix{\WaveOperatorJVV{J}}{\WaveOperatorJVL{J}}{\WaveOperatorJLV{J}}{\WaveOperatorJLL{J}} \ ,
\end{equation}
such that massive modes of negative parity occupy the upper-left submatrices of~$\WaveOperatorJVV{J}$, the massless modes are contained in~$\WaveOperatorJLL{J}$, and the kinetic mixings between massive and massless modes are contained in~$\WaveOperatorJVL{J}$ and~$\WaveOperatorJLV{J}$. The key insight is that these mixings can be rotated away, i.e.~$\ReducedWaveOperatorJ{J}$ can be block-diagonalized by the transformation
\begin{equation}\label{eq:tilde_OJ}
	\ReducedWaveOperatorJ{J}=\NormalMatrix{\mathsf{1}}{\WaveOperatorJVL{J}\WaveOperatorJLLInverse{J}}{\mathsf{0}}{\mathsf{1}} 
	\ReducedWaveOperatorRotatedJ{J} 
	\NormalMatrix{\mathsf{1}}{\mathsf{0}}{\WaveOperatorJLLInverse{J}\WaveOperatorJLV{J}}{\mathsf{1}},\quad 
	\ReducedWaveOperatorRotatedJ{J} =
	\NormalMatrix{\WaveOperatorRotatedJVV{J}}{\mathsf{0}}{\mathsf{0}}{\WaveOperatorJLL{J}} ,
\end{equation}
where the first block in~\cref{eq:tilde_OJ} --- which encodes all the information about the massive spectrum --- decomposes as a polynomial in~$k$ according to\,\footnote{In the literature,~$\WaveOperatorRotatedJVV{J}$ is called the \emph{Schur complement}.}
\begin{equation}\label{DiagonalizedWaveOperator}
\WaveOperatorRotatedJVV{J} = \WaveOperatorJVV{J}-\WaveOperatorJVL{J}\WaveOperatorJLLInverse{J}\WaveOperatorJLV{J}
\equiv
\KineticJ{J}k^2+\MassJ{J} \ .
\end{equation}
Here~$\KineticJ{J}$ and~$\MassJ{J}$ in~\cref{DiagonalizedWaveOperator} are the kinetic- and mass- 
 matrices of the massive modes, respectively, which depend exclusively on the Lagrangian coupling coefficients. The no-ghost criterion simply becomes that the kinetic matrix be negative-definite
\begin{equation}
\label{eq:kin_matrix_negative}
\KineticJ{J}< 0 \ ,\quad \forall J \ .
\end{equation}
The condition in~\cref{eq:kin_matrix_negative} is necessarily equivalent to that in~\cref{NoGhostCriterionMixed}. Note that if the kinetic terms are canonicalised, then the masses correspond to the eigenvalues of~$\PostMassJ{J}\equiv-\KineticInverseJ{J}\MassJ{J}$, which coincide with the~$\PVSquareMass{J}{s}$ as can be straightforwardly verified. Note that our new method is not yet implemented in~\PSALTer{}, but we illustrate its use in~\cref{ECSpectroscopy}.

\section{Examples with code}\label{SymbolicImplementation}

\paragraph*{How to use this section} We now demonstrate the capabilities of the algorithm presented in~\cref{TheoreticalDevelopment}, when it is implemented in the latest version of the \PSALTer{} software. There are two types of examples we consider here. The first are various toy-models involving~$p$-form fields which were chosen simply to illustrate the versatility of the software. The other example we study is a specific subclass of Einstein--Cartan gravity that propagates two massives scalars --- this has been carefully chosen not only because the the \PSALTer{} results can be cross-checked analytically, see~\cref{app:cross-checks}, but also because such a theory can have interesting applications in cosmology and particle physics. Note that the sources for these examples can be found in the supplement~\cite{Supplement}.

\paragraph*{Syntax highlighting} From this point on we will occasionally present code listings which have syntax highlighting, our conventions for which are as follows. Symbols belonging to the \WolframLanguage{}~(\Mathematica{}) are \lstinline!brown!, those belonging to \xAct{} are \lstinline!blue!, those belonging to~\PSALTer{} are \lstinline!green!, and those which will be defined as part of the user session are \lstinline!red!. The start of each new input cell in a \Mathematica{} notebook is denoted by `\lstinline!In[#]:=!', and the start of each output cell is denoted by `\lstinline!Out[#]:=!'. Comments within the code appear in \lstinline!(*gray*)! and strings~(which are not symbols) are shown in \lstinline!"orange"!.

\paragraph*{Loading the software} As explained in~\cite{Barker:2024juc}, the software is loaded via the \lstinline!Get! command:
\begin{lstlisting}
In[#]:= Get["xAct`PSALTer`"]; 
\end{lstlisting}
During the loading process, several other \Mathematica{} dependencies are loaded, including \xTensor{}~\cite{Martin-Garcia:2007bqa,Martin-Garcia:2008yei}, \SymManipulator{}~\cite{SymManipulator:2021}, \xPerm{}~\cite{Martin-Garcia:2008ysv}, \xCore{}~\cite{xCore:2018}, \xTras{}~\cite{Nutma:2013zea} and \xCoba{}~\cite{xCoba:2020} from \xAct{}. The~\PSALTer{} package will next pre-define the geometric environment; a flat manifold \lstinline!M4! with metric \lstinline!G!, the totally antisymmetric tensor \lstinline!epsilonG!, and the derivative on flat spacetime \lstinline!CD!. The lower-case Latin alphabet \lstinline!a!, \lstinline!b!,$\ldots$,\lstinline!z!, is fully reserved for Minkowski spacetime indices, which automatically format as Greek letters~$\alpha$,~$\beta,\ldots,\zeta$. For example, \lstinline!G[-m,-n]! formats as~$\tensor{\eta}{_{\mu\nu}}$ and \lstinline!CD[-m]@! formats as~$\PD{_\mu}$, whilst \lstinline!epsilonG[-m,-n,-r,-s]! formats as~$\tensor{\epsilon}{_{\mu\nu\rho\sigma}}$.\footnote{Note how the nearest Greek equivalents to Latin counterparts are automatically used for rendering.}

\subsection{Various~$p$-form toy models}\label{sec:p-form_models}

\subsubsection{Definitions}

\paragraph*{Scalar field} First, we introduce a scalar (zero-form)~$\phi$, which we call \lstinline!ScalarField!:
\begin{lstlisting}
In[#]:= DefField[ScalarField[], PrintAs -> "\[Phi]", PrintSourceAs -> "\[Rho]"]; 
\end{lstlisting}
The scalar source, denoted by~$\rho$, is automatically defined by \PSALTer{}. The output is shown in~\cref{FieldKinematicsScalarField}; we see that the scalar field~$\phi$ contains only a~$J^P=0^+$ mode.

\begin{table*}[!t]
	\includegraphics[width=\linewidth]{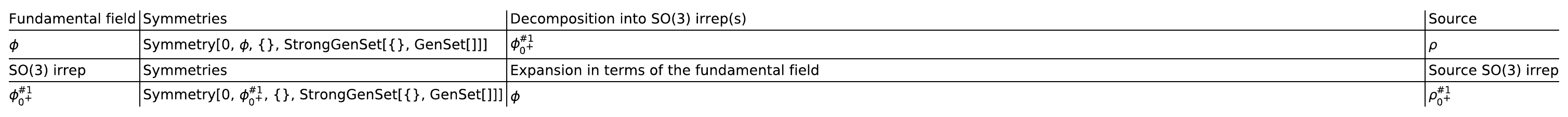}
	\caption{The declaration of \lstinline!ScalarField!, which contains only a~$J^P=0^+$ mode. These definitions are used in~\cref{ParticleSpectrographZeroByThreeCSKRTheory}.}
\label{FieldKinematicsScalarField}
\end{table*}

\paragraph*{Vector field} We next introduce a vector (one-form)~$\VectorField{_\mu}$, which we call \lstinline!VectorField!:
\begin{lstlisting}
In[#]:= DefField[VectorField[-m], PrintAs -> "\[ScriptCapitalA]", PrintSourceAs -> "\!\(\ * SubscriptBox[\(\[ScriptCapitalJ]\), \((\[ScriptCapitalA])\)]\)"]; 
\end{lstlisting}
The output is shown in~\cref{FieldKinematicsVectorField}; we see that~$J^P=0^+$ and~$J^P=1^-$ modes are present in the field, and that the source~$\JForm{A}{^\mu}$ and its corresponding irreps are also defined. 

\begin{table*}[!t]
	\includegraphics[width=\linewidth]{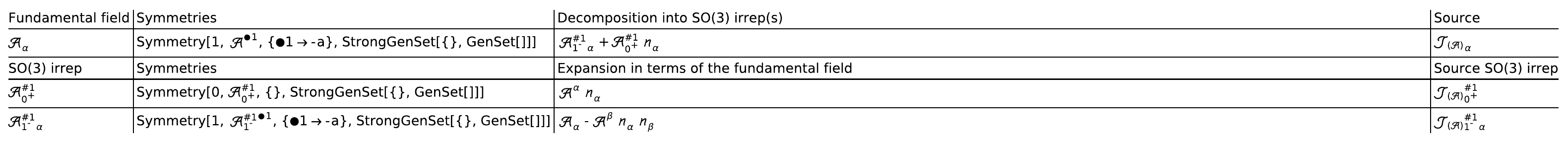}
	\caption{The declaration of \lstinline!VectorField!, with~$J^P=0^+$ and~$J^P=1^-$ modes. These definitions are used in~\cref{ParticleSpectrographOneByTwoCSKRTheory}.}
\label{FieldKinematicsVectorField}
\end{table*}

\paragraph*{Two-form field} A rank-two antisymmetric tensor (two-form)~$\TwoFormField{}{_{\mu\nu}}$ is defined as \lstinline!TwoFormField!: 
\begin{lstlisting}
In[#]:= DefField[TwoFormField[-m, -n], Antisymmetric[{-m, -n}], PrintAs -> "\[ScriptCapitalB]", PrintSourceAs -> "\!\(\ * SubscriptBox[\(\[ScriptCapitalJ]\), \((\[ScriptCapitalB])\)]\)"]; 
\end{lstlisting}
The output is shown in~\cref{FieldKinematicsTwoFormField}; we see that it carries~$J^P=1^+$ and~$J^P=1^-$ modes, and that the source~$\JForm{B}{^{\mu\nu}}$ and its irreps are also defined.

\begin{table*}[t!]
	\includegraphics[width=\linewidth]{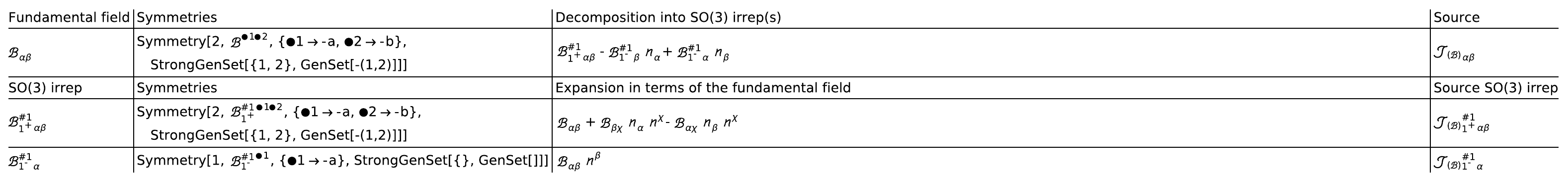}
	\caption{The declaration of \lstinline!TwoFormField!, which contains~$J^P=1^+$ and~$J^P=1^-$ modes. These definitions are used in~\cref{ParticleSpectrographParityViolatingMassiveTwoFormTheory,ParticleSpectrographParityIndefiniteMassiveTwoFormTheory,ParticleSpectrographOneByTwoCSKRTheory}.}
\label{FieldKinematicsTwoFormField}
\end{table*}

\paragraph*{Three-form field} Finally, a rank-three totally antisymmetric tensor (three-form)~$\ThreeFormField{_{\mu\nu\sigma}}$ is defined as \lstinline!ThreeFormField!:
\begin{lstlisting}
In[#]:= DefField[ThreeFormField[-m, -n, -r], Antisymmetric[{-m, -n, -r}], PrintAs -> "\[ScriptCapitalC]", PrintSourceAs -> "\!\(\ * SubscriptBox[\(\[ScriptCapitalJ]\), \((\[ScriptCapitalC])\)]\)"]; 
\end{lstlisting}
The output is shown in~\cref{FieldKinematicsThreeFormField}; we see that it carries~$J^P=1^+$ and~$J^P=0^-$ irreps, and that the source~$\JForm{C}{^{\mu\nu\sigma}}$ and its irreps are also defined.

\begin{table*}[t!]
	\includegraphics[width=\linewidth]{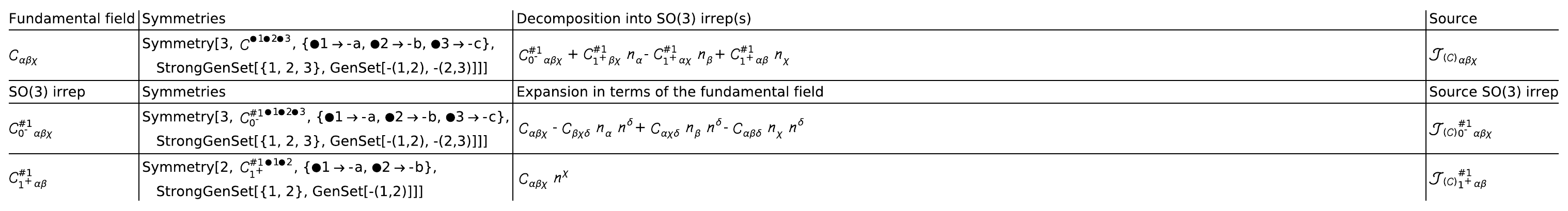}
	\caption{The declaration of \lstinline!ThreeFormField!, which contains~$J^P=0^-$ and~$J^P=1^+$ modes. These definitions are used in~\cref{ParticleSpectrographZeroByThreeCSKRTheory}.}
\label{FieldKinematicsThreeFormField}
\end{table*}

\paragraph*{Lagrangian coupling coefficients} In what follows, we shall consider three-parameter models, with the three couplings~$\alpha$,~$\beta$ and~$\gamma$ as named variables \lstinline!Coupling1!,
~\lstinline!Coupling2! and~\lstinline!Coupling3!:
\begin{lstlisting}
In[#]:= DefConstantSymbol[Coupling1, PrintAs -> "\[Alpha]"]; 
\end{lstlisting}
\begin{lstlisting}
In[#]:= DefConstantSymbol[Coupling2, PrintAs -> "\[Beta]"]; 
\end{lstlisting}
\begin{lstlisting}
In[#]:= DefConstantSymbol[Coupling3, PrintAs -> "\[Gamma]"]; 
\end{lstlisting}
Note that \PSALTer{} strictly requires all the operators in the Lagrangian density to be \emph{linearly} parametrized by these Lagrangian coupling coefficients.

\subsubsection{Spectroscopy}

\paragraph*{Parity-violating massive two-form} The simplest parity-violating model that we could think of involves a two-form~$\TwoFormField{^{\mu\nu}}$ with a purely parity-violating mass term 
\begin{equation}\label{eq:massive_2form_parity-violating}
S = \int\diff^4x \left[\alpha\PD{_{[\mu}}\TwoFormField{_{\nu\rho]}}\PD{^{[\mu}}\TwoFormField{^{\nu\rho]}}+\gamma\tensor{\epsilon}{^{\mu\n\rho\sigma}}\TwoFormField{_{\mu\nu}}\TwoFormField{_{\rho\sigma}}\right] \ .
\end{equation}
Note that in~\cref{eq:massive_2form_parity-violating} both terms are linearly parameterised by~$\alpha$ and~$\gamma$, and square brackets denote antisymmetrization of the enclosed indexes. Henceforth, we will always omit the source coupling, which in~\cref{eq:massive_2form_parity-violating} would require adding the term~$\int\diff^4x\ \TwoFormField{_{\mu\nu}}\JForm{B}{^{\mu\nu}}$. This is because the source~$\SourceUp{X}{\mu}$ in~\cref{BasicPositionAction} is introduced as a formal test field, useful only for the computations in~\cref{TheoreticalDevelopment}. The source-free model defined in~\cref{eq:massive_2form_parity-violating} has precisely the same particle spectrum. In fact, \PSALTer{} does not accept source couplings to be input by the user; they are automatically included when~\cref{eq:massive_2form_parity-violating} is analysed with the command:
\begin{lstlisting}
In[#]:= ParticleSpectrum[Coupling3 * epsilonG[-m, -n, -r, -s] * TwoFormField[m, n] * TwoFormField[r, s] - (2 * Coupling1 * CD[-n][TwoFormField[-m, -r]] * CD[r][TwoFormField[m, n]])/3 + (Coupling1 * CD[-r][TwoFormField[-m, -n]] * CD[r][TwoFormField[m, n]])/3, TheoryName -> "ParityViolatingMassiveTwoFormTheory", MaxLaurentDepth -> 3]; 
\end{lstlisting}
The output of the \PSALTer{} software is given in a few seconds; it is shown in~\cref{ParticleSpectrographParityViolatingMassiveTwoFormTheory}. The result here is trivial: despite the rich appearance of~\cref{eq:massive_2form_parity-violating}, the theory does not contain any dynamical degrees of freedom (d.o.f).

\begin{figure}[h!]
	\includegraphics[width=\linewidth]{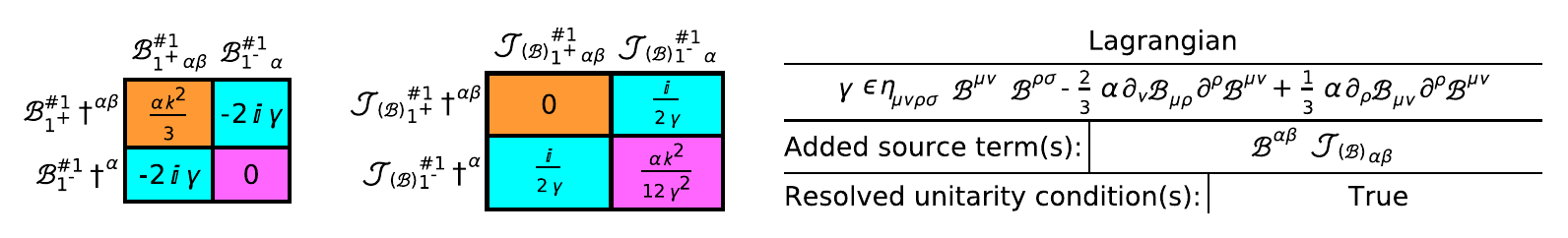}
	\caption{The spectrograph of the parity-violating two-form model in~\cref{eq:massive_2form_parity-violating}.}
\label{ParticleSpectrographParityViolatingMassiveTwoFormTheory}
\end{figure}

\paragraph*{Parity-indefinite massive two-form}~The next model under consideration is a direct generalization of~\cref{eq:massive_2form_parity-violating}, by allowing for both parity-preserving and parity-violating mass terms
\begin{equation}\label{eq:gen_two-form}
S=\int\diff^4x\left[\alpha\PD{_{[\mu}}\B{_{\nu\rho]}}\PD{^{[\mu}}\B{^{\nu\rho]}}+\beta\B{_{\mu\nu}}\B{^{\mu\nu}}+\gamma\tensor{\epsilon}{^{\mu\n\rho\sigma}}\B{_{\mu\nu}}\B{_{\rho\sigma}}\right] \ ,
\end{equation}
through the addition of the~$\beta$ coupling. The dynamics of~\cref{eq:gen_two-form} is probed using the following input:
\begin{lstlisting}
In[#]:= ParticleSpectrum[Coupling2 * TwoFormField[-m, -n] * TwoFormField[m, n] + Coupling3 * epsilonG[-m, -n, -r, -s] * TwoFormField[m, n] * TwoFormField[r, s] - (2 * Coupling1 * CD[-n][TwoFormField[-m, -r]] * CD[r][TwoFormField[m, n]])/3 + (Coupling1 * CD[-r][TwoFormField[-m, -n]] * CD[r][TwoFormField[m, n]])/3, TheoryName -> "ParityIndefiniteMassiveTwoFormTheory", MaxLaurentDepth -> 3]; 
\end{lstlisting}
The output of the \PSALTer{} software is shown in~\cref{ParticleSpectrographParityIndefiniteMassiveTwoFormTheory}. If we adhere to the formal notation from~\cref{MassiveSpectrum}, we say that the theory in~\cref{eq:gen_two-form} propagates a massive spin-one particle with square mass~$\PVSquareMass{1}{1}=-3(\beta^2+4\gamma^2)/\alpha\beta$. The particle is not a ghost if~$\alpha>0$, and it is neither a ghost nor a tachyon if additionally~$\beta<0$.

\begin{figure}[h!]
	\includegraphics[width=\linewidth]{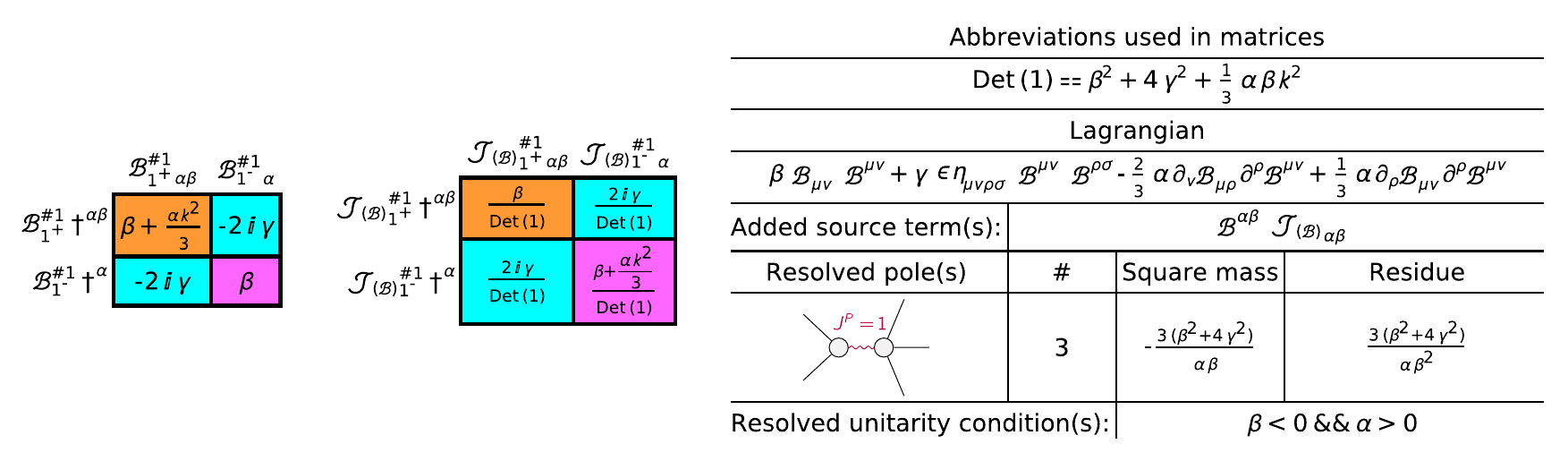}
	\caption{The spectrograph of the parity-violating two-form model in~\cref{eq:gen_two-form}.}
\label{ParticleSpectrographParityIndefiniteMassiveTwoFormTheory}
\end{figure}

\paragraph*{`One-by-two' Cremmer--Scherk--Kalb--Ramond (CSKR) theory} Yet another parity-indefinite action, that in common with the parity-violating two-form theory propagates a single massive spin-one particle, also involves a two-form~$\TwoFormField{_{\mu\nu}}$ coupled appropriately to a vector field~$\VectorField{_\mu}$~(see e.g.~\cite{Aurilia:1981xg})
\begin{align}\label{OneByTwoCSKRTheory}
	S = \int\mathrm{d}^4 x\left[\alpha\PD{_{[\mu}}\VectorField{_{\nu]}}\PD{^{[\mu}}\VectorField{^{\nu]}}+\beta\PD{_{[\mu}}\TwoFormField{_{\nu\rho]}}\PD{^{[\mu}}\TwoFormField{^{\nu\rho]}}+\gamma\tensor{\epsilon}{^{\mu\nu\rho\sigma}}\TwoFormField{_{\mu\nu}}\PD{_{[\rho}}\VectorField{_{\sigma]}}\right] \ ,
\end{align}
where the constants in~\cref{OneByTwoCSKRTheory} are not supposed to be consistent with those in the previous examples, not even up to the mass dimension:\footnote{Indeed, in~\cref{OneByTwoCSKRTheory} the constants are all dimensionless.} we are simply re-using symbols which have already been defined in the user session. We refer to this as the `\emph{one-by-two}' CSKR model because it mixes one- and two-forms.\footnote{More generally, in~$d=q+p+1$-dimensions, the~$p$-by-$q$ CSKR model mixes a~$p$-form with a~$q$-form; it is built from two gauge-invariant kinetic terms and a topological interaction. Evidently in~$d=4$ there is exactly one further CSKR theory, which we discuss presently.} The model in~\cref{OneByTwoCSKRTheory} is probed using the following input:
\begin{lstlisting}
In[#]:= ParticleSpectrum[-2 * Coupling3 * epsilonG[-m, -n, -r, -s] * TwoFormField[r, s] * CD[n][VectorField[m]] - 2 * Coupling1 * CD[-m][VectorField[-n]] * CD[n][VectorField[m]] + 2 * Coupling1 * CD[-n][VectorField[-m]] * CD[n][VectorField[m]] - (2 * Coupling2 * CD[-n][TwoFormField[-m, -r]] * CD[r][TwoFormField[m, n]])/3 + (Coupling2 * CD[-r][TwoFormField[-m, -n]] * CD[r][TwoFormField[m, n]])/3, TheoryName -> "OneByTwoCSKRTheory", MaxLaurentDepth -> 3]; 
\end{lstlisting}
The output is shown in~\cref{ParticleSpectrographOneByTwoCSKRTheory}. As anticipated, the particle content of~\cref{OneByTwoCSKRTheory} comprises a massive spin-one particle; its square mass is~$\PVSquareMass{1}{1}=-3\gamma^2/\alpha\beta$, and the no-ghost and no-tachyon conditions are~$\alpha<0$ and~$\beta>0$. Note that if a theory possesses gauge invariances, \PSALTer{} automatically identifies the associated source constraints; in the present case, the action is invariant under the independent transformations~$\TwoFormField{_{\mu\nu}}\mapsto\TwoFormField{_{\mu\nu}}+2\PD{_{[\mu}}\GaugePotential{_{\nu]}}$ and~$\VectorField{_{\mu}}\mapsto\VectorField{_{\mu}}+\PD{_\mu}\GaugePotential{}$, where~$\GaugePotential{_\mu}\equiv\GaugePotential{_\mu}(x)$ and~$\GaugePotential{}\equiv\GaugePotential{}(x)$ are local gauge generators. As a consequence of these gauge invariances, the sources~$\JForm{A}{^\mu}$ and~$\JForm{B}{^{\mu\nu}}$ are subject to the constraints~$\PD{_\mu}\JForm{A}{^\mu}=0$ and~$\PD{_{\mu}}\JForm{B}{^{\mu\nu}}=0$, which are automatically identified and taken account of in the \PSALTer{} code.

\begin{figure}[h!]
	\includegraphics[width=\linewidth]{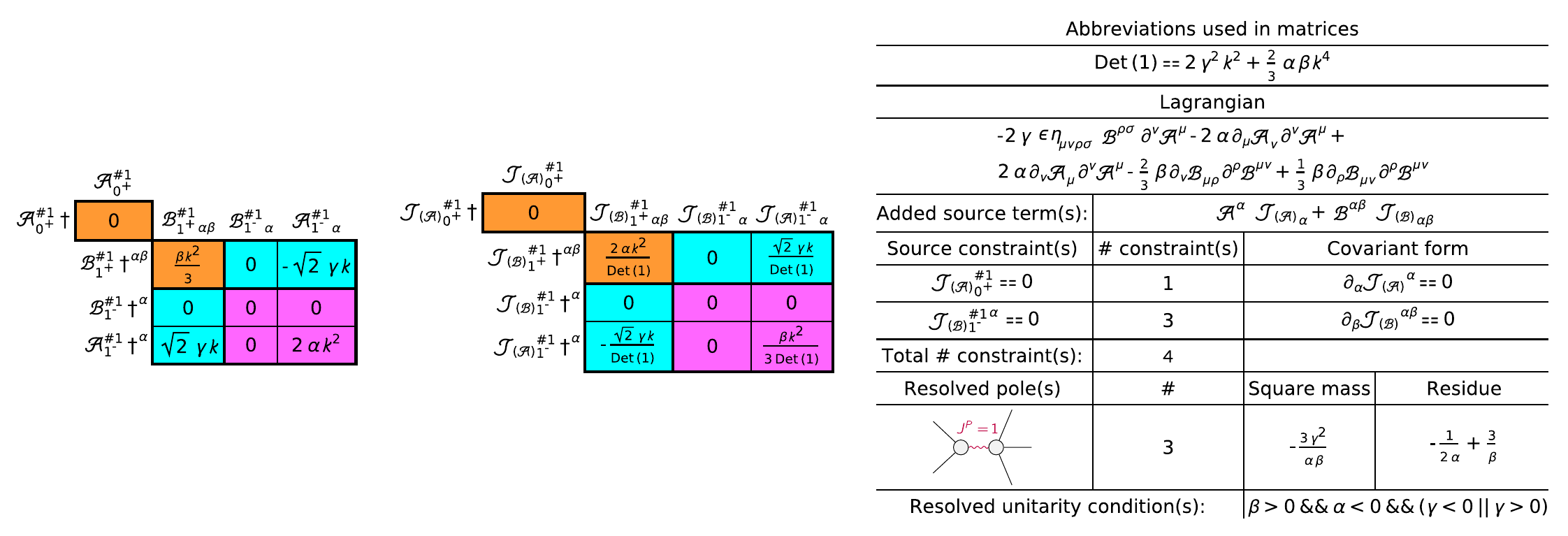}
	\caption{The spectrograph of the `one-by-two' CSKR model in~\cref{OneByTwoCSKRTheory}. All quantities are defined in~\cref{FieldKinematicsVectorField,FieldKinematicsTwoFormField}.}
\label{ParticleSpectrographOneByTwoCSKRTheory}
\end{figure}

\paragraph*{`Zero-by-three' CSKR theory} The other four-dimensional CSKR theory is the `zero-by-three' model, which mixes a scalar with a three-form. The action reads~(see e.g.~\cite{Aurilia:1981xg})
\begin{align}\label{ZeroByThreeCSKRTheory}
	S = \int\mathrm{d}^4 x\left[\alpha\PD{_\mu}\phi\PD{^\mu}\phi+\beta\PD{_{[\mu}}\ThreeFormField{_{\nu\rho\sigma]}}\PD{^{[\mu}}\ThreeFormField{^{\nu\rho\sigma]}}+\gamma\tensor{\epsilon}{^{\mu\nu\rho\sigma}}\ThreeFormField{_{\mu\nu\rho}} \PD{_{\sigma}}\phi\right],
\end{align}
and the dynamics of~\cref{ZeroByThreeCSKRTheory} is probed using the following input:
\begin{lstlisting}
In[#]:= ParticleSpectrum[Coupling3 * epsilonG[-m, -n, -r, -s] * ThreeFormField[n, r, s] * CD[m][ScalarField[]] + Coupling1 * CD[-m][ScalarField[]] * CD[m][ScalarField[]] - (3 * Coupling2 * CD[-r][ThreeFormField[-m, -n, -s]] * CD[s][ThreeFormField[m, n, r]])/4 + (Coupling2 * CD[-s][ThreeFormField[-m, -n, -r]] * CD[s][ThreeFormField[m, n, r]])/4, TheoryName -> "ZeroByThreeCSKRTheory", MaxLaurentDepth -> 3]; 
\end{lstlisting}
The output is shown in~\cref{ParticleSpectrographZeroByThreeCSKRTheory}, from which we see that for~$\alpha>0$ and~$\beta<0$,~the theory in~\cref{ZeroByThreeCSKRTheory} propagates a healthy massive spin-zero mode with mass~$\PVSquareMass{0}{1}= -6\gamma^2/\alpha\beta$. Since the action in~\cref{ZeroByThreeCSKRTheory} is invariant under~$\ThreeFormField{_{\mu\nu\rho}}\mapsto\ThreeFormField{_{\mu\nu\rho}}+\PD{_{[\mu}}\GaugePotential{_{\nu\rho]}}$ for local gauge generator~$\GaugePotential{_{\mu\nu}}\equiv\GaugePotential{_{[\mu\nu]}}\equiv\GaugePotential{_{\mu\nu}}(x)$, there is an associated source constraint which is given by~$\PD{_{\mu}}\JForm{C}{^{\mu\nu\rho}}=0$.

\begin{figure}[h!]
	\includegraphics[width=\linewidth]{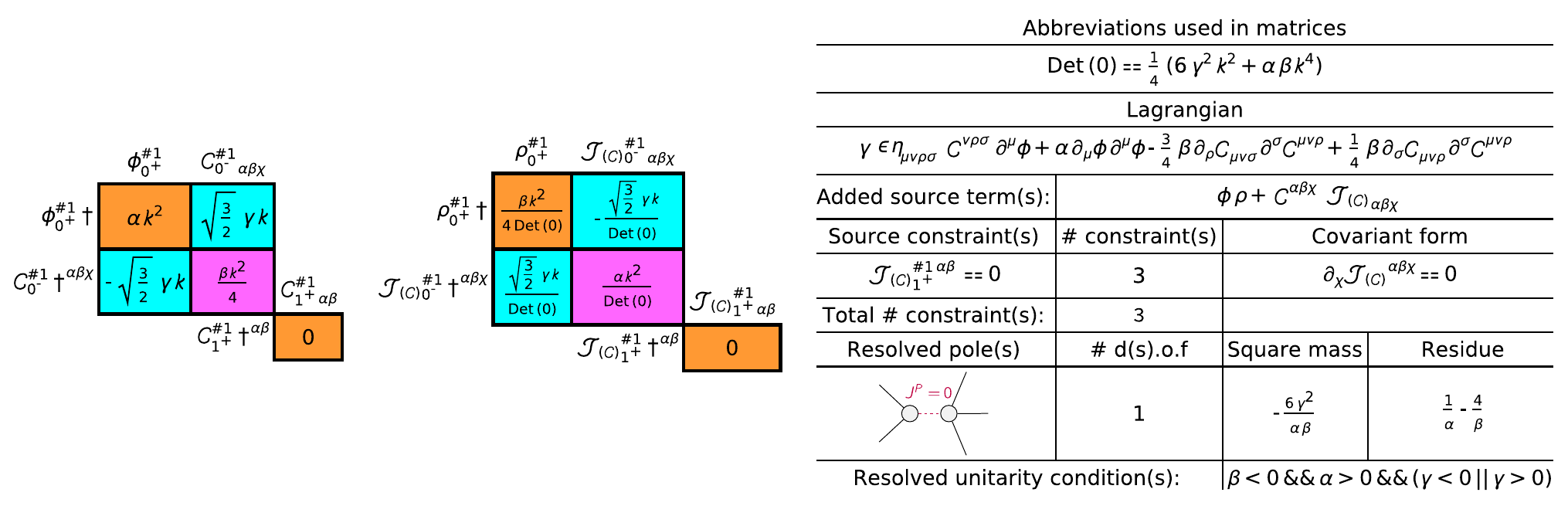}
	\caption{The spectrograph of the `zero-by-three' CSKR model in~\cref{ZeroByThreeCSKRTheory}. All quantities are defined in~\cref{FieldKinematicsThreeFormField,FieldKinematicsScalarField}.}
\label{ParticleSpectrographZeroByThreeCSKRTheory}
\end{figure}

\subsection{Einstein--Cartan gravity}\label{sec:EC_gravity}

\paragraph*{Localizing the Poincar\'{e} group} Gravitational interactions emerge organically by gauging the Lorentz group; i.e., promoting the global symmetry under Lorentz rotations --- which is an established feature of all laboratory physics --- to a local symmetry. By combining this with the local symmetry of translations --- which in general relativity is already manifest as general covariance --- one arrives at a gauge theory of the whole Poincar\'{e} group~\cite{Utiyama:1956sy,Kibble:1961ba,Sciama:1962}. This calls for the introduction of two gauge fields: the (co-)tetrad~$\tensor{e}{^\alpha_{\SmashAcute\mu}}$ and spin connection~$\tensor{\omega}{^{\alpha\beta}_{\SmashAcute\mu}}\equiv\tensor{\omega}{^{[\alpha\beta]}_{\SmashAcute\mu}}$, resulting in the Einstein--Cartan (EC) formulation of gravity.\footnote{Extra pedantry may avoid confusion resulting from our choice of words here. Strictly speaking, `\emph{EC gravity}' refers exclusively to the framing of the Einstein--Hilbert action in this formulation, and then only when it is combined with the geometric interpretation in which the spacetime genuinely has curvature and torsion~\cite{Cartan:1922,Cartan:1923,Cartan:1924,Cartan:1925,Einstein:1925, Einstein:1928,Einstein:19282}. Physics in general, and particle physics in particular, are utterly oblivious to our interpretation of the geometry. It comes as no surprise, therefore, that one can also gauge the Poincar\'{e} group entirely within Minkowski spacetime~\cite{Utiyama:1956sy, Sciama:1962, Kibble:1961ba} (see also~\cite{Barker:2020elg,Barker:2020gcp,Barker:2021oez,Rigouzzo:2022yan,VandepeerBarker:2022xnp,Barker:2022jsh,Barker:2023bmr,Barker:2023bmr}). Properly, only this Minkowski formulation is referred to as Poincar\'{e} gauge theory (PGT). As with EC gravity, PGT has also become associated for historical reasons with a \emph{theory}, not just with a formulation. PGT extends the Einstein--Hilbert action to include all other scalar invariants which are quadratic in the translational and rotational field strength tensors (geometrically, the torsion and the curvature in~\cref{TorsionCurvature}). We deal with this much larger model in~\cref{app:PGT_comparison}.} To be consistent with the notation used so far, Greek letters such as~$\a$ and~$\b$ here continue to stand for internal Lorentz indexes, as manipulated with the Minkowski metric, whereas accented Greek letters such as~$\SmashAcute\mu$ and~$\SmashAcute\nu$ stand for curved spacetime indexes. The curved-space metric is~$\tensor{g}{_{\SmashAcute\mu\SmashAcute\nu}}=\tensor{e}{^\alpha_{\SmashAcute\mu}}\tensor{e}{^\beta_{\SmashAcute\nu}}\tensor{\eta}{_{\alpha\beta}}$, and we assume the auxiliary identities~$\tensor{e}{^\alpha_{\SmashAcute\mu}}\tensor{e}{_\alpha^{\SmashAcute\nu}}=\tensor*{\delta}{_{\SmashAcute\mu}^{\SmashAcute\nu}}$ and~$\tensor{e}{^\alpha_{\SmashAcute\mu}}\tensor{e}{_\beta^{\SmashAcute\mu}}=\tensor*{\delta}{_\beta^\alpha}$ which function as extra kinematic restrictions. The associated field strength tensors, out of which the action is built, are torsion and curvature
\begin{equation}\label{TorsionCurvature}
	\begin{aligned}
		\ECT{^\alpha_{\SmashAcute\mu\SmashAcute\nu}} & \equiv \PD{_{\SmashAcute\mu}}\tensor{e}{^\alpha_{\SmashAcute\nu}} - \PD{_{\SmashAcute\nu}} \tensor{e}{^\alpha_{\SmashAcute\mu}} + \tensor{\omega}{^{\alpha}_{\beta\SmashAcute\mu}}\tensor{e}{^\beta_{\SmashAcute\nu}} - \tensor{\omega}{^{\alpha}_{\beta\SmashAcute\nu}}\tensor{e}{^\beta_{\SmashAcute\mu}}  \ ,
	\\
		\ECR{^{\alpha\beta}_{\SmashAcute\mu\SmashAcute\nu}} & \equiv \PD{_{\SmashAcute\mu}}\tensor{\omega}{^{\alpha\beta}_{\SmashAcute\nu}} - \PD{_{\SmashAcute\nu}}\tensor{\omega}{^{\alpha\beta}_{\SmashAcute\mu}} + \tensor{\omega}{^\alpha_{\gamma\SmashAcute\mu}}\tensor{\omega}{^{\gamma\beta}_{\SmashAcute\nu}} - \tensor{\omega}{^\alpha_{\gamma\SmashAcute\nu}}\tensor{\omega}{^{\gamma\beta}_{\SmashAcute\mu}} \ .
	\end{aligned}
\end{equation}

To extract the flat particle content of a theory --- provided that Minkowski spacetime is an admissible perturbative background in that no `accidental' gauge symmetries are present\footnote{Any gauge symmetry which is either broken non-linearly or absent on non-flat backgrounds is termed `accidental'~\cite{Percacci:2020ddy}. This feature necessarily signals a pathology~\cite{Karananas:2024hoh,Karananas:2024qrz}.} --- the tetrad is perturbed around the `Kronecker' choice of vacuum (see alternative vacua in~\cite{Barker:2023bmr,Blixt:2022rpl,Blixt:2023qbg}) so that 
\begin{equation}\label{DefTetradPerturbation}
	 \tensor{e}{^\alpha_{\SmashAcute{\mu}}} \equiv \tensor*{\delta}{^\alpha_{\SmashAcute{\mu}}}+\tensor{f}{^\alpha_{\SmashAcute\mu}} \ ,
	\quad
	\tensor{e}{_\alpha^{\SmashAcute{\mu}}} \equiv \tensor*{\delta}{_\alpha^{\SmashAcute{\mu}}}-\tensor{f}{_\alpha^{\SmashAcute{\mu}}}+\mathcal{O}\big(f^2\big) \ ,
\end{equation}
and~\cref{DefTetradPerturbation} defines a concrete perturbation scheme in~$\tensor{f}{^\alpha_{\SmashAcute\mu}}$ which makes it evident that --- at the quadratic order of the free theory in~\cref{BasicPositionAction} --- Greek and accented Greek indices can be freely exchanged. To complete our setup of the weak-field regime, we assume that~$\tensor{\omega}{^{\alpha\beta}_{\SmashAcute\mu}}$ is inherently perturbative.

\subsubsection{Definitions}

\paragraph*{Tetrad perturbation} We define the asymmetric rank-two \lstinline!TetradPerturbation!:
\begin{lstlisting}
In[#]:= DefField[TetradPerturbation[-a, -b], PrintAs -> "\[ScriptF]", PrintSourceAs -> "\[Tau]"]; 
\end{lstlisting}
	The output is shown in~\cref{FieldKinematicsTetradPerturbation}. Neglecting the (higher-order) distinction between indices, the field~$\tensor{f}{^\a_\b}$ is found to contain~$2^+$,~$1^-$,~$1^+$ and two~$0^+$ modes. The conjugate source~$\tensor{\tau}{_\a^\b}$ has a physical interpretation as the asymmetric stress-energy tensor, and it is automatically defined by \PSALTer{}.

\begin{table*}[t!]
	\includegraphics[width=\linewidth]{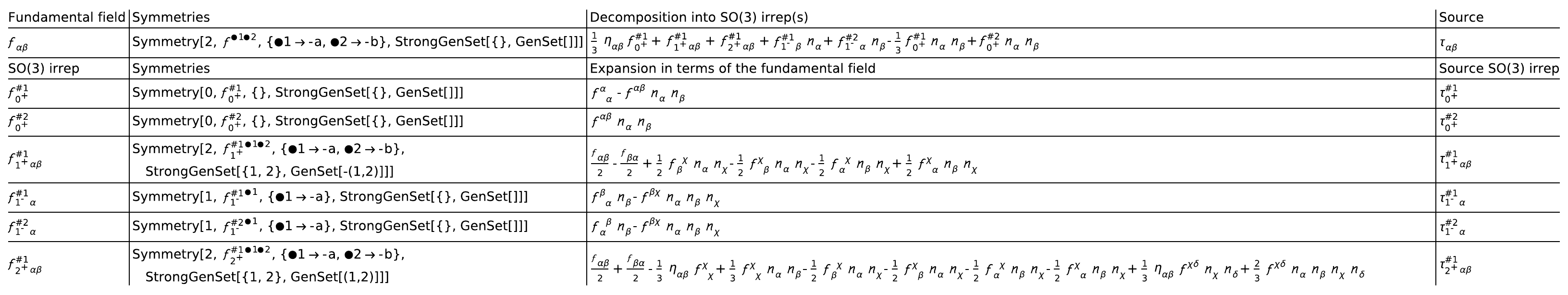}
	\caption{The declaration of \lstinline!TetradPerturbation!. These definitions are used in~\cref{ParticleSpectrographScalarParityViolatingPGT,ParticleSpectrographGeneralParityViolatingPGT}.}
	\label{FieldKinematicsTetradPerturbation}
\end{table*}

\paragraph*{Spin connection} We define the pair-antisymmetric \lstinline!SpinConnection! as:
\begin{lstlisting}
In[#]:= DefField[SpinConnection[-a, -b, -c], Antisymmetric[{-b, -c}], PrintAs -> "\[Omega]", PrintSourceAs -> "\[Sigma]"]; 
\end{lstlisting}
The output is shown in~\cref{FieldKinematicsSpinConnection}. Working again exclusively with the Lorentz indices at lowest order, the field~$\tensor{\omega}{^{\alpha\beta}_{\gamma}}$ is found to contain~$2^+$,~$2^-$, two~$1^-$, two~$1^+$,~$0^+$ and~$0^-$ modes. The conjugate source~$\tensor{\sigma}{_{\alpha\beta}^{\gamma}}$ has a physical interpretation as the spin current, and it is automatically defined by \PSALTer{}.

\begin{table*}[t!]
	\includegraphics[width=\linewidth]{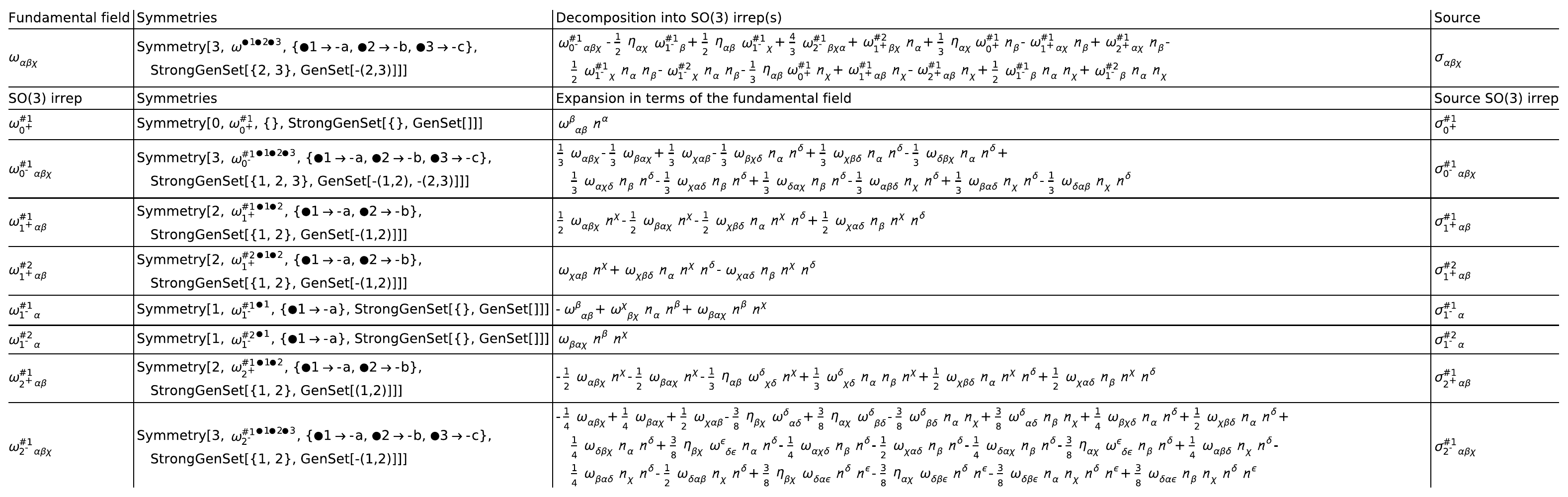}
	\caption{The declaration of \lstinline!SpinConnection!. These definitions are used in~\cref{ParticleSpectrographScalarParityViolatingPGT,ParticleSpectrographGeneralParityViolatingPGT}.}
	\label{FieldKinematicsSpinConnection}
\end{table*}   

\paragraph*{Lagrangian coupling coefficients} As with our analysis in~\cref{ParticleSpectrographZeroByThreeCSKRTheory}, we also define \lstinline!C1! so as to denote the coupling~$c_1$:
\begin{lstlisting}
In[#]:= DefConstantSymbol[C1, PrintAs -> "\!\(\ * SubscriptBox[\(\[ScriptC]\), \(1\)]\)"]; 
\end{lstlisting}
For the following considerations, there are ten Lagrangian coupling coefficients that need to be introduced, ranging up to \lstinline!C10! which denotes~$c_{10}$. Of these, the first five are dimensionless, and the final five are of mass dimension two.

\subsubsection{Spectroscopy}\label{ECSpectroscopy}

\paragraph*{Parity-indefinite Einstein--Cartan gravity} The tetrad and spin connection have in total 40 components (16 are in~$\tensor{e}{^\alpha_{\SmashAcute\mu}}$ and 24 in the pair-antisymmetric~$\tensor{\omega}{^{\alpha\beta}_{\SmashAcute\mu}}$), and if the action is taken to contain terms which are at most quadratic in the derivatives, then the theory can accommodate up to 20 propagating d.o.f.\footnote{These are the two massless graviton polarisations, and three pairs of massive particles of spin two (two times five d.o.f), spin one (two times three d.o.f) and spin zero (two times one d.o.f).}~It is well known~\cite{Sezgin:1979zf,Hayashi:1980qp,Karananas:2014pxa} that even linearly, not all of these can harmoniously coexist. Although there are choices of parameters for which particles carrying spin~$J\ge 1$ may be healthy at the linearized level, it is widely believed that when it comes to (at least parity-preserving) EC gravity,  the only non-linearly viable models propagate exclusively scalar modes~\cite{Hecht:1996np,Chen:1998ad,Yo:1999ex,Yo:2001sy}.\footnote{See also~\cite{Yo:2006qs,Shie:2008ms,Chen:2009at,Baekler:2010fr,Ho:2011qn,Ho:2011xf,Ho:2015ulu,Tseng:2018feo,Zhang:2019mhd,Zhang:2019xek,MaldonadoTorralba:2020mbh,delaCruzDombriz:2021nrg,Barker:2024dhb} for applications and~\cite{Puetzfeld:2004yg,Ni:2009fg,Puetzfeld:2014sja,Ni:2015poa,Barker:2020gcp} for reviews.} The most general parity-indefinite action that propagates the graviton and scalars only, comprises all invariants which are at most quadratic in torsion and the scalar~$\ECR{}$ and pseudoscalar~$\Holst{}$ curvatures, whose definitions are
\begin{equation}\label{ScalarConditions}
	\ECR{} \equiv \tensor{e}{_\alpha^{\SmashAcute\mu}}\tensor{e}{_\beta^{\SmashAcute\nu}}\ECR{^{\alpha\beta}_{\SmashAcute\mu\SmashAcute\nu}} \ ,\quad \Holst{} \equiv \epsilon^{\alpha\beta\kappa\lambda} \tensor{\eta}{_{\alpha\gamma}}\tensor{\eta}{_{\beta\delta}}\tensor{e}{_\kappa^{\SmashAcute\mu}}\tensor{e}{_\lambda^{\SmashAcute\nu}}\ECR{^{\gamma\delta}_{\SmashAcute\mu\SmashAcute\nu}} \ .
\end{equation}
To accompany~\cref{ScalarConditions} we define~$e\equiv \det(\tensor{e}{^\alpha_{\SmashAcute\mu}})$, and the Lorentz-indexed torsion~$\ECT{_{\mu\nu\rho}}\equiv\tensor{\eta}{_{\mu\sigma}}\tensor{e}{_\nu^{\SmashAcute\mu}}\tensor{e}{^\rho_{\SmashAcute\nu}}\ECT{^\sigma_{\SmashAcute\mu\SmashAcute\nu}}$, along with its trace~$\ECT{_\nu}\equiv\tensor{\eta}{^{\mu\rho}}\ECT{_{\mu\nu\rho}}$. Accounting for all possible  combinations, this action reads\,\footnote{\label{foot:holst}Note that up to a total derivative
\be
\ECR{} \propto \tensor{\epsilon}{^{\m\n\rho\s}} \ECT{_{\m\n\lambda}}\ECT{_{\rho\s}^{\lambda}} \ ,
\ee
so it is not really necessary to include the~$c_2$-term in the action. 
}
\begin{align}\label{eq:EC_parity_indefinite_general}
	S = \int \diff^4x \,e \Big[& c_1\ECR{} +c_2 \Holst{} +c_3 \ECR{}^2 +c_4 \ECR{}\Holst{} + c_5 \Holst{}^2+ c_6 \ECT{_{\mu\nu\rho}}\ECT{^{\mu\nu\rho}} \nonumber \\
	& + c_7 \ECT{_{\mu\nu\rho}}\ECT{^{\nu\rho\mu}} + c_8 \ECT{_\mu}\ECT{^\mu} +c_9\tensor{\epsilon}{^{\mu\nu\rho\sigma}}\ECT{_{\lambda\mu\nu}}\ECT{^{\lambda}_{\rho\sigma}} +c_{10}\tensor{\epsilon}{^{\m\n\rho\s}}\ECT{_{\mu\nu\lambda}}\ECT{_{\rho\sigma}^{\lambda}}\Big] \ .
\end{align}
The theory in~\cref{eq:EC_parity_indefinite_general} is an extension of that proposed in~\cite{Karananas:2024xja} (see also~\cite{Karananas:2025xcv,BeltranJimenez:2019hrm}), which uses terms solely quadratic in~$\ECR{}$ and~$\Holst{}$. It was shown in~\cite{Karananas:2024hoh,Karananas:2024qrz} that, in isolation, the square of the scalar curvature  propagates the Einstein graviton on a de Sitter background, but that this species becomes strongly coupled on Minkowski spacetime. The inclusion of the Einstein--Hilbert term guarantees, however, that no accidental symmetries arise and that perturbation theory makes sense also on flat backgrounds. Then,~$\ECR{}$~and~$\ECR{}^2$ propagate a positive-parity scalar mode \emph{in addition to} the massless graviton.\footnote{Even when the geometry is taken to be torsion-free, this mode persists and is instead associated with the scalar mode of the metric tensor. This is the well-known Starobinsky scalaron~\cite{Starobinsky:1980te}.} By contrast,~$\Holst{}^2$ propagates a negative-parity scalar mode~\cite{Hecht:1996np}, emerging from the axial vector part of torsion. Given its non-linear consistency, the evident lack of accidental gauge symmetries, and the interesting roles the scalars of gravitational origin can play in particle physics and cosmology~\cite{Karananas:2024xja,Gialamas:2024iyu,Karananas:2025xcv}, the model defined by~\cref{eq:EC_parity_indefinite_general} stands out and certainly deserves further scrutiny. To study the general case of~\cref{eq:EC_parity_indefinite_general} we input:
\begin{lstlisting}
In[#]:= ParticleSpectrum[(2 * C6 + C7) * SpinConnection[-a, -b, -c] * SpinConnection[a, b, c] + C10 * epsilonG[-b, -c, -d, -i] * SpinConnection[-a, d, i] * SpinConnection[a, b, c] + (C1 - 2 * C6 - 3 * C7) * SpinConnection[a, b, c] * SpinConnection[-b, -a, -c] - 2 * C10 * epsilonG[-a, -c, -d, -i] * (*omitted 3222 characters for brevity*) -d, -i] * CD[c][TetradPerturbation[a, b]] * CD[i][TetradPerturbation[d, -b]] + 4 * C4 * epsilonG[-c, -d, -i, -j] * CD[-b][SpinConnection[a, -a, b]] * CD[j][SpinConnection[c, d, i]], TheoryName -> "ScalarParityViolatingPGT", MaxLaurentDepth -> 1, AspectRatio -> Portrait, ShowPropagator -> False];
\end{lstlisting}
The resulting particle spectrum is shown in~\cref{ParticleSpectrographScalarParityViolatingPGT}. The graviton is, as expected, present since it is the natural companion of the Einstein--Hilbert term. It is healthy as long as
\begin{equation}\label{eq:massless_noghost}
c_1 < 0 \ .
\end{equation}
We also notice that the particle content contains two scalars which are dynamical. The current functionality in \PSALTer{} does not yet allow these states to be resolved, due to the fact that their square masses~$\PVSquareMass{0}{1}$ and~$\PVSquareMass{0}{2}$ are not, in general, rational functions of the Lagrangian coupling coefficients in~\cref{eq:EC_parity_indefinite_general}. Let us therefore discuss in detail the no-ghost and no-tachyon constraints for these scalars, by applying manually the improved techniques of~\cref{MassiveSpectrum}. We start from the~$4\times4$ coefficient matrix of the scalar~$J=0$ sector. By comparing~\cref{eq:canonical_block_structure,ParticleSpectrographScalarParityViolatingPGT} we find
\begin{equation}\label{eq:tilde_OJ_actual}
\begin{gathered}
	\WaveOperatorJVV{0} =\Chequer{ -24 c_5 k^2 - \Upsilon_4 }{ 2i(3c_4 k^2 + \Upsilon_2) }{ 2i(3c_4 k^2 + \Upsilon_2) }{ \f 1 2 (12 c_3 k^2 -\Upsilon1) } \ ,
\quad
	\WaveOperatorJVL{0} =\NormalMatrix{-2\sqrt{2}\Upsilon_2 k }{ 0 }{ -i\f{\Upsilon_1}{\sqrt 2} k }{ 0 } \ ,
\\
	\WaveOperatorJLV{0} =\NormalMatrix{-2\sqrt{2}\Upsilon_2 k }{ i\f{\Upsilon_1}{\sqrt 2} k  }{ 0}{ 0 } \ ,
\quad
	\WaveOperatorJLL{0} =\ChequerMain{ \Upsilon_3 k^2 }{ 0 }{ 0 }{ 0 } \ ,
\end{gathered}
\end{equation}
where the definitions of the coupling abbreviations in~\cref{eq:tilde_OJ_actual} can also be found in~\cref{ParticleSpectrographScalarParityViolatingPGT}. Then, from~\cref{eq:tilde_OJ_actual} we see that the kinetic- and mass-matrices in~\cref{DiagonalizedWaveOperator} are
\begin{equation}\label{eq:kin_mass_matrices}
	\KineticJ{0} = 6\Chequer{-4c_5 }{ i c_4}{i c_4 }{ c_3 },
	\quad
	\MassJ{0} = \frac{1}{\Upsilon_3}\Chequer{ 8\Upsilon_2^2-\Upsilon_3 \Upsilon_4 }{ 2i \Upsilon_2(\Upsilon_1+\Upsilon_3) }{ 2i \Upsilon_2 (\Upsilon_1+\Upsilon_3) }{ -\f{\Upsilon_1}{2}(\Upsilon_1+\Upsilon_3) }.
\end{equation}
According to~\cref{eq:kin_matrix_negative}, the theory is ghost-free provided that~$\KineticJ{0}$ is negative-definite, i.e.
\begin{equation}\label{eq:no_ghost_scalar_EC}
	\Big[ c_5>0\Big]\wedge \Big[ 4c_3 c_5 - c_4^2 > 0 \Big]\ ,
\end{equation}
and from~\cref{eq:no_ghost_scalar_EC} it follows that~$c_3>0$. Since~$\MassJ{0}$ is a~$2\times 2$ matrix, tachyon-freedom further requires~\cite{Blagojevic:2018dpz}
\begin{equation}
\label{eq:mass_constraints}
	\left[ \tr\left(\PostMassJ{0}\right)^2 - 4 \det\left(\PostMassJ{0}\right) > 0 \right] \wedge \left[ \tr\left(\PostMassJ{0}\right) > 0 \right]\wedge\left[ \det\left(\PostMassJ{0}\right) > 0 \right], \quad \PostMassJ{0}\equiv-\KineticInverseJ{0}\MassJ{0}.
\end{equation}
It can be easily checked that the second and third inequalities in~\cref{eq:mass_constraints} translate into
\begin{equation}\label{eq:mvpre}
	\begin{aligned}
		&\left[\frac{c_3(\Upsilon_3 \Upsilon_4 - 8\Upsilon_2^2)-2(\Upsilon_1+\Upsilon_3)(2c_4 \Upsilon_2+c_5 \Upsilon_1)}{\Upsilon_3(c_4^2-4c_3c_5)}> 0 \right]
		\\&\hspace{240pt}
		\wedge \left[\frac{(\Upsilon_1+\Upsilon_3)(\Upsilon_1 \Upsilon_4+8\Upsilon_2^2)}{\Upsilon_3(c_4^2-4c_3c_5)}> 0\right] \ ,
	\end{aligned}
\end{equation}
respectively. The latter inequality in~\cref{eq:mvpre}, once the massless~\cref{eq:massless_noghost} and massive~\cref{eq:no_ghost_scalar_EC} no-ghost conditions are taken into account, together with~$\Upsilon_1+\Upsilon_3=2c_1$, gives
\begin{equation}
\label{eq:no_tachyon_cond_1}
\Upsilon_3(\Upsilon_1 \Upsilon_4+8\Upsilon_2^2) > 0 \ ,
\end{equation}
and from the former --- following more-or-less verbatim the procedure spelled out in~\cite{Blagojevic:2018dpz} for simplifying such expressions --- one finds that
\begin{equation}
\label{eq:no_tachyon_cond_2}
\Upsilon_1 \Upsilon_3 <0 \ .
\end{equation}
In terms of the Lagrangian coupling coefficients, the conditions in~\cref{eq:no_tachyon_cond_1,eq:no_tachyon_cond_2} read
\begin{equation}\label{eq:no_tachyon_cond_3}
\begin{gathered}
	\bigg[(2c_6-c_7+3c_8)\left[(2c_1-2c_6+c_7-3c_8)(c_1-4(c_6+c_7))+8(c_2+c_{10}-2c_9)^2\right] > 0\bigg]
\\
	\wedge\bigg[(2c_6-c_7+3c_8)(2c_1-2c_6+c_7-3c_8) < 0\bigg] \ ,
\end{gathered}
\end{equation}
which are the well-known constraints~\cite{Blagojevic:2018dpz}, albeit written in a different notation. 
Accordingly, the theory in~\cref{eq:EC_parity_indefinite_general} is unitary if~\cref{eq:massless_noghost,eq:no_ghost_scalar_EC,eq:no_tachyon_cond_3} are satisfied.

\begin{figure}[htbp]
	\includegraphics[width=.9\linewidth]{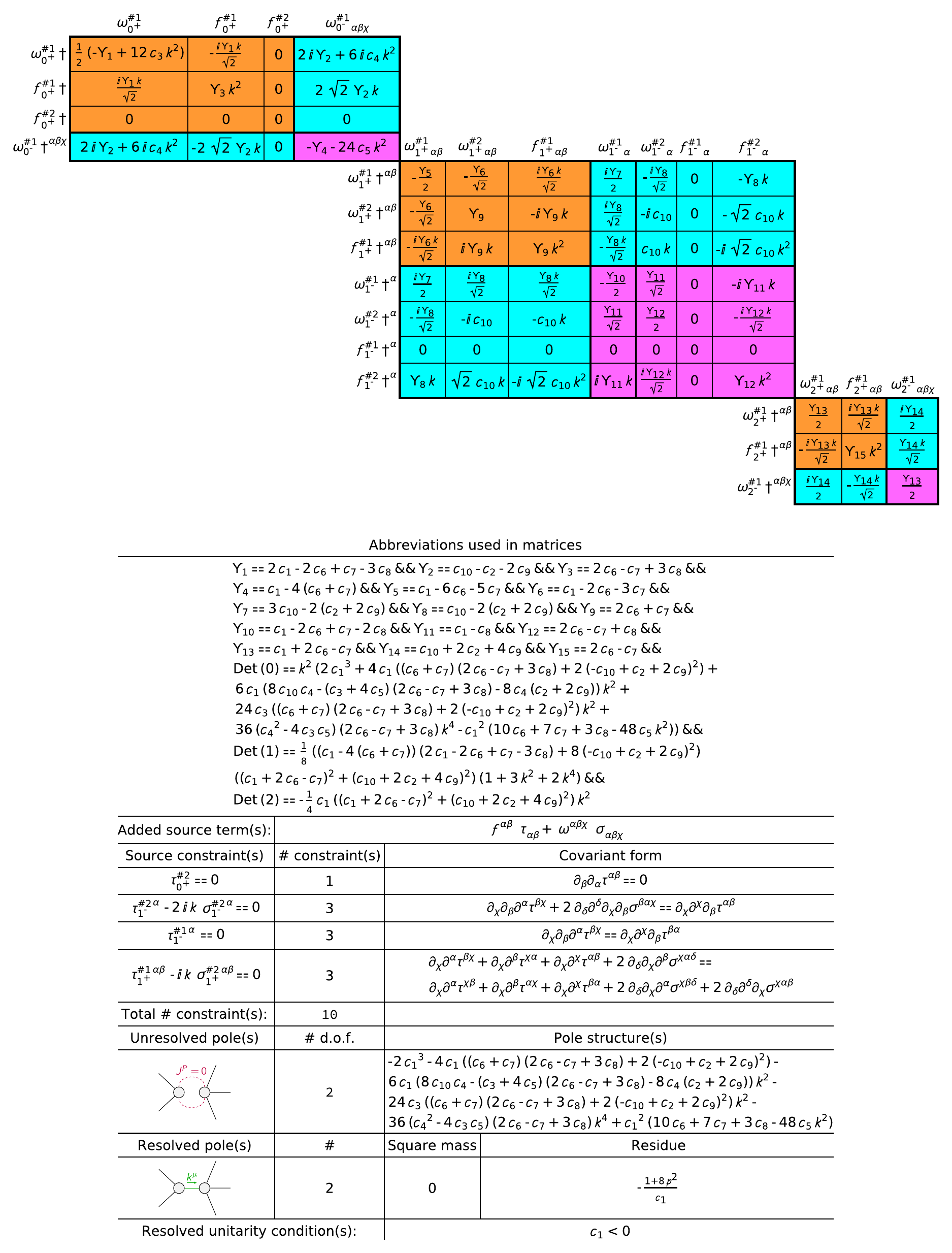}
	\caption{The spectrograph of the theory defined by~\cref{eq:EC_parity_indefinite_general}. In the current version of \PSALTer{}, the analysis is halted once it has been determined that the square masses are not rational functions of the Lagrangian coupling coefficients. Thus, the only `resolved' unitarity condition corresponds to the graviton in~\cref{eq:massless_noghost}. The existence of two `unresolved' poles is indicated. Our improved algorithm, not yet implemented in \PSALTer{}, recovers the associated conditions~\cref{eq:no_ghost_scalar_EC,eq:no_tachyon_cond_3}. 
    Note also the appearance of ten gauge generators --- owing to the Poincar\'e gauge redundancy --- ensuring the absence of accidental symmetries. 
    All quantities are defined in~\cref{FieldKinematicsTetradPerturbation,FieldKinematicsSpinConnection}.}
\label{ParticleSpectrographScalarParityViolatingPGT}
\end{figure}

\section{Conclusions}\label{Conclusions}

\paragraph*{Results of this paper} The spin-projection operator algorithm for parity-violating tensorial field theories, with its accompanying conventions, was implemented as part of the the Particle Spectrum for Any Tensor Lagrangian~(\PSALTer{}) software initiative.~\PSALTer{} is an open-source package contribution to the \xAct{} project, designed for use with \Mathematica{} (see~\cref{Install} for instructions on how to obtain and install the software). As an illustration, the new functionality in~\PSALTer{} was calibrated against a number of examples, confirming the analytic results. As a byproduct, we suggested a simplified way to obtain no-ghost conditions. This does not involve the computation of any propagator residues over massive poles, which greatly facilitates the analysis. 

\paragraph*{Further work} The techniques developed in this paper open the door to a systematic and comprehensive investigation into the particle content of parity-violating gravitational theories and in particular metric-affine gravity (see~\cite{Percacci:2020ddy} for the most comprehensive study of parity-preserving MAG), which is completely unexplored. 

\begin{acknowledgments}
This work was improved by many useful discussions with Justin Feng, Mike Hobson, Anthony Lasenby, Yun-Cherng Lin, Carlo Marzo and Claire Rigouzzo. We are grateful to Oleg Melichev, Roberto Percacci and especially Sebastian Zell for comments on the manuscript. WB is grateful for the support of Girton College, Cambridge, Marie Skłodowska-Curie Actions and the Institute of Physics of the Czech Academy of Sciences.

    This work used the DiRAC Data Intensive service~(CSD3 \href{www.csd3.cam.ac.uk}{www.csd3.cam.ac.uk}) at the University of Cambridge, managed by the University of Cambridge University Information Services on behalf of the STFC DiRAC HPC Facility~(\href{www.dirac.ac.uk}{www.dirac.ac.uk}). The DiRAC component of CSD3 at Cambridge was funded by BEIS, UKRI and STFC capital funding and STFC operations grants. DiRAC is part of the UKRI Digital Research Infrastructure. This work also used the Newton compute server, access to which was provisioned by Will Handley using an ERC grant.

\leavevmode

\paragraph*{Disclaimer} Co-funded by the European Union (Physics for Future – Grant Agreement No. 101081515). Views and opinions expressed are however those of the author(s) only and do not necessarily reflect those of the European Union or European Research Executive Agency. Neither the European Union nor the granting authority can be held responsible for them.
\end{acknowledgments}

\appendix

\section{Construction of operators}\label{ConstructionOfSpinProjectionOperators}

\paragraph*{The spatial hypersurface} In this appendix we bootstrap the construction of SPOs, including the new parity-violating SPOs, and set out the precise conventions that are implemented in \PSALTer{}. Projection by spin~$J$ and parity~$P$ is facilitated by choosing a preferred frame. We can derive this frame from the particle four-momentum~$\tensor{k}{^\mu}$, which can be chosen to be either timelike or null for massive and massless particles, respectively. In the timelike case, there is a unit vector~$\N{^\mu}\equiv\tensor{k}{^\mu}/k$, where~$k^2\equiv\tensor{k}{^\mu}\tensor{k}{_\mu}$, so that~$\N{^\mu}$ coincides with the preferred observer's four-velocity
\begin{equation}\label{UnitTimelikeComponents}
	\left[\N{^\mu}\right]^\mathrm{T}=\begin{bmatrix}1 & 0 & 0 & 0\end{bmatrix}.
\end{equation}
The usual transverse and longitudinal  projectors, parallel and perpendicular to the spatial hypersurface,  are
\begin{equation}
\label{eq:app_vector_projectors}
    \ProjPara{_\mu^\nu}\equiv\tensor*{\delta}{_\mu^\nu}-\N{_\mu}\N{^\mu} \ ,\quad \ProjPerp{_\mu^\nu}\equiv\N{_\mu}\N{^\mu} \ , 
\end{equation}
respectively.  We also introduce the overbar notation for indices, such as~$\tensor{V}{^{\overline{\mu}}}=\ProjPara{_\nu^\mu}\tensor{V}{^\nu}$, for a generic vector~$\tensor{V}{^\nu}$.

\paragraph*{Parity-preserving SPOs} The parity-preserving SPOs include both diagonal and off-diagonal SPOs, and are constructed exclusively from~$\ProjPara{_\mu^\nu}$ and~$\ProjPerp{_\mu^\nu}$ in~\cref{eq:app_vector_projectors}. They satisfy the following identities
\begin{subequations}
	\begin{gather}
		\SPODownUp{J}{P}{X}{Y}{i}{j}{\mu}{\nu}\equiv\SPOUpDown{J}{P}{Y}{X}{j}{i}{\nu}{\mu},\label{eq:app_Hermicity}
		\\
		\sum_{J,P}\sum_{\States{J}{P}{X}{i}}\SPODownUp{J}{P}{X}{X}{i}{i}{\mu}{\nu}\equiv\tensor*{\Delta}{^{\nu_X}_{\mu_X}},\label{eq:app_Completeness}
		\\	
		\SPODownUp{J}{P}{X}{Y}{i}{j}{\mu}{\nu}\SPODownUp{\Jp{}}{\Pp{}}{Y}{Z}{k}{l}{\nu}{\sigma}
		\equiv
		\tensor*{\delta}{_{jk}}
		\tensor*{\delta}{_{J\Jp{}}}
		\tensor*{\delta}{_{P\Pp{}}}
		\SPODownUp{J}{P}{X}{Z}{i}{l}{\mu}{\sigma},\label{eq:app_Orthonormality}
		\\
		P\FieldDown{X}{\mu}^*
		\SPOUpDown{J}{P}{X}{X}{i}{i}{\mu}{\nu}
		\FieldUp{X}{\nu}\geq 0,\label{eq:app_Positivity}
	\end{gather}
\end{subequations}
where~\cref{eq:app_Hermicity,eq:app_Completeness,eq:app_Orthonormality,eq:app_Positivity} encode symmetry,\footnote{Note that in~\cite{Barker:2024juc} the property of Hermicity was used instead of symmetry: this is because we now assume the parity-preserving SPOs to be real.}~completeness, orthonormality,\footnote{Note that in~\cite{Barker:2024juc} the more restrictive relation~$\SPODownUp{J}{P}{X}{X}{i}{j}{\mu}{\nu}\SPODownUp{\Jp{}}{\Pp{}}{X}{X}{k}{l}{\nu}{\sigma}=\tensor*{\delta}{_{jk}}\tensor*{\delta}{_{J\Jp{}}}\tensor*{\delta}{_{P\Pp{}}}\SPODownUp{J}{P}{X}{X}{i}{l}{\mu}{\sigma}$ was used instead of~\cref{eq:app_Orthonormality}, though the latter was in fact also true.} and positivity, respectively. 

\paragraph*{Reduced-index SPOs} In order to actually \emph{construct} these parity-preserving SPOs, we may proceed as follows. By the usual methods of Young tableaux and trace-free decomposition, the \emph{reduced-index}~$J^P$ states may be extracted manually so as to define the reduced-index SPOs
\begin{equation}\label{eq:app_ReducedIndex}
	\FieldDownState{J}{P}{X}{i}{\mu}
	\equiv
	\ReducedSPODownUp{J}{P}{X}{i}{\mu}{\nu}
	\FieldDown{X}{\nu},
\end{equation}
where~$\ParallelFieldIndices{J}{P}{\mu}$ is a reduced collection of parallel indices specific to the~$J^P$ state. The term \emph{reduced} here means that there may be fewer indices than in~$\FieldIndices{X}{\mu}$ for some or all of the fields~$X$ which contain states with this~$J^P$. For example, a high-rank tensor can contain many scalar states, however scalars do not require indices. The reduced-index states vanish upon contraction of any parallel indices~(i.e. they are trace-free), and these indices will also carry further symmetry properties. As a consequence of these constraints, each reduced-index state has only~$2J+1$ independent components, corresponding to the spin multiplicity. The reduced-index SPOs in~\cref{eq:app_ReducedIndex} are real projection operators constructed --- once again --- exclusively from~$\ProjPara{_\mu^\nu}$ and~$\ProjPerp{_\mu^\nu}$. They need not be normalised in any sense, and their definitions may vary according to conventions. Here, for example, we contrast with~\cite{Barker:2024juc}, in that we will no longer allow~$\tensor{\epsilon}{_{\mu\nu\sigma\lambda}}$ to be used in the definition of reduced-index SPOs: this is because we wish to closely track the parity of the projected states. Consequently, the reduced-index states are always tensors~(in the sense that they are never pseudotensors), and if there are~$N$ reduced indices then the parity is always~$P\equiv~(-1)^{N}$. The parity-preserving SPOs are typically more cumbersome than their reduced-index counterparts, but their properties in~\crefrange{eq:app_Hermicity}{eq:app_Positivity} make them formally useful. Within a given~$J^P$ sector they are given by
\begin{equation}\label{eq:app_Normalisation}
	\begin{aligned}
		&\SPODownUp{J}{P}{X}{Y}{i}{j}{\mu}{\nu}
	\equiv
	\Normalisation{J}{P}{X}{i}
	\Normalisation{J}{P}{Y}{j}
		\ReducedSPOUpDown{J}{P}{X}{i}{\sigma}{\mu}
		\ReducedSPODownUp{J}{P}{Y}{j}{\sigma}{\nu},
	\end{aligned}
\end{equation}
where, for any given choice of unnormalised reduced-index SPOs, the non-zero~$\Normalisation{J}{P}{X}{i}\in\mathbb{R}$ are fixed~(each up to a sign) by the requirements of~\cref{eq:app_Completeness}.\footnote{Note that in~\cite{Barker:2024juc} the use of~$\tensor{\epsilon}{_{\mu\nu\sigma\lambda}}$ in some of the reduced-index SPOs leads instead to~$\Normalisation{J}{P}{X}{i}\in\mathbb{C}$.} The construction in~\cref{eq:app_Normalisation} evidently ensures~\cref{eq:app_Hermicity}. By moving to the frame~$\N{^\mu}$ it is clear for finite fields that~$P\equiv\text{sgn}\left(\FieldDownState{J}{P}{X}{i}{\mu}^*\FieldUpState{J}{P}{X}{i}{\mu}\right)$ due to our choice of signature, and so~\cref{eq:app_Normalisation} also implies~\cref{eq:app_Positivity}. The property in~\cref{eq:app_Orthonormality} already follows from the fact that the extraction of~$J^P$ states is an irreducible decomposition. In fact, we may conclude that
\begin{equation}\label{eq:app_Normalisation2}
	\ReducedSPODownUp{J}{P}{X}{i}{\mu}{\sigma}
	\ReducedSPOUpDown{\Jp{}}{\Pp{}}{X}{j}{\nu}{\sigma}
	\equiv
	\frac{
	\tensor*{\delta}{_{ij}}
	\tensor*{\delta}{_{J\Jp{}}}
	\tensor*{\delta}{_{P\Pp{}}}
}{
	\Normalisation{J}{P}{X}{i}
	\Normalisation{\Jp{}}{\Pp{}}{X}{j}	
}
	\tensor*{\Delta}{^{\ParallelFieldIndices{\Jp{}}{\Pp{}}{\nu}}_{\ParallelFieldIndices{J}{P}{\mu}}}
	,
\end{equation}
where we use the same notation as in~\cref{Completeness}, so that~\cref{eq:app_Normalisation2,eq:app_Normalisation} together imply~\cref{eq:app_Orthonormality}.

\paragraph*{Parity-violating SPOs} We will now extend the above considerations to the case of parity-violating SPOs, comprising only off-diagonal projectors. We define 
\begin{equation}\label{eq:app_Eps}
	\Eps{_{\mu\nu\sigma}}\equiv\Eps{_{\overline{\mu\nu\sigma}}}\equiv\tensor{\epsilon}{_{\mu\nu\sigma\rho}}\N{^{\rho}},
\end{equation}
so that from~\cref{eq:app_Eps} a natural definition for~$\Pp{}\neq P$ is
\begin{equation}\label{eq:app_ParityViolating}
	\begin{aligned}
		&\PVSPODownUp{J}{P}{\Pp{}}{X}{Y}{i}{j}{\mu}{\nu}
	\equiv
	\NormalisationPV{J}\,
	\Normalisation{J}{P}{X}{i}
	\Normalisation{J}{\Pp{}}{Y}{j}
	\ReducedSPOUpDown{J}{P}{X}{i}{\sigma}{\mu}
		\Eps{_{\ParallelFieldIndices{J}{P}{\sigma}}^{\ParallelFieldIndices{J}{\Pp{}}{\rho}}}
	\ReducedSPODownUp{J}{\Pp{}}{Y}{j}{\rho}{\nu}\quad \forall J<2,
	\end{aligned}
\end{equation}
where~$\NormalisationPV{J}$ is a new normalisation constant, which will be explained in a moment. {} In practice, the convention in~\cref{eq:app_ParityViolating} works well for~$J<2$ because the~$0^+$ and~$0^-$ states have zero and three indices, respectively, while the~$1^+$ and~$1^-$ states have two and one indices respectively: these add up to the three indices of~\cref{eq:app_Eps}. The even-odd partitioning leads to a cancellation of signs, so that~\cref{eq:app_ParityViolating} extends~\cref{eq:app_Hermicity} to the more general property
\begin{equation}\label{eq:app_NewHermicity}
	\PVSPODownUp{J}{P}{\Pp{}}{X}{Y}{i}{j}{\mu}{\nu}\equiv \PVSPOUpDown{J}{\Pp{}}{P}{Y}{X}{j}{i}{\nu}{\mu}.
\end{equation}
Note that~\cref{eq:app_NewHermicity} would only imply Hermicity for~$\NormalisationPV{J}\in\mathbb{R}$. It will now be shown that the~$\NormalisationPV{J}$~(and by extension the parity-violating SPOs) are in fact imaginary, so that~\cref{eq:app_NewHermicity} instead implies Hermicity or skew-Hermicity depending on~$P$ and~$\Pp{}$.

\paragraph*{Hermicity or orthonormality} When allowing for parity violation, we must extend the orthonormality condition in~\cref{eq:app_Orthonormality} to
\begin{equation}\label{eq:app_NewOrthonormality}
	\PVSPODownUp{J}{P}{\Pp{}}{X}{Y}{i}{j}{\mu}{\nu}
	\PVSPODownUp{\Jp{}}{\Ppp{}}{\Pppp{}}{Y}{Z}{k}{l}{\nu}{\sigma}
	\equiv
	\tensor*{\delta}{_{jk}}
	\tensor*{\delta}{_{J\Jp{}}}
	\tensor*{\delta}{_{\Pp{}\Ppp{}}}
	\PVSPODownUp{J}{P}{\Pppp{}}{X}{Z}{i}{l}{\mu}{\sigma}.
\end{equation}
Note that~\cref{eq:app_NewOrthonormality} is where we expect~$\NormalisationPV{J}$ to become important: its value will not depend on the~(arbitrary) way in which the reduced-index SPOs are weighted, but rather it will be determined by the combinatoric properties of the totally antisymmetric tensor in~\cref{eq:app_Eps}. First, consider the easy cases~$P=\Pp{}=\Ppp{}\neq\Pppp{}$ or~$P\neq\Pp{}=\Ppp{}=\Pppp{}$, i.e. the products of parity-preserving and parity-violating SPOs, and vice versa. For these cases, the definitions in~\cref{eq:app_Normalisation,eq:app_Normalisation2,eq:app_ParityViolating} automatically imply~\cref{eq:app_NewOrthonormality} for any values of~$\NormalisationPV{J}$. The only other case that can arise is the product of two parity-violating SPOs, i.e.~$\Ppp{}=\Pp{}\neq\Pppp{}= P$. This case yields~(for the only non-vanishing products in which~$\Jp{}= J$ and~$\States{J}{\Pp{}}{Y}{k}=\States{J}{\Pp{}}{Y}{j}$)
\begin{align}
	\PVSPODownUp{J}{P}{\Pp{}}{X}{Y}{i}{j}{\mu}{\nu}
	&
	\PVSPODownUp{J}{\Pp}{P}{Y}{Z}{j}{l}{\nu}{\sigma}
	\equiv
	\nonumber\\
	&
	\NormalisationPV{J}^2
	\Normalisation{J}{P}{X}{i}
	\Normalisation{J}{P}{Z}{l}
	\ReducedSPOUpDown{J}{P}{X}{i}{\sigma}{\mu}
		\Eps{_{\ParallelFieldIndices{J}{P}{\sigma}}^{\ParallelFieldIndices{J}{\Pp{}}{\rho}}}
		\Eps{_{\ParallelFieldIndices{J}{\Pp{}}{\rho}}^{\ParallelFieldIndices{J}{P}{\pi}}}
	\ReducedSPODownUp{J}{P}{Z}{l}{\pi}{\sigma}\quad \forall J<2.
	\label{eq:app_NormalisationRequirement}
\end{align}
We thus see how~\cref{eq:app_NormalisationRequirement} makes it clear why~$\NormalisationPV{J}$ is a~$J$-dependent factor. By comparing~\cref{eq:app_NormalisationRequirement} with~\cref{eq:app_NewOrthonormality} the criterion for~$\NormalisationPV{J}$ to satisfy~\cref{eq:app_NewOrthonormality} is determined to be
\begin{equation}\label{eq:app_Criterion}
		\NormalisationPV{J}^2
		\Eps{_{\ParallelFieldIndices{J}{P}{\sigma}}^{\ParallelFieldIndices{J}{\Pp{}}{\rho}}}
		\Eps{_{\ParallelFieldIndices{J}{\Pp{}}{\rho}}^{\ParallelFieldIndices{J}{P}{\pi}}}
		=
		\tensor*{\Delta}{^{\ParallelFieldIndices{J}{P}{\pi}}_{\ParallelFieldIndices{J}{P}{\sigma}}}.
\end{equation}
The values of~$\NormalisationPV{J}$ thus depend on the totally antisymmetric tensor in~\cref{eq:app_Eps}. The relevant identities are
\begin{equation}\label{eq:app_NormalisationPV}
		\Eps{^{\overline{\mu\nu\sigma}}}
		\Eps{_{\overline{\mu\nu\sigma}}}= -6,
		\quad
		\Eps{_{\overline{\mu}}^{\overline{\nu\sigma}}}
		\Eps{_{\overline{\nu\sigma}}^{\overline{\lambda}}} = 
		-2\tensor*{\delta}{_{\overline{\mu}}^{\overline{\lambda}}}
		,
		\quad
		\Eps{_{\overline{\mu\nu}}^{\overline{\sigma}}}
		\Eps{_{\overline{\sigma}}^{\overline{\lambda\rho}}} = 
		-2\tensor*{\delta}{_{[\overline{\mu}}^{[\overline{\lambda}}}
		\tensor*{\delta}{_{\overline{\nu}]}^{\overline{\rho}]}}
		,
\end{equation}
and these signal a potential problem:~\cref{eq:app_Criterion,eq:app_NormalisationPV} are not consistent with~$\NormalisationPV{J}\in\mathbb{R}$. If the solutions
\begin{equation}\label{eq:app_LowSpinSols}
	\NormalisationPV{0} \equiv  i/\sqrt{6},\quad
	\NormalisationPV{1} \equiv  i/\sqrt{2},
\end{equation}
are allowed, then the symmetry condition in~\cref{eq:app_NewOrthonormality} is no longer consistent with the more basic Hermicity condition. This leads to a ``\emph{catch-22}'' whereby, if one enforces Hermicity, then the factor~$\NormalisationPV{J}^2$ in~\cref{eq:app_Criterion} becomes~$\NormalisationPV{J}^*\NormalisationPV{J}>0$ so that even imaginary solutions fail to satisfy~\cref{eq:app_Criterion,eq:app_NormalisationPV}. We thus arrive at an interesting conclusion: \textit{Hermicity and orthonormality are mutually-exclusive properties of parity-violating spin-projection operators}, as pointed out in~\cite{Karananas:2014pxa}. 
It will be argued presently that orthonormality is a more convenient property than Hermicity. For this paper, therefore, the parity-preserving SPOs will be real~(and Hermitian), whilst the parity-violating SPOs will be imaginary~(and skew-Hermitian).\footnote{In fact, (skew-)Hermicity does not have to be tied to the real or imaginary character of SPOs. In our case,~\cref{eq:app_Hermicity} implies that we are selecting a convention in \PSALTer{} where the parity-preserving SPOs are real, and our ansatz in~\cref{eq:app_ParityViolating} then forces the parity-violating SPOs to be imaginary. One can, in principle, construct an alternative to~\cref{eq:app_ParityViolating} which is not linear in~$\NormalisationPV{J}$ but rather \emph{bilinear} in some other parameter. This would lead to \emph{all} SPOs being real and orthonormal~\cite{Karananas:2014pxa}.} Accordingly,~\cref{eq:app_NewHermicity} may be supplemented by the relation
\begin{equation}\label{eq:app_ChequerHermitianIndices}
	\PVSPOUpDown{J}{P}{\Pp{}}{X}{Y}{i}{j}{\mu}{\nu}^* \equiv 
	P\Pp{}
	\PVSPOUpDown{J}{\Pp{}}{P}{Y}{X}{j}{i}{\nu}{\mu}.
\end{equation}
When the SPOs are arranged in a block structure,~\cref{eq:app_NewHermicity,eq:app_ChequerHermitianIndices} lead to matrix representations with the property of \emph{chequer}-Hermitcity. Chequer-Hermicity is discussed in detail in~\cref{ChequerHermicity}.

\paragraph*{Higher-spin cases} For~$J\geq 2$, an alternative to the convention in~\cref{eq:app_ParityViolating} must be agreed upon, so that~\cref{eq:app_NewHermicity,eq:app_NewOrthonormality,eq:app_ChequerHermitianIndices} are preserved. The current implementation in~\PSALTer{} does not support tensors above the third rank: this means that~$J\leq 3$, and since there is no third-rank~$J^P=3^+$ representation to allow for parity violation in the~$J=3$ sector, new conventions are only needed for~$J=2$. Accordingly, these conventions will be
\begin{equation}\label{eq:app_ParityViolatingTwo}
	\begin{aligned}
		&\PVSPODownUp{2}{+}{-}{X}{Y}{i}{j}{\mu}{\nu}
	\equiv
	\NormalisationPV{2}\,
	\Normalisation{2}{+}{X}{i}
	\Normalisation{2}{-}{Y}{j}
		\ReducedSPOUpDownGeneral{2}{+}{X}{i}{\overline{\sigma\pi}}{\mu}
		\Eps{_{\overline{\sigma}}^{\overline{\rho\kappa}}}
		\ReducedSPODownUpGeneral{2}{-}{Y}{j}{\overline{\rho\kappa\pi}}{\nu},
	\end{aligned}
\end{equation}
where the~\PSALTer{} conventions for the reduced-index SPOs are
\begin{equation}\label{eq:app_ReducedIndexTwoConventions}
	\ReducedSPOUpDownGeneral{2}{+}{X}{i}{\left[\overline{\sigma\pi}\right]}{\mu} \equiv 
	\ReducedSPODownUpGeneral{2}{-}{Y}{j}{\left(\overline{\rho\kappa}\right)\overline{\pi}}{\nu} \equiv 
	\Eps{^{\overline{\rho\kappa\pi}}}
	\ReducedSPODownUpGeneral{2}{-}{Y}{j}{\overline{\rho\kappa\pi}}{\nu} \equiv 0.
\end{equation}
From~\cref{eq:app_Criterion,eq:app_NormalisationPV,eq:app_ParityViolatingTwo,eq:app_ReducedIndexTwoConventions} it follows that~\cref{eq:app_LowSpinSols} is extended by
\begin{equation}\label{eq:app_HighSpinSols}
	\NormalisationPV{2} \equiv  i/\sqrt{2}.
\end{equation}
With the final determination in~\cref{eq:app_HighSpinSols}, all the parity-violating SPOs have been constructed.

\section{Explicit operator formulae}\label{ExplicitSpinProjectorOperators}

\paragraph*{Spin-projection operator tables} In this appendix we provide explicit formulae for the SPOs which are created automatically by \PSALTer{} at runtime. The production of these formulae is not part of the \PSALTer{} functionality, as it was pointed out in~\cite{Barker:2024juc} that \PSALTer{} itself makes the tabulation of SPOs redundant when presenting future spectroscopy results.\footnote{This is because the field kinematics presented in~\cref{FieldKinematicsScalarField,FieldKinematicsVectorField,FieldKinematicsTwoFormField,FieldKinematicsThreeFormField,FieldKinematicsTetradPerturbation,FieldKinematicsSpinConnection} already provides an implicit statement of all the SPOs used, and their explicit formulae may be recovered from these by an application of~\cref{eq:app_ReducedIndex,eq:app_Normalisation,eq:app_ParityViolating,eq:app_ParityViolatingTwo}.} The current paper, however, makes some significant developments in the theory and conventions of the SPOs themselves: explicit formulae may therefore be a useful companion to~\cref{ConstructionOfSpinProjectionOperators}. The SPOs corresponding to the massive two-form and `one-by-two' CSKR theory are somewhat trivial, and these we omit. We provide in~\cref{SPOMatricesZeroByThreeCSKRTheorySpin0,SPOMatricesZeroByThreeCSKRTheorySpin1} the SPOs corresponding to the analyses of `zero-by-three' CSKR theory. We also provide in~\cref{SPOMatricesEinsteinCartanTheorySpin0,SPOMatricesEinsteinCartanTheorySpin1,SPOMatricesEinsteinCartanTheorySpin2} the SPOs corresponding to the analyses of EC theory (i.e., Poincar\'e gauge theory).

\begin{table*}[!t]
	\includegraphics[width=\linewidth]{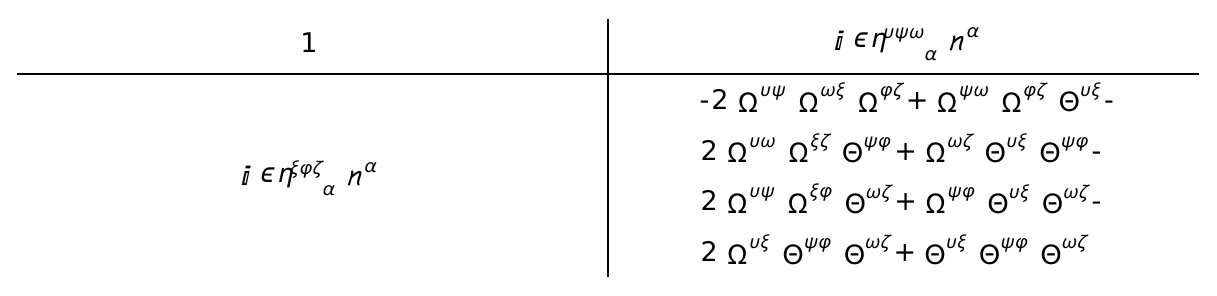}
	\caption{The matrix of spin-projection operators associated with the spin-zero sector of `zero-by-three' CSKR theory. The row-rank ordering of this matrix corresponds exactly to the first diagonal block in~\cref{ParticleSpectrographZeroByThreeCSKRTheory}. See~\cref{UnitTimelikeComponents,eq:app_vector_projectors} for definitions of quantities. Indices to be contracted with the complex conjugate fields~$\FieldDownConj{X}{\mu}$ (if any) are drawn from the set~$\left\{\xi,\varphi,\zeta\right\}$ in order; those to be contracted with~$\FieldUp{X}{\mu}$ are drawn from~$\left\{\upsilon,\psi,\omega\right\}$.}
\label{SPOMatricesZeroByThreeCSKRTheorySpin0}
\end{table*}
\begin{table*}[!t]
	\includegraphics[width=0.5\linewidth]{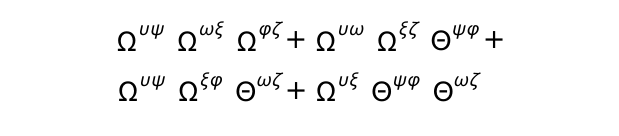}
	\caption{The matrix of spin-projection operators associated with the spin-one sector of `zero-by-three' CSKR theory. The row-rank ordering of this matrix corresponds exactly to the second diagonal block in~\cref{ParticleSpectrographZeroByThreeCSKRTheory}. See~\cref{UnitTimelikeComponents,eq:app_vector_projectors} for definitions of quantities. Indices to be contracted with the complex conjugate fields~$\FieldDownConj{X}{\mu}$ (if any) are drawn from the set~$\left\{\xi,\varphi,\zeta\right\}$ in order; those to be contracted with~$\FieldUp{X}{\mu}$ are drawn from~$\left\{\upsilon,\psi,\omega\right\}$.}
\label{SPOMatricesZeroByThreeCSKRTheorySpin1}
\end{table*}
\begin{table*}[!t]
	\includegraphics[width=\linewidth]{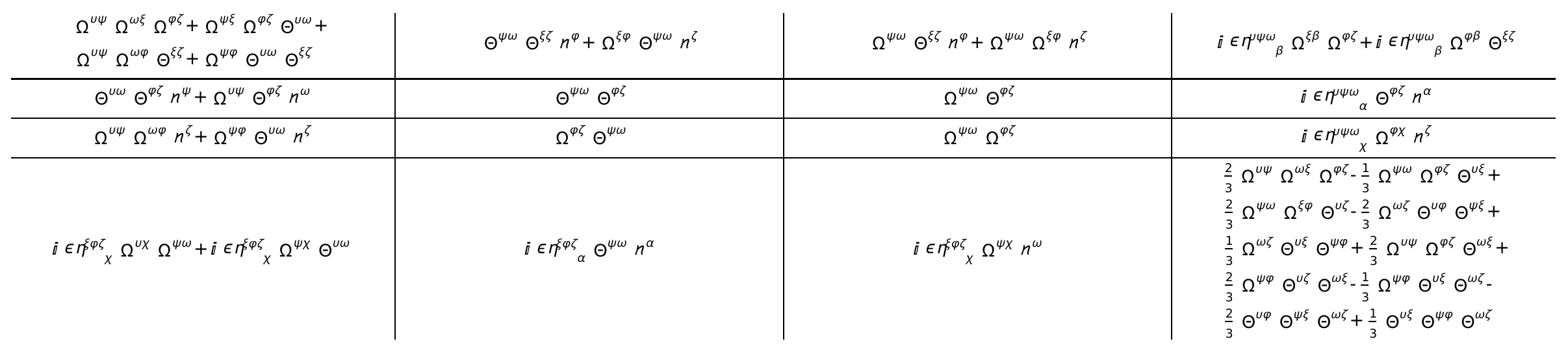}
	\caption{The matrix of spin-projection operators associated with the spin-zero sector of Poincar\'e gauge theory. The row-rank ordering of this matrix corresponds exactly to the first diagonal block in~\cref{ParticleSpectrographScalarParityViolatingPGT,ParticleSpectrographGeneralParityViolatingPGT}. See~\cref{UnitTimelikeComponents,eq:app_vector_projectors} for definitions of quantities. Indices to be contracted with the complex conjugate fields~$\FieldDownConj{X}{\mu}$ (if any) are drawn from the set~$\left\{\xi,\varphi,\zeta\right\}$ in order; those to be contracted with~$\FieldUp{X}{\mu}$ are drawn from~$\left\{\upsilon,\psi,\omega\right\}$.}
\label{SPOMatricesEinsteinCartanTheorySpin0}
\end{table*}
\begin{table*}[!t]
	\includegraphics[width=\linewidth]{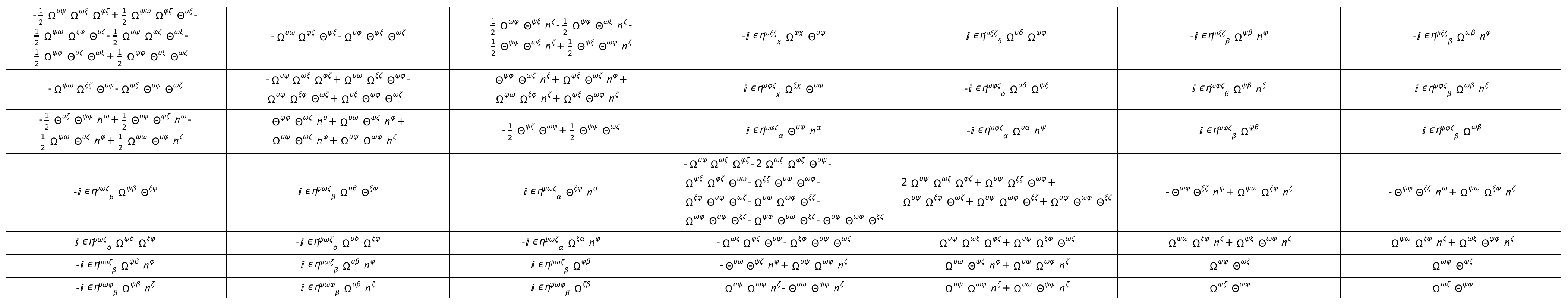}
	\caption{The matrix of spin-projection operators associated with the spin-one sector of Poincar\'e gauge theory. The row-rank ordering of this matrix corresponds exactly to the second diagonal block in~\cref{ParticleSpectrographScalarParityViolatingPGT,ParticleSpectrographGeneralParityViolatingPGT}. See~\cref{UnitTimelikeComponents,eq:app_vector_projectors} for definitions of quantities. Indices to be contracted with the complex conjugate fields~$\FieldDownConj{X}{\mu}$ (if any) are drawn from the set~$\left\{\xi,\varphi,\zeta\right\}$ in order; those to be contracted with~$\FieldUp{X}{\mu}$ are drawn from~$\left\{\upsilon,\psi,\omega\right\}$.}
\label{SPOMatricesEinsteinCartanTheorySpin1}
\end{table*}
\begin{table*}[!t]
	\includegraphics[width=\linewidth]{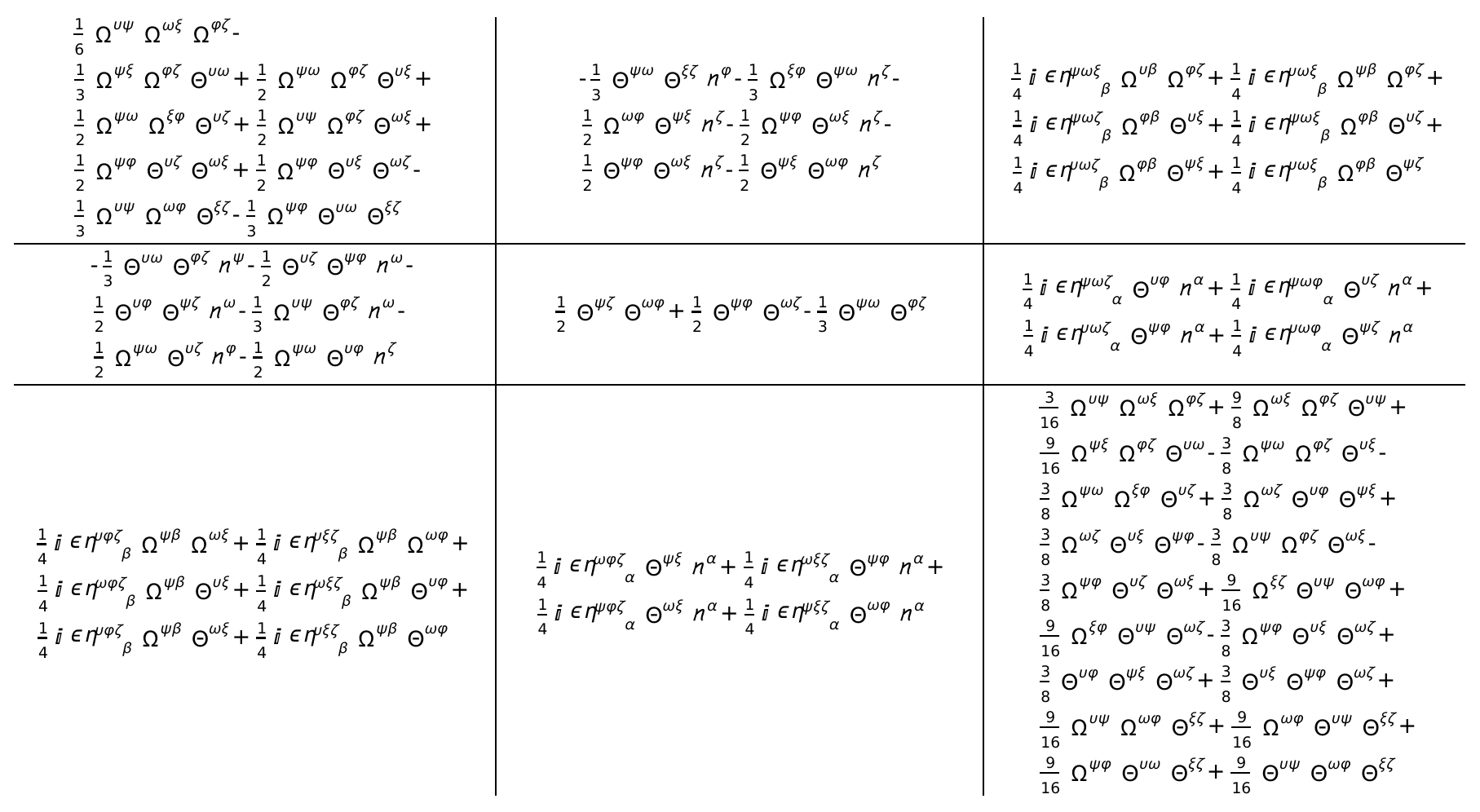}
	\caption{The matrix of spin-projection operators associated with the spin-two sector of Poincar\'e gauge theory. The row-rank ordering of this matrix corresponds exactly to the third diagonal block in~\cref{ParticleSpectrographScalarParityViolatingPGT,ParticleSpectrographGeneralParityViolatingPGT}. See~\cref{UnitTimelikeComponents,eq:app_vector_projectors} for definitions of quantities. Indices to be contracted with the complex conjugate fields~$\FieldDownConj{X}{\mu}$ (if any) are drawn from the set~$\left\{\xi,\varphi,\zeta\right\}$ in order; those to be contracted with~$\FieldUp{X}{\mu}$ are drawn from~$\left\{\upsilon,\psi,\omega\right\}$.}
\label{SPOMatricesEinsteinCartanTheorySpin2}
\end{table*}

\section{Chequer-Hermicity}\label{ChequerHermicity}

\paragraph*{A new kind of structured matrix} In this appendix we introduce the concept of \emph{chequer}-Hermicity as a generalisation of Hermicity which becomes physically relevant in the analysis of parity-violating particle spectra. A chequer-Hermitian matrix admits a~$2\times 2$ block structure in which the diagonal blocks are Hermitian, but the off-diagonal blocks are skew-Hermitian; thus, it has the structure of a chequerboard. The spectral theorem applies only to the Hermitian or skew-Hermitian limits of a chequer-Hermitian matrix. Nonetheless, we observe that chequer-Hermicity leads to some remarkably convenient properties.

\begin{definition}
	We define the \emph{chequer-Hermitian conjugate}~$\GeneralConj{}$ of an arbitrary complex square matrix~$\General{}$ as resulting from the following operation on its~$2\times 2$ block structure
	\begin{equation}\label{ChequerHermitianDef}
		\General{}\equiv
		\Chequer{\MCa{}}{\MCb{}}{\MCc{}}{\MCd{}}
		\implies
		\GeneralConj{}\equiv
		\Chequer{\MCaConj{}}{-\MCcConj{}}{-\MCbConj{}}{\MCdConj{}}
		,
	\end{equation}
	where~$\MCa{}$ and~$\MCd{}$ are square complex matrices, and~$\MCb{}$ and~$\MCc{}$ are~(possibly rectangular) complex matrices which are conformable for operations such as~$\MCb{}\MCc{}$ and~$\MCc{}\MCa{}$ etc. The dimensions of these matrices are determined by the context of the problem.
\end{definition}

\begin{remark}
	The notation in~\cref{ChequerHermitianDef} is chosen to reflect the parity indices of different kinds of SPOs. Thus~($+$) and~($-$) denote parity-preserving sectors, and~($\pm$) and~($\mp$) denote parity-violating sectors. Within this paper, parity is the context which will provide the dimensions of the~$2\times 2$ division in chequer-Hermitian conjugation.
\end{remark}

\begin{definition}
	We say~$\ChequerHermitian{}$ is \emph{chequer-Hermitian} when~$\ChequerHermitian{}\equiv\ChequerHermitianConj{}$.
\end{definition}

\begin{corollary}
	Let~$\ChequerHermitian{}$ be chequer-Hermitian, then it has the~$2\times 2$ block structure
	\begin{equation}\label{ChequerHermitian}
		\ChequerHermitian{}\equiv\ChequerHermitianConj{}
		\implies
		\ChequerHermitian{}\equiv
		\Chequer{\Ca{}}{\Cb{}}{-\CbConj{}}{\Cd{}}
		,
		\quad
		\Ca{}\equiv\CaConj{},
		\quad
		\Cd{}\equiv\CdConj{}.
	\end{equation}
\end{corollary}

\begin{proof}
	This follows immediately from~\cref{ChequerHermitianDef}.
\end{proof}

\begin{corollary}
	The product of two chequer-Hermitian matrices is also chequer-Hermitian.
\end{corollary}

\begin{proof}
	This follows immediately from~\cref{ChequerHermitian}.
\end{proof}

\begin{corollary}
	Let~$\ChequerHermitian{}$ be chequer-Hermitian and invertible, then~$\ChequerHermitian{}^{-1}$ is also chequer-Hermitian.
\end{corollary}

\begin{proof}
	The proof follows immediately from the well-known formula for the inverse of a block matrix. Alternatively, we can infer the result as follows. Let~$\ChequerHermitian{}\equiv\HermitianDiagonal{}+\SkewHermitianOffDiagonal{}$ be decomposed into the diagonal~($d$) and off-diagonal~($o$) block parts which must be respectively Hermitian~($h$) and skew-Hermitian~($s$) by the property~$\ChequerHermitian{}\equiv\ChequerHermitianConj{}$. Without assuming chequer-Hermicity of the inverse, we have~$\ChequerHermitian{}^{-1}\equiv\InverseHermitianDiagonal{}+\InverseHermitianOffDiagonal{}+\InverseSkewHermitianDiagonal{}+\InverseSkewHermitianOffDiagonal{}$. By letting~$\ChequerHermitian{}^{-1}$ act as a left inverse and taking the Hermitian conjugate we have
\begin{subequations}
	\begin{align}
		&\InverseHermitianDiagonal{}\cdot\HermitianDiagonal{}
			+\InverseHermitianOffDiagonal{}\cdot\SkewHermitianOffDiagonal{}
			+\InverseSkewHermitianDiagonal{}\cdot\HermitianDiagonal{}
			+\InverseSkewHermitianOffDiagonal{}\cdot\SkewHermitianOffDiagonal{}
		\nonumber\\	&\hspace{120pt}
			\equiv
			\HermitianDiagonal{}\cdot\InverseHermitianDiagonal{}
			+\SkewHermitianOffDiagonal{}\cdot\InverseSkewHermitianOffDiagonal{}
			-\HermitianDiagonal{}\cdot\InverseSkewHermitianDiagonal{}
			-\SkewHermitianOffDiagonal{}\cdot\InverseHermitianOffDiagonal{}
			\equiv\mathsf{1},
			\label{InverseExpansion1}
			\\
			&\InverseHermitianOffDiagonal{}\cdot\HermitianDiagonal{}
			+\InverseHermitianDiagonal{}\cdot\SkewHermitianOffDiagonal{}
			+\InverseSkewHermitianOffDiagonal{}\cdot\HermitianDiagonal{}
			+\InverseSkewHermitianDiagonal{}\cdot\SkewHermitianOffDiagonal{}
		\nonumber\\	&\hspace{120pt}
			\equiv
			\HermitianDiagonal{}\cdot\InverseSkewHermitianOffDiagonal{}
			+\SkewHermitianOffDiagonal{}\cdot\InverseHermitianDiagonal{}
			-\HermitianDiagonal{}\cdot\InverseHermitianOffDiagonal{}
			-\SkewHermitianOffDiagonal{}\cdot\InverseSkewHermitianDiagonal{}
			\equiv\mathsf{0},
			\label{InverseExpansion2}
	\end{align}
\end{subequations}
	and from~\cref{InverseExpansion1,InverseExpansion2} we deduce that~$\InverseHermitianDiagonal{}+\InverseSkewHermitianOffDiagonal{}-\InverseHermitianOffDiagonal{}-\InverseSkewHermitianDiagonal{}$ is the right inverse of~$\ChequerHermitian{}$. The uniqueness of the inverse for square~$\ChequerHermitian{}$ implies that~$\InverseHermitianOffDiagonal{}\equiv\InverseSkewHermitianDiagonal{}\equiv\mathsf{0}$ so that~$\ChequerHermitian{}^{-1}\equiv\left(\ChequerHermitian{}^{-1}\right)^{\ddagger}$ as required.
\end{proof}

\begin{theorem}
	Let~$\ChequerHermitian{}$ be chequer-Hermitian and singular, with an orthonormal set of complex right null eigenvectors~$\left\{\RightNullVector{i}\right\}$, then an orthonormal set of left null eigenvectors~$\left\{\LeftNullVector{i}\right\}$ may be constructed according to
	\begin{equation}\label{RightToLeft}
		\begin{aligned}
		&\ChequerHermitian{}\cdot\RightNullVector{i}\equiv\mathsf{0},
		\quad
		\RightNullVectorConj{i}\cdot\RightNullVector{j}\equiv\tensor*{\delta}{_{ij}},
		\quad
		\RightNullVector{i}\equiv
		\ChequerVector{\PlusNullVector{i}}{\MinusNullVector{i}}
		\\&\hspace{120pt}
		\implies
		\LeftNullVectorConj{i}\cdot\ChequerHermitian{}\equiv\mathsf{0},
		\quad
		\LeftNullVectorConj{i}\cdot\LeftNullVector{j}\equiv\tensor*{\delta}{_{ij}},
		\quad
		\LeftNullVectorConj{i}\equiv
		\ChequerVectorConj{\PlusNullVectorConj{i}}{-\MinusNullVectorConj{i}}
		.
		\end{aligned}
	\end{equation}
\end{theorem}

\begin{proof}
	This follows immediately from~\cref{ChequerHermitian}.
\end{proof}

\section{Moore--Penrose pseudoinversion}\label{MoorePenrose}

\paragraph*{The natural choice of pseudoinverse} In this appendix we introduce the Moore--Penrose pseudoinverse of a general complex square matrix, and provide formulae for the pseudoinverse of Hermitian and chequer-Hermitian matrices. We also show that the pseudoinverse of a chequer-Hermitian matrix can be computed from the null eigenvectors of the matrix. The very simple formulae that follow cement the status of the Moore--Penrose pseudoinverse as the natural choice of pseudoinverse for implementation in particle spectroscopy.

\begin{definition}
	We define the unique \emph{Moore--Penrose pseudoinverse}~$\GeneralMP{}$ of an arbitrary complex square matrix~$\General{}$ to have the following four properties:
		\begin{equation}
			\General{}\cdot\GeneralMP{}\cdot\General{}\equiv\General{},
			\quad
			\GeneralMP{}\cdot\General{}\cdot\GeneralMP{}\equiv\GeneralMP{},
			\quad
			\General{}\cdot\GeneralMP{}\equiv\left(\General{}\cdot\GeneralMP{}\right)^\dagger,
			\quad
			\GeneralMP{}\cdot\General{}\equiv\left(\GeneralMP{}\cdot\General{}\right)^\dagger.
			\label{MP1}
		\end{equation}
\end{definition}

\begin{corollary}
	Let~$\General{}$ be an arbitrary complex square matrix, then
	\begin{equation}\label{MPFormula}
		\GeneralMP{}\equiv\left(\General{}^\dagger\cdot\General{}\right)^{+}\cdot\General{}^\dagger
		\equiv\General{}^\dagger\cdot\left(\General{}\cdot\General{}^\dagger\right)^{+}.
	\end{equation}
\end{corollary}

\begin{proof}
	This~(well-known) formula is verified by substituting~\cref{MPFormula} into~\cref{MP1}.
\end{proof}

\begin{remark}
	\cref{MPFormula} is useful because it allows~$\GeneralMP{}$ to be computed if a general formula is known for the Moore--Penrose pseudoinverse of Hermitian matrices such as~$\General{}^\dagger\cdot\General{}$ or~$\General{}\cdot\General{}^\dagger$.
\end{remark}

\begin{corollary}
	Let~$\Hermitian{}$ be Hermitian and singular, with an orthonormal set of complex null eigenvectors~$\left\{\RightNullVector{i}\right\}$, then~$\HermitianMP{}$ is given by
	\begin{equation}\label{MPHermitian}
		\HermitianMP{}\equiv
			\bigg(\mathsf{1}-\sum_i \RightNullVector{i}\cdot\RightNullVectorConj{i}\bigg)\cdot
			\bigg(\Hermitian{}+\sum_j \RightNullVector{j}\cdot\RightNullVectorConj{j}\bigg)^{-1}\cdot
			\bigg(\mathsf{1}-\sum_k \RightNullVector{k}\cdot\RightNullVectorConj{k}\bigg).
	\end{equation}
\end{corollary}

\begin{proof}
	This formula is verified by substituting~\cref{MPChequerHermitian} into~\cref{MP1}.
\end{proof}

\begin{remark}
	\cref{MPHermitian} is useful because it allows~$\HermitianMP{}$ to be computed from the null eigenvectors of~$\Hermitian{}$.
\end{remark}

\begin{corollary}
	Let~$\Hermitian{}$ be Hermitian and singular, then~$\HermitianMP{}$ is also Hermitian.
\end{corollary}

\begin{proof}
	This follows immediately from~\cref{MPHermitian}.
\end{proof}

\begin{theorem}
	Let~$\ChequerHermitian{}$ be chequer-Hermitian and singular, with orthonormal sets of complex right and left null eigenvectors~$\left\{\RightNullVector{i}\right\}$ and~$\left\{\LeftNullVector{i}\right\}$ respectively, then the Moore--Penrose pseudoinverse~$\ChequerHermitianMP{}$ is given by
	\begin{equation}\label{MPChequerHermitian}
		\ChequerHermitianMP{}\equiv
			\bigg(\mathsf{1}-\sum_i \RightNullVector{i}\cdot\RightNullVectorConj{i}\bigg)\cdot
			\bigg(\ChequerHermitian{}+\sum_j \LeftNullVector{j}\cdot\RightNullVectorConj{j}\bigg)^{-1}\cdot
			\bigg(\mathsf{1}-\sum_k \LeftNullVector{k}\cdot\LeftNullVectorConj{k}\bigg).
	\end{equation}
\end{theorem}

\begin{proof}
	This formula is verified by substituting~\cref{MPChequerHermitian} into~\cref{MP1}. Alternatively, it may be deduced directly from~\cref{MPFormula} by noting that~$\ChequerHermitian{}^\dagger\cdot\ChequerHermitian{}$ and~$\ChequerHermitian{}\cdot\ChequerHermitian{}^\dagger$ are Hermitian and singular, and hence pseudoinvertible via the formula in~\cref{MPHermitian}.
\end{proof}

\begin{remark}
	\cref{MPChequerHermitian,RightToLeft} are useful because they allow~$\ChequerHermitianMP{}$ to be computed from the~(right) null eigenvectors of~$\ChequerHermitian{}$. Note that~\cref{MPChequerHermitian} is consistent with~(and is a minimal modification of) the formula in~\cref{MPHermitian} for the Moore--Penrose pseudoinverse of a Hermitian matrix.
\end{remark}

\begin{corollary}
	Let~$\ChequerHermitian{}$ be chequer-Hermitian and singular, then the Moore--Penrose pseudoinverse~$\ChequerHermitianMP{}$ is also chequer-Hermitian.
\end{corollary}

\begin{proof}
	This follows immediately from~\cref{MPChequerHermitian}.
\end{proof}

\section{Operator coefficient matrices}\label{OperatorCoefficientMatrices}

\paragraph*{Physicality equals chequer-Hermicity} In this appendix we show that the operator coefficient matrices used in~\cref{TheoreticalDevelopment,NoGhostCriterionAppendix} have a chequer-Hermitian structure. Despite this fact, we particularly emphasise that \emph{the operators themselves are always Hermitian}~\cite{Karananas:2014pxa} --- this is in line with the basic requirements of physicality. The apparent discrepancy arises only because of our discussion in~\cref{ConstructionOfSpinProjectionOperators}, in which a chequer-Hermitian but orthonormal basis of SPOs was selected. To begin, notice how~\cref{ToMatrices1} implies that the block-consituents~$\WaveOperatorJ{J}$ of the wave operator coefficient matrix~$\WaveOperator{}$ defined in~\cref{BuildWaveOperator} are structured matrices. To see this, we substitute~\cref{ToMatrices1} into~\cref{MomentumRepresentation_suppressed} so that the quadratic theory (with sources suppressed) becomes 
\begin{equation}\label{ActionExpanded}
	S
	=
	\frac{1}{(2\pi)^4}\int\mathrm{d}^4k
	\sum_X
	\Bigg[
	\FieldDownConj{X}{\mu}
		\sum_Y
		\sum_{J,P,\Pp{}}\sum_{\States{J}{P}{X}{i},\States{J}{\Pp{}}{Y}{j}}\tensor{\left[\WaveOperatorJ{J}\right]}{_{\States{J}{P}{X}{i}\States{J}{\Pp{}}{Y}{j}}}\PVSPOUpDown{J}{P}{\Pp{}}{X}{Y}{i}{j}{\mu}{\nu}\FieldUp{Y}{\nu}
	\Bigg].
\end{equation}
Since~$S$ must be real, it is possible to take the complex conjugate of~\cref{ActionExpanded} and use~\cref{NewHermicity,ChequerHermitianIndices} to show 
\begin{equation}\label{ActionExpandedConjugate}
	S
	=
	\frac{1}{(2\pi)^4}\int\mathrm{d}^4k
	\sum_X
	\Bigg[
		\FieldUp{X}{\mu}
		\sum_Y
		\sum_{J,P,\Pp{}}\sum_{\States{J}{P}{X}{i},\States{J}{\Pp{}}{Y}{j}}
		\tensor{\left[\WaveOperatorJConj{J}\right]}{_{\States{J}{P}{X}{i}\States{J}{\Pp{}}{Y}{j}}}
		\PVSPOUpDown{J}{\Pp{}}{P}{Y}{X}{j}{i}{\nu}{\mu}
		\FieldDownConj{Y}{\nu}
	\Bigg].
\end{equation}
By comparing~\cref{ActionExpanded,ActionExpandedConjugate} it follows that
\begin{equation}
\tensor{\left[\WaveOperatorJ{J}\right]}{_{\States{J}{P}{X}{i}\States{J}{\Pp{}}{Y}{j}}}\PVSPOUpDown{J}{P}{\Pp{}}{X}{Y}{i}{j}{\mu}{\nu} = \tensor{\left[\WaveOperatorJConj{J}\right]}{_{\States{J}{P}{X}{i}\States{J}{\Pp{}}{Y}{j}}}
		\PVSPOUpDown{J}{\Pp{}}{P}{X}{Y}{j}{i}{\mu}{\nu}  \ ,
\end{equation}
where we use the notation defined in~\cref{ChequerHermitianDef}. This means that~$\WaveOperatorJ{J}$ and~$\WaveOperatorJConj{J}$ serve equally well as the coefficient matrix representation of the wave operator. Moving forwards, therefore, it is always safe to assume
\begin{equation}\label{ChequerHermitianWaveOperator}
	\WaveOperatorJ{J}= \WaveOperatorJConj{J}\quad \forall J \iff \WaveOperator{} = \WaveOperatorConj{},
\end{equation}
where the implication in~\cref{ChequerHermitianWaveOperator} follows from the block structure in~\cref{BuildWaveOperator}. This result follows immediately from the completely general claim that the wave operator itself must be Hermitian
\begin{equation}\label{WaveOperatorHermitian}
	\WaveOperatorTensorUpDown{X}{Y}{\mu}{\nu}^*\equiv\WaveOperatorTensorDownUp{Y}{X}{\nu}{\mu},
\end{equation}
since, if~\cref{WaveOperatorHermitian} is taken to be a convincing starting point, then one need only substitute~\cref{ToMatrices1} to arrive at~\cref{ChequerHermitianWaveOperator}. In summary, only the chequer-Hermitian part of the coefficient matrix~$\WaveOperatorJ{J}$ contributes to the physics. Equivalently, even if~$\WaveOperatorJ{J}$ is not explicitly chequer-Hermitian at its point of construction, it suffices to work only with its chequer-Hermitian part.\footnote{And indeed, the this operation is always equivalent to modifying the theory by the addition of boundary terms.}

\paragraph*{Propagator coefficient matrix} By this point, our conclusion in~\cref{ChequerHermitianWaveOperator} and the mathematical results of~\cref{ChequerHermicity,MoorePenrose} allow us to derive an explicit and highly compact formula the propagator coefficient matrix. Let~$\NullVectorsPV{J}{a}$ be labels for the null vectors of the chequer-Hermitian wave operator coefficient matrix block~$\WaveOperatorJ{J}$, so that~$\big\{\RightNullVectorPV{J}{a}\big\}$ and~$\big\{\LeftNullVectorPV{J}{a}\big\}$ are orthonormal sets of right and left null vectors respectively~(note that the latter can be deduced from the former via~\cref{RightToLeft}). For almost all models of physical relevance, it is the case that these vectors are purely functions of~$k$, being independent of the Lagrangian coupling coefficients.\footnote{Note that this is not always the case: there are instances of higher-spin Fronsdal-type models which have `parametric' gauge symmetries, in which the null vectors are smoothly parameterised by the Lagrangian coupling coefficients.} From~\cref{MPChequerHermitian} it then follows that
\begin{equation}\label{MoorePenroseJPractice}
	\begin{gathered}
	\MoorePenroseJ{J}\equiv\VimilarityJ{J}\cdot\WaveOperatorFullJInv{J}\cdot\SimilarityJ{J},
	\\
	\VimilarityJ{J}\equiv\mathsf{1}-\sum_{\NullVectorsPV{J}{a}}\RightNullVectorPV{J}{a}\cdot\RightNullVectorPVConj{J}{a},
	\quad
	\WaveOperatorFullJ{J}\equiv\WaveOperatorJ{J}+\sum_{\NullVectorsPV{J}{a}}\LeftNullVectorPV{J}{a}\cdot\RightNullVectorPVConj{J}{a},
	\quad
	\SimilarityJ{J}\equiv\mathsf{1}-\sum_{\NullVectorsPV{J}{a}}\LeftNullVectorPV{J}{a}\cdot\LeftNullVectorPVConj{J}{a},
	\end{gathered}
\end{equation}
where~$\WaveOperatorFullJ{J}$ is also chequer-Hermitian. 

\section{No-ghost criterion}\label{NoGhostCriterionAppendix}

\paragraph*{Restoring Hermicity} In this appendix we obtain our central result in~\cref{NoGhostCriterionMixed}, which is the most delicate change to the algorithm induced by parity violation. From the definition of the saturated propagator~\cref{eq:saturated_propagator_SPOs}, the no-ghost criterion can be expressed in terms of the coefficient matrix as
\begin{equation}\label{NoGhostCriterion}
	\NewRes{k^2}{\PVSquareMass{J}{s}}\Bigg(
	\sum_{X,Y}\sum_{P,\Pp{}}\sum_{\States{J}{P}{X}{i},\States{J}{\Pp{}}{Y}{j}}\tensor{\left[\MoorePenroseJ{J}\right]}{_{\States{J}{P}{X}{i}\States{J}{\Pp{}}{Y}{j}}}\SourceDownConj{X}{\mu}\PVSPOUpDown{J}{P}{\Pp{}}{X}{Y}{i}{j}{\mu}{\nu}\SourceUp{Y}{\nu}\Bigg)\geq 0\quad \forall J,\PVMasses{J}{s}.
\end{equation}
The source currents~$\SourceUp{X}{\mu}$ are arbitrary, and the SPOs remain finite in the rest frame of the massive~$\PVMasses{J}{s}$-particle. We can therefore write~\cref{NoGhostCriterion} in a more compact form as
\begin{equation}\label{NoGhostCriterionCompact}
	\tr\left(\MoorePenroseSJ{s}{J}\cdot\LimitPartialSourceMatrix{s}{J}\right)\geq 0\quad \forall J,\quad \forall \PVMasses{J}{s},
	\quad
	\MoorePenroseSJ{s}{J}\equiv\NewRes{k^2}{\PVSquareMass{J}{s}}\left(\MoorePenroseJ{J}\right),
	\quad	
	\LimitPartialSourceMatrix{s}{J}\equiv\NewLim{k^2}{\PVSquareMass{J}{s}}\left(\PartialSourceMatrix{J}\right),
\end{equation}
where the source matrix is defined as~$\tensor{\left[\PartialSourceMatrix{J}\right]}{_{\States{J}{P}{X}{i}\States{J}{\Pp{}}{Y}{j}}}\equiv\SourceDownConj{X}{\mu}\PVSPOUpDown{J}{P}{\Pp{}}{X}{Y}{i}{j}{\mu}{\nu}\SourceUp{Y}{\nu}$. In the~$2\times 2$ block form provided by the parity indices~(see~\cref{ChequerHermicity}), the relevant chequer-Hermitian matrices are notated firstly in~\cref{eq:saturated_propagator_SPOs} --- where~$\MoorePenroseJMP{J}\equiv-\left(\MoorePenroseJPM{J}\right)^\dagger$ by the condition~$\MoorePenroseJ{J}\equiv\left(\MoorePenroseJ{J}\right)^\ddagger$ --- and secondly as
\begin{equation}
	\PartialSourceMatrix{J}\equiv
	\Chequer{\SourceDownConj{X}{\mu}\PVSPOUpDown{J}{+}{+}{X}{Y}{i}{j}{\mu}{\nu}\SourceUp{Y}{\nu}}{\SourceDownConj{X}{\mu}\PVSPOUpDown{J}{+}{-}{X}{Y}{i}{j}{\mu}{\nu}\SourceUp{Y}{\nu}}{\SourceDownConj{X}{\mu}\PVSPOUpDown{J}{-}{+}{X}{Y}{i}{j}{\mu}{\nu}\SourceUp{Y}{\nu}}{\SourceDownConj{X}{\mu}\PVSPOUpDown{J}{-}{-}{X}{Y}{i}{j}{\mu}{\nu}\SourceUp{Y}{\nu}}.
\end{equation}
At this point it is convenient to trade the orthonormality of the SPOs in exchange for the Hermicity of the propagator coefficient matrix by defining\footnote{Note that this merely constitutes a change of basis for the coefficient matrix: as emphasised in~\cref{OperatorCoefficientMatrices} the underlying \emph{operators} are always Hermitian.}
\begin{equation}
	\HermitianMoorePenroseJ{J}
	\equiv
	\Chequer{\mathsf{1}}{\mathsf{0}}{\mathsf{0}}{-\mathsf{1}}\cdot\MoorePenroseJ{J},
	\quad
	\ConstrainedPartialSourceMatrix{J}\equiv
	\PartialSourceMatrix{J}\cdot\Chequer{\mathsf{1}}{\mathsf{0}}{\mathsf{0}}{-\mathsf{1}},
\end{equation}
so that~$\HermitianMoorePenroseJ{J}\equiv\left(\HermitianMoorePenroseJ{J}\right)^\dagger$ and without loss of generality the no-ghost condition in~\cref{NoGhostCriterionCompact} becomes
\begin{equation}\label{NoGhostCriterionMoreCompact}
	\tr\left(\ResHermitianMoorePenroseJ{s}{J}\cdot\LimitConstrainedPartialSourceMatrix{s}{J}\right)\geq 0\quad \forall J,\quad \forall \PVMasses{J}{s},
	\quad
	\ResHermitianMoorePenroseJ{s}{J}\equiv\NewRes{k^2}{\PVSquareMass{J}{s}}\left(\HermitianMoorePenroseJ{J}\right),
	\quad	
	\LimitConstrainedPartialSourceMatrix{s}{J}\equiv\NewLim{k^2}{\PVSquareMass{J}{s}}\left(\ConstrainedPartialSourceMatrix{J}\right).
\end{equation}
If~$\HermitianMoorePenroseJ{J}$ has a simple pole as~$k^2\mapsto\PVSquareMass{J}{s}$ then by Hermicity~$\ResHermitianMoorePenroseJ{s}{J}$ must have one real non-zero eigenvalue~$\PartialEigenvalue{s}{J}$ which depends exclusively on the Lagrangian coupling coefficients\footnote{See also similar arguments in~\cite{Lin:2018awc}.}. The corresponding~(normalised) eigenvector can be thought of as the direct sum of vectors\footnote{We do not yet assume that~$\PartialEigenvector{s}{J}{+}$ or~$\PartialEigenvector{s}{J}{-}$ are individually eigenvectors.} belonging to parity-even and parity-odd sub-spaces, so that
\begin{equation}\label{ClearDivision}
	\begin{gathered}
	\ResHermitianMoorePenroseJ{s}{J}
	\equiv
	\PartialEigenvalue{s}{J}
	\Eigenvector{s}{J}\cdot\EigenvectorConj{s}{J}
	\equiv
	\PartialEigenvalue{s}{J}
	\Chequer{\PartialEigenvector{s}{J}{+}\cdot\PartialEigenvectorConj{s}{J}{+}}{\PartialEigenvector{s}{J}{+}\cdot\PartialEigenvectorConj{s}{J}{-}}{\PartialEigenvector{s}{J}{-}\cdot\PartialEigenvectorConj{s}{J}{+}}{\PartialEigenvector{s}{J}{-}\cdot\PartialEigenvectorConj{s}{J}{-}},
	\\
	\Eigenvector{s}{J}\equiv\ChequerVector{\PartialEigenvector{s}{J}{+}}{\PartialEigenvector{s}{J}{-}},
	\quad
	\EigenvectorConj{s}{J}\cdot\Eigenvector{s}{J}\equiv 1.
	\end{gathered}
\end{equation}
\cref{ClearDivision} indicates that there will be two kinds of scenarios. If~$\PartialEigenvalue{s}{J}$ is non-zero only within one diagonal block, then the massive~$\PVMasses{J}{s}$-particle may be associated with the corresponding parity of that block. Otherwise, parity is not a quantum number of the~$\PVMasses{J}{s}$-particle state.

\paragraph*{Parity-indefinite particles} We first consider the general case where~$\PartialEigenvector{s}{J}{+}$ and~$\PartialEigenvector{s}{J}{-}$ are simultaneously non-vanishing. The no-ghost criterion in~\cref{NoGhostCriterionMoreCompact} takes the component form
\begin{equation}\label{ExplicitProduct}
	\PartialEigenvalue{s}{J}\sum_{X,Y}\sum_{\States{J}{P}{X}{i},\States{J}{\Pp{}}{Y}{j}}
	\Pp{}
	\tensor{\left[\EigenvectorConj{s}{J}\right]}{_{\States{J}{P}{X}{i}}}
	\NewLim{k^2}{\PVSquareMass{J}{s}}
	\left(
	\SourceDownConj{X}{\mu}\PVSPOUpDown{J}{P}{\Pp{}}{X}{Y}{i}{j}{\mu}{\nu}\SourceUp{Y}{\nu}
	\right)
	\tensor{\left[\Eigenvector{s}{J}\right]}{_{\States{J}{\Pp{}}{Y}{j}}}
	\geq 0\quad \forall J,\quad \forall \PVMasses{J}{s}.
\end{equation}
Recalling once again that~$\SourceUp{X}{\mu}$ is arbitrary, we can equivalently parameterise it by arbitrary~$\SourceTildeUp{X}{\mu}$ such that
\begin{subequations}
\begin{align}
	\SourceUp{X}{\mu}
	&\equiv
	\sum_Y
	\sum_{J,P,\Pp{}}
	\sum_{\States{J}{P}{X}{i}\States{J}{\Pp{}}{Y}{j}}
	\PVSPOUpDown{J}{P}{\Pp{}}{X}{Y}{i}{j}{\mu}{\nu}
	\SourceTildeUp{Y}{\nu}
	\tensor{\left[\TotalEigenvector{s}{J}{\Ppp{}}{Z}{k}\right]}{_{\States{J}{P}{X}{i}\States{J}{\Pp{}}{Y}{j}}},
	\label{ExplicitNotation1}
	\\
	\SourceDownConj{X}{\mu}
	&\equiv
	\sum_Y
	\sum_{J,P,\Pp{}}
	\sum_{\States{J}{P}{X}{i}\States{J}{\Pp{}}{Y}{j}}
	\tensor{\left[\TotalEigenvectorConj{s}{J}{\Ppp{}}{Z}{k}\right]}{_{\States{J}{P}{Y}{i}\States{J}{\Pp{}}{X}{j}}}
	\SourceTildeDownConj{Y}{\nu}
	\PVSPOUpDown{J}{P}{\Pp{}}{Y}{X}{i}{j}{\nu}{\mu},
	\label{ExplicitNotation2}
\end{align}
\end{subequations}
where~$\TotalEigenvectorConj{s}{J}{\Ppp{}}{Z}{k}$ is any complete~(i.e. full-rank) row-matrix of orthonormal basis vectors, of which we choose the vector at position label~$\States{J}{\Ppp{}}{Z}{k}$ to be~$\EigenvectorConj{s}{J}$. Note that~\cref{ExplicitNotation2} follows from~\cref{ExplicitNotation1} due to the chequer-Hermitian property in~\cref{ChequerHermitianIndices}. When~\cref{ExplicitNotation1,ExplicitNotation2} are substituted into~\cref{ExplicitProduct} we obtain
\begin{equation}\label{ExpandedCriterion}
	\PartialEigenvalue{s}{J}
	\NewLim{k^2}{\PVSquareMass{J}{s}}
	\left(
	\Ppp{}
	\SourceTildeDownConj{Z}{\mu}
	\PVSPOUpDown{J}{\Ppp{}}{\Ppp{}}{Z}{Z}{k}{k}{\mu}{\nu}\SourceTildeUp{Z}{\nu}\right)
	\geq 0\quad \forall J,\quad \forall \PVMasses{J}{s}.
\end{equation}
Due to the positivity property of the SPOs in~\cref{eq:app_Positivity}, it follows that~\cref{ExpandedCriterion} implies the simple result~$\PartialEigenvalue{s}{J}>0$ for all the states~$\PVMasses{J}{s}$ across all~$J$. The most economical way to determine the eigenvalue is by taking the invariant trace of the residue matrix, so that~\cref{ExpandedCriterion} reduces simply to~\cref{NoGhostCriterionMixed}. As already stated, the mixed-parity scenario may be easily detected by simply inspecting the block-structure of the residue matrix~$\MoorePenroseSJ{s}{J}$. Once this is done, the formula in~\cref{NoGhostCriterionMixed} may be implemented without needing to compute the eigenvector~$\Eigenvector{s}{J}$ or performing any other operations.

\paragraph*{Parity-definite particles} The only other contingency that can arise is one where~$\PartialEigenvector{s}{J}{P}$ is non-vanishing, but~$\PartialEigenvector{s}{J}{\Pp{}}$ is vanishing for~$\Pp{}\neq P$. In this case, the considerations that led to~\cref{NoGhostCriterionMixed} still hold. As before,~$P$ can be determined easily by inspection of the block structure, at which point is it more sensible to denote the various masses using the label~$\Masses{J}{P}{s}$ rather than~$\PVMasses{J}{s}$. The no-ghost criterion in~\cref{NoGhostCriterionMixed} then becomes
\begin{equation}\label{NoGhostCriterion2}
	\NewRes{k^2}{\SquareMass{J}{P}{s}}\left(P\tr \MoorePenroseJP{J}{P}\right)>0\quad \forall J,\ P,\quad \forall \Masses{J}{P}{s}.
\end{equation}
Of course,~\cref{NoGhostCriterion2} is also the no-ghost criterion in cases where~$\MoorePenroseJ{J}$ is already block-diagonal before its pole residues are computed, such as happens without any parity violation: it was obtained already in~\cite{Lin:2018awc,Barker:2024juc}.

\section{Analytic calibration}\label{app:cross-checks}

\paragraph*{Parity-violating massive two-form} The action for a two-form with a parity-odd mass term is given in~\cref{eq:massive_2form_parity-violating}. This can be brought into the following `first-order' form 
\begin{equation}\label{eq:massive_2form_parity-violating_X}
	S = \int\diff^4x \left[-\tensor{\epsilon}{^{\m\n\rho\s}}\PD{_\rho} \TwoFormField{_{\m\n}}\tensor{X}{_\s} +\f{3}{2\a} \tensor{X}{_\m}\tensor{X}{^\m}  +\g \tensor{\epsilon}{^{\m\n\rho\s}}\TwoFormField{_{\m\n}}\TwoFormField{_{\rho\s}}\right] \ ,
\end{equation}
where~$X_\m$ is an auxiliary four-vector field; its equation of motion is
\begin{equation}\label{eq:app_X_pure_2form}
	\tensor{X}{_\s} = -\f{\a}{3}\tensor{\epsilon}{_{\m\n\rho\s}}\PD{^\rho} \TwoFormField{^{\m\n}} \ ,
\end{equation}
and when~\cref{eq:app_X_pure_2form} is plugged into~\cref{eq:massive_2form_parity-violating_X}, we find~\cref{eq:massive_2form_parity-violating}. On the other hand, the equations of motion for the two-form give
\begin{equation}
	\TwoFormField{_{\m\n}} = -\f{1}{\g} \tensor{F}{_{\m\n}} \ ,
\end{equation}
with~$\tensor{F}{_{\m\n}}\equiv \PD{_\m} \tensor{X}{_\n} - \PD{_\n}\tensor{X}{_\m}$. Therefore,~\cref{eq:app_X_pure_2form} yields
\begin{equation}\label{eq:app_X=0}
	\tensor{X}{_\m} = 0 \ , 
\end{equation}
so there are no propagating d.o.f. Equivalently, taking~\cref{eq:massive_2form_parity-violating_X} on-shell gives
\begin{equation}
	S = \int\diff^4x \left[ -\f{2}{\g} \tensor{F}{^{\m\n}}\tensor{\widetilde F}{_{\m\n}} + \f{3}{2\a} \tensor{X}{_\m}\tensor{X}{^\m} \right] \ , 
\end{equation}
with~$\tensor{\widetilde F}{_{\m\n}} \equiv \f 1 2 \tensor{\epsilon}{_{\m\n\rho\s}} \tensor{F}{^{\rho\s}}$. Notice that the vector appears without a kinetic term --- remember,~$\tensor{F}{^{\m\n}}\tensor{\widetilde F}{_{\m\n}}$ is a total derivative --- and so its equations of motion are as in~\cref{eq:app_X=0}.

\paragraph*{Parity-indefinite massive two-form} There is no difficulty in working out analytically the dynamics of a massive two-form with both parity-even and parity-odd mass terms. In terms of~$\tensor{X}{_\m}$,~\cref{eq:gen_two-form} becomes
\begin{equation}\label{eq:massive_2form_full_action_X}
	S =\int\diff^4x \left[-\tensor{\epsilon}{^{\m\n\rho\s}}\PD{_\rho} \TwoFormField{_{\m\n}}\tensor{X}{_\s} +\f{3}{2\a} \tensor{X}{_\m}\tensor{X}{^\m} +\b\TwoFormField{_{\mu\nu}}\TwoFormField{^{\mu\nu}} +\g \tensor{\epsilon}{^{\m\n\rho\s}}\TwoFormField{_{\m\n}}\TwoFormField{_{\rho\s}} \right] \ .
\end{equation}
The equations of motion read
\begin{equation}
	\b\TwoFormField{_{\m\n}} +\g\tensor{\epsilon}{_{\m\n\rho\s}}\TwoFormField{^{\rho\s}} +2\tensor{\widetilde F}{_{\m\n}} = 0 \ ,
\end{equation}
from which we find
\begin{equation}
	\TwoFormField{_{\m\n}} = -\f{1}{\b^2+4\g^2} \left(\g \tensor{F}{_{\m\n}} + \f{\b}{2}\tensor{\widetilde F}{_{\m\n}}\right) \ ,
\end{equation}
and~\cref{eq:massive_2form_full_action_X} becomes (after dropping full derivatives)
\begin{equation}
	S = \int \diff^4x\left[\f{\b}{4(\b^2+4\g^2)} \tensor{F}{_{\m\n}}\tensor{F}{^{\m\n}} + \f{3}{2a} \tensor{X}{_\m}\tensor{X}{^\m}\right] \ .
\end{equation}
In full accordance with the \PSALTer{} result, we see that the theory propagates a healthy massive spin-one field with square mass~$- 3(\b^2+4\g^2)/\a\b$, provided that~$\a >0$ and~$\b <0$.

\paragraph*{One-by-two CSKR theory} Instead of the usual Higgs mechanism, there exists yet another way to induce a mass for a vector field, ``topologically.'' This requires that it couple to a massless two-form, the latter eventually playing the role of the St\"uckelberg field. Let us make this maximally explicit, by considering~\cref{OneByTwoCSKRTheory}. As we did in the pure two-form case, we rewrite the model by using a vector~$\tensor{X}{_\m}$ as
\begin{equation}
	S =\int\diff^4x \left[\alpha\PD{_{[\mu}}\VectorField{_{\nu]}}\PD{^{[\mu}}\VectorField{^{\nu]}}-\tensor{\epsilon}{^{\m\n\rho\s}}\PD{_\rho} \TwoFormField{_{\m\n}}\tensor{X}{_\s} +\f{3}{2\b} \tensor{X}{_\m}\tensor{X}{^\m} +\gamma\tensor{\epsilon}{^{\mu\nu\rho\sigma}}\TwoFormField{_{\mu\nu}}\PD{_{[\rho}}\VectorField{_{\sigma]}}\right] \ .
\end{equation}
The equations of motion for~$\VectorField{_\m}$,~$\tensor{X}{_\m}$ and~$\TwoFormField{_{\m\n}}$ give
\begin{subequations}
\begin{gather}
	2\a\PD{^\n}\PD{_{[\m}}\VectorField{_{\n]}} -\g \tensor{\epsilon}{_{\m\n\rho\s}}\PD{^\s}\TwoFormField{^{\n\rho}}= 0 \ ,\label{eq:eom_B}\\
	\tensor{X}{_\m} = -\f{\b}{3}\tensor{\epsilon}{_{\m\n\rho\s}}\PD{^\s}\TwoFormField{^{\n\rho}} \ ,\label{eq:eom_X}\\
	\PD{_{[\m}}\tensor{X}{_{\n]}} = \g \PD{_{[\m}}\VectorField{_{\n]}} \ ,\label{eq:eom_B2}
\end{gather}
\end{subequations}
respectively; note that~\cref{eq:eom_B2} dictates that 
\begin{equation}
\label{eq:XBchi}
	\tensor{X}{_\m} = \g (\VectorField{_\m} -\PD{_\m} \chi) \ ,
\end{equation}
with~$\chi$ a scalar. Combining appropriately~\cref{eq:eom_B,eq:eom_X,eq:XBchi}, we find
\begin{equation}
	2\a\PD{^\n}\PD{_{[\m}}\VectorField{_{\n]}} +\f{3\g^2}{\b}(\VectorField{_\m} -\PD{_\m}\chi)= 0 \ ,
\end{equation}
which is the equation of motion for a massive spin-one field, with square mass~$-3\g^2/\a\b$, and the consistency conditions on the coefficients are~$\a<0$ and~$\b>0$. 

\paragraph*{Zero-by-three CSKR} For the model of Eq.~\cref{ZeroByThreeCSKRTheory}, 
 the equations of motion for the three-form and the scalar are
\begin{subequations}
\begin{gather}
	\PD{^\s}\left( \b \PD{_{[\m}}\ThreeFormField{_{\n\rho\s]}} +\f \g 2 \tensor{\epsilon}{_{\m\n\rho\s}}\phi \right) = 0 \ , \\ 
	\a \PD{^\alpha}\PD{_{\alpha}} \phi +\f{\g}{2}\tensor{\epsilon}{^{\m\n\rho\s}}\PD{_\s}\ThreeFormField{_{\m\n\rho}} = 0 \ ,
\end{gather}
\end{subequations}
respectively. We see that 
\begin{equation}
	\PD{_{[\m}}\ThreeFormField{_{\n\rho\s]}} = - \f{\g}{2\b} \tensor{\epsilon}{_{\m\n\rho\s}}\phi \ ,
\end{equation}
and 
\begin{equation}
	\a\PD{^\alpha}\PD{_\alpha}\phi -\f{6\g^2}{\b}\phi = 0 \ ,
\end{equation}
meaning that the theory propagates a spin-zero particle with square mass~$-6\g^2/\a\b$ which is healthy as long as~$\a>0$ and~$\b<0$. Once again, the explicit computation is in full agreement with the \PSALTer{} result.

\paragraph*{Parity-indefinite Einstein--Cartan gravity} When it comes to studying such models `by hand', it is useful to transition from the gauge picture with variables~$\tensor{e}{^\a_{\SmashAcute\m}}$ and~$\tensor{\omega}{^{\a\b}_{\SmashAcute\m}}$, to the affine one with variables the metric~$\tensor{g}{_{\SmashAcute\m\SmashAcute\n}}$ and (affine) connection~$\Con{^{\SmashAcute\m}_{\SmashAcute\n\SmashAcute\rho}}$; the latter are related to the former as
\begin{equation}
	\tensor{g}{_{\SmashAcute\m\SmashAcute\n}} = \tensor{e}{^\a_{\SmashAcute\m}}\tensor{e}{^\b_{\SmashAcute\n}}\tensor{\eta}{_{\a\b}} \ ,\quad \Con{^{\SmashAcute\m}_{\SmashAcute\n\SmashAcute\rho}} = \tensor{e}{_\a^{\SmashAcute\m}}\left(\PD{_{\SmashAcute\n}} \tensor{e}{^\a_{\SmashAcute\rho}}+\tensor{\omega}{^\a_{\SmashAcute\m\b}}\tensor{e}{^\b_{\SmashAcute\rho}} \right) \ . 
\end{equation}
From the above it can be shown that the affine torsion and curvature tensors read
\begin{equation}
	\ECT{^{\SmashAcute\m}_{\SmashAcute\n\SmashAcute\rho}} = \Con{^{\SmashAcute\m}_{\SmashAcute\n\SmashAcute\rho}} - \Con{^{\SmashAcute\m}_{\SmashAcute\rho\SmashAcute\n}} \ ,\quad \ECR{^{\SmashAcute\rho}_{\SmashAcute\s\SmashAcute\m\SmashAcute\n}} = \PD{_{\SmashAcute\m}}\Con{^{\SmashAcute\rho}_{\SmashAcute\n\SmashAcute\s}}-\PD{_{\SmashAcute\n}}\Con{^{\SmashAcute\rho}_{\SmashAcute\m\SmashAcute\s}} + \Con{^{\SmashAcute\rho}_{\SmashAcute\m\SmashAcute\lambda}}\Con{^{\SmashAcute\lambda}_{\SmashAcute\n\SmashAcute\s}} - \Con{^{\SmashAcute\rho}_{\SmashAcute\n\SmashAcute\lambda}}\Con{^{\SmashAcute\lambda}_{\SmashAcute\m\SmashAcute\s}} \ .
\end{equation}
For what follows, we also introduce the vector~$\tensor{v}{_{\SmashAcute\m}}$, pseudovector~$\tensor{a}{_{\SmashAcute\m}}$ and reduced torsion tensor~$\tensor{\tau}{_{\SmashAcute\m\SmashAcute\n\SmashAcute\rho}}$ defined as~\cite{Karananas:2021zkl,Karananas:2024xja}
\begin{equation}
\label{eq:app_torsion_irreps}
	\begin{gathered}
	\tensor{v}{^{\SmashAcute\m}} = \ECT{^{\SmashAcute\n}_{\SmashAcute\m\SmashAcute\n}} \ ,\quad \tensor{a}{^{\SmashAcute\m}} = \tensor{E}{^{\SmashAcute\m\SmashAcute\n\SmashAcute\rho\SmashAcute\s}}\ECT{_{\SmashAcute\n\SmashAcute\rho\SmashAcute\s}} \ ,\\ 
		\tensor{\tau}{_{\SmashAcute\m\SmashAcute\n\SmashAcute\rho}} =\f 2 3 \ECT{_{\SmashAcute\m\SmashAcute\n\SmashAcute\rho}} +\f 1 3 \left(\tensor{g}{_{\SmashAcute\m\SmashAcute\n}}\tensor{v}{_{\SmashAcute\rho}}-\tensor{g}{_{\SmashAcute\rho\SmashAcute\m}}\tensor{v}{_{\SmashAcute\n}}\right) - \f 1 3 \left(\ECT{_{\SmashAcute\n\SmashAcute\rho\SmashAcute\m}}-\ECT{_{\SmashAcute\rho\SmashAcute\n\SmashAcute\m}}\right) \ , 
	\end{gathered}
\end{equation}
with~$\tensor{\tau}{^{\SmashAcute\n}_{\SmashAcute\n\SmashAcute\m}}=\tensor{\tau}{^{\SmashAcute\n}_{\SmashAcute\m\SmashAcute\n}}=\tensor{E}{^{\SmashAcute\m\SmashAcute\n\SmashAcute\rho\SmashAcute\s}}\tensor{\tau}{_{\SmashAcute\n\SmashAcute\rho\SmashAcute\s}}=0$, and~$\tensor{E}{^{\SmashAcute\m\SmashAcute\n\SmashAcute\rho\SmashAcute\s}}=\tensor{\epsilon}{^{\SmashAcute\m\SmashAcute\n\SmashAcute\rho\SmashAcute\s}}/\sqrt{g}$, and~$g=-\det(\tensor{g}{_{\SmashAcute\m\SmashAcute\n}})$. The scalar~$\ECR{}$ and pseudoscalar~$\Holst{}$ curvatures in the affine basis read
\begin{equation}
	\ECR{} = \tensor{g}{^{\SmashAcute\s\SmashAcute\n}}\tensor*{\delta}{^{\SmashAcute\m}_{\SmashAcute\rho}} \ECR{^{\SmashAcute\rho}_{\SmashAcute\s\SmashAcute\m\SmashAcute\n}} \ ,\quad \Holst{}= \tensor{E}{^{\SmashAcute\rho\SmashAcute\s\SmashAcute\m\SmashAcute\n}}\ECR{_{\SmashAcute\rho\SmashAcute\s\SmashAcute\m\SmashAcute\n}} \ . 
\end{equation}
The most general parity-indefinite action that comprises all invariants which are at most quadratic in torsion and the scalar and pseudoscalar curvatures was given in the main text, see~\cref{eq:EC_parity_indefinite_general}. In terms of~\cref{eq:app_torsion_irreps}, it reads
\begin{align}
	S = \int \diff^4x \sqrt{g} \Bigg[c_1\ECR{} &+c_2 \Holst{} + c_3 \ECR{}^2 +c_4\ECR{}\Holst{} + c_5\Holst{}^2 \nonumber\\
	&+ \f{C_{vv}}{3} \tensor{v}{_{\SmashAcute\m}} \tensor{v}{^{\SmashAcute\m}} -\f{C_{aa}}{24}\tensor{a}{_{\SmashAcute\m}} \tensor{a}{^{\SmashAcute\m}} +\f{C_{\tau\tau}}{2}\tensor{\tau}{_{\SmashAcute\m\SmashAcute\n\SmashAcute\rho}}\tensor{\tau}{^{\SmashAcute\m\SmashAcute\n\SmashAcute\rho}}\nonumber\\
	&\qquad\quad+\f{2C_{va}}{3}\tensor{a}{_{\SmashAcute\m}} \tensor{v}{^{\SmashAcute\m}}  +\f{\tilde C_{\tau\tau}}{2}\tensor{E}{^{\SmashAcute\m\SmashAcute\n\SmashAcute\rho\SmashAcute\s}}\tensor{\tau}{_{\SmashAcute\lambda\SmashAcute\m\SmashAcute\n}}\tensor{\tau}{^{\SmashAcute\lambda}_{\SmashAcute\rho\SmashAcute\s}}\Bigg] \ ,
\end{align}
where
\begin{align}
\label{eq:app_notation_conv-1}
&C_{vv} =2c_6-c_7+3c_8  \ ,\quad C_{aa} =  4(c_6+c_7) \ ,\quad C_{va} =2c_9-c_{10} \ ,
\end{align}
and~$C_{\tau\tau}$ and~$\tilde C_{\tau\tau}$ depend on~$c_6,c_7$ and~$c_9,c_{10}$, respectively --- the explicit relations can be easily worked out but are completely irrelevant for the following. We can  get rid of the quadratic-in-curvature terms by introducing two auxiliary fields~$\chi$ and~$\phi$
\begin{align}
	S = \int \diff^4x\sqrt{g}\Bigg[(c_1+\chi) \ECR{} &+\left(c_2+q\chi+\phi\right) \Holst{} - \f{ \chi^2}{4c_3} - \f{c_3\phi^2 }{4c_3c_5-c_4^2}\nonumber\\
	&\quad+\f{C_{vv}}{3} \tensor{v}{_{\SmashAcute\m}}\tensor{v}{^{\SmashAcute\m}} -\f{C_{aa}}{24}\tensor{a}{_{\SmashAcute\m}} \tensor{a}{^{\SmashAcute\m}} -\f{C_{\tau\tau}}{2}\tensor{\tau}{_{\SmashAcute\m\SmashAcute\n\SmashAcute\rho}}\tensor{\tau}{^{\SmashAcute\m\SmashAcute\n\SmashAcute\rho}}\nonumber\\
	&\qquad\qquad+\f{2C_{va}}{3}\tensor{a}{_{\SmashAcute\m}}\tensor{v}{^{\SmashAcute\m}} +\f{\tilde C_{\tau\tau}}{2}\tensor{E}{^{\SmashAcute\m\SmashAcute\n\SmashAcute\rho\SmashAcute\s}}\tensor{\tau}{_{\SmashAcute\lambda\SmashAcute\m\SmashAcute\n}}\tensor{\tau}{^\lambda_{\SmashAcute\rho\SmashAcute\s}}
\Bigg] \ ,
\end{align}
with
\be
q=\f{c_4}{2c_3} \ . 
\ee
The next step consists in plugging into the above the standard decomposition of the scalar~$\ECR{}$ and pseudoscalar~$\Holst{}$ curvatures in Riemannian (denoted with a~`$~\mathring{}~$' on top) and post-Riemannian contributions
\begin{equation}
	\begin{aligned}
		\ECR{} & = \tensor{\mathring R}{} +2 \tensor{\mathring \nabla}{_{\SmashAcute\m}} \tensor{v}{^{\SmashAcute\m}} -\f 2 3 \tensor{v}{_{\SmashAcute\m}} \tensor{v}{^{\SmashAcute\m}} + \f{1}{24}\tensor{a}{_{\SmashAcute\m}} \tensor{a}{^{\SmashAcute\m}} +\f 1 2 \tensor{\tau}{_{\SmashAcute\m\SmashAcute\n\SmashAcute\rho}}\tensor{\tau}{^{\SmashAcute\m\SmashAcute\n\SmashAcute\rho}} \ ,\\ 
		\Holst{} & = -\tensor{\mathring \nabla}{_{\SmashAcute\m}} \tensor{a}{^{\SmashAcute\m}} +\f 2 3 \tensor{a}{_{\SmashAcute\m}} \tensor{v}{^{\SmashAcute\m}} +\f 1 2 \tensor{E}{^{\SmashAcute\m\SmashAcute\n\SmashAcute\rho\SmashAcute\s}}\tensor{\tau}{_{\SmashAcute\lambda\SmashAcute\m\SmashAcute\n}}\tensor{\tau}{^{\SmashAcute\lambda}_{\SmashAcute\rho\SmashAcute\s}} \ ,
	\end{aligned}
\end{equation}
to obtain
\begin{align}
\label{eq:S_eff}
	S =  \int \diff^4x \sqrt{g}\Bigg[&(c_1+\chi)\tensor{\mathring R}{} -\f{\chi^2}{4c_3} -\f{c_3\phi^2 }{4c_3c_5-c_4^2}- 2v^{\SmashAcute\m} \p_{\SmashAcute\m} \chi 
	\nonumber\\ 
    &\quad+a^{\SmashAcute\m}\left(q\p_{\SmashAcute\m} \chi +\p_{\SmashAcute\m} \phi\right)  -\f{2}{3}\left(\f{\Upsilon_1}{2}+\chi\right)v_{\SmashAcute\m} v^{\SmashAcute\m} 
	\nonumber\\
    &\qquad+\f{1}{24}\left(\Upsilon_4+\chi\right)\tensor{a}{_{\SmashAcute\m}} \tensor{a}{^{\SmashAcute\m}} +\f{c_{\tau\tau}+\chi}{2}\tensor{\tau}{_{\SmashAcute\m\SmashAcute\n\SmashAcute\rho}}\tensor{\tau}{^{\SmashAcute\m\SmashAcute\n\SmashAcute\rho}}
	\nonumber\\
    &\quad\qquad-\f{2}{3}\left(\Upsilon_2 -q\chi-\phi\right)\tensor{a}{_{\SmashAcute\m}} \tensor{v}{^{\SmashAcute\m}} + \f{\tilde c_{\tau\tau}+q\chi+\phi}{2}\tensor{E}{^{\SmashAcute\m\SmashAcute\n\SmashAcute\rho\SmashAcute\s}}\tensor{\tau}{_{\SmashAcute\lambda\SmashAcute\m\SmashAcute\n}}\tensor{\tau}{^{\SmashAcute\lambda}_{\SmashAcute\rho\SmashAcute\s}}\Bigg] \ ,
\end{align}
where 
\be
\Upsilon_1 = 2c_1-2c_6+c_7+3c_8\ ,\quad \Upsilon_2 = c_{10}-2c_9-c_2 \ ,\quad \Upsilon_4 = c_1-4(c_6+c_7) \ ,
\ee
have already appeared in the main text (see~\cref{ParticleSpectrographScalarParityViolatingPGT}), and we also introduced 
\be
c_{\tau\tau}=C_{\tau\tau}+ c_1 \ ,\quad \tilde c_{\tau\tau} = \tilde C_{\tau\tau}+ c_2 \ .
\ee
One notices that torsion appears algebraically in the action~\cref{eq:S_eff} and can thus be integrated out via the corresponding equations of motion. Variation of the above wrt to~$\tau$ dictates that the reduced tensor vanish on-shell
\begin{equation}
\label{eq:eom_tau}
	\tensor{\tau}{_{\SmashAcute\m\SmashAcute\n\SmashAcute\rho}} = 0 \ ,
\end{equation}
while for the vector and pseudovector, we find
\begin{subequations}
\begin{align}
\label{eq:eom_v}
	&\tensor{v}{_{\SmashAcute\m}} =- 3 \f{\left(\Upsilon_4-4q\Upsilon_2+(1+4q^2)\chi+4q\phi\right)\p_{\SmashAcute\m}\chi -4(\Upsilon_2-q\chi-\phi)\PD{_{\SmashAcute\m}}\phi}{D} \ ,\\
\label{eq:eom_a}
	&\tensor{a}{_{\SmashAcute\m}} = -12\f{\left(q\Upsilon_1+2\Upsilon_2-2\phi\right)\p_{\SmashAcute\m}\chi+\left(\Upsilon_1+2\chi\right)\PD{_{\SmashAcute\m}}\phi}{D} \ ,
\end{align}
\end{subequations}
and to keep the expressions short we introduced
\begin{equation}
\label{eq:demoninator}
D=  \Upsilon_1 \Upsilon_4 + 8\Upsilon_2^2+(\Upsilon_1-16q\Upsilon_2+2\Upsilon_4)\chi +2(1+4q^2)\chi^2-16\phi(\Upsilon_2-q\chi)+8\phi^2 \ .
\end{equation}
Plugging~\crefrange{eq:eom_tau}{eq:eom_a} into~\cref{eq:S_eff}, we obtain 
\begin{equation}
\label{eq:S_eff_2}
	S =\int\diff^4x \sqrt{g} \left[(c_1+\chi)\tensor{\mathring R}{} + \PD{_{\SmashAcute\m}} \upphi^\text{T} \cdot\mathsf{g}\cdot\PD{^{\SmashAcute\m}}\upphi - \f{\chi^2}{4c_3} - \f{c_3\phi^2 }{4c_3c_5-c_4^2} \right] \ ,
\end{equation}
where~$\upphi^\text{T}=\left[\chi,\phi\right]$, and 
\begin{equation}
	\Kinetic{} = \f{3}{D} \NormalMatrix{\Kinetic{_{\chi\chi}} }{ \Kinetic{_{\chi \phi}} }{ \Kinetic{_{\chi\phi}} }{ \Kinetic{_{\phi\phi}} },
\end{equation}
is the metric of the kinetic manifold whose components are  
\begin{equation}
	\begin{aligned}
	&\Kinetic{_{\chi\chi}} = \Upsilon_4-2q(q\Upsilon_1+4\Upsilon_2)+(1+4q^2)\chi+8 q\phi,\quad \Kinetic{_{\chi\phi}} = -2\left(q\Upsilon_1+2\Upsilon_2-2\phi\right),
	\\
	&\Kinetic{_{\phi\phi}} = -2(\Upsilon_1+2\chi).
	\end{aligned}
\end{equation}
We can eliminate the nonminimal coupling of~$\chi$ to gravity via a Weyl rescaling of~$\tensor{g}{_{\m\n}}$
\begin{equation}
\label{eq:Weyl_transformation}
	\tensor{g}{_{\m\n}} \mapsto \Omega^{-2} \tensor{g}{_{\m\n}} \ ,\quad \Omega^2 = \f{c_1+\chi}{c_1} \ .
\end{equation}
This results into the following Einstein-frame action
\begin{equation}
\label{eq:actionPGT_weyl_transformed}
	S =  \int \diff^4x\sqrt{g}\left[ c_1 \tensor{\mathring R}{} + \PD{_{\SmashAcute\m}} \upphi^\text{T}\cdot\NewKinetic{}\cdot\PD{^{\SmashAcute\m}} \upphi -V \right] \ ,
\end{equation}
where the Weyl-transformed field-space metric~$\NewKinetic{}$ is given by
\begin{equation}\label{eq:field_space_metric_EF}
	\NewKinetic{} = \f{3}{D\Omega^2} \NormalMatrix{ \Kinetic{_{\chi\chi}} -\f{D}{2c_1\Omega^2} }{ \Kinetic{_{\chi \phi}} }{ \Kinetic{_{\chi\phi}} }{ \Kinetic{_{\phi\phi}} },
\end{equation}
and the potential reads
\be
V = \f{c_1^2}{(c_1+\chi)^2}\left(\f{\chi^2}{4c_3} +\f{c_3\phi^2 }{4c_3c_5-c_4^2}\right) \ .
\ee
Since gravity is canonical, we can immediately conclude that 
\be
c_1<0 \ ,
\ee
as expected. Turning now to the scalar sector, we first consider the potential, that is extremized for 
\be
\label{eq:pot_minimum}
\chi=\phi = 0 \ .
\ee
For its Hessian to be positive-definite when evaluated on the extremum, we find
\be
\left(c_3>0\right) \wedge \left( 4c_3 c_5 -c_4^2 > 0\right) \ ,
\ee
from which it follows that~$c_5>0$. These are exactly the no-ghost conditions we derived with the SPOs in the main text. Finally, we consider the kinetic terms of the fields.\footnote{Although not needed for the considerations here, note that there is no difficulty in diagonalizing the kinetic terms of the scalars and also making one of them ($\chi$) canonical; this is achieved by introducing 
\begin{equation}
\Phi = \log\left[\f{c_4 \Upsilon_1+4c_3(\Upsilon_2-\phi)}{2c_3(\Upsilon_1+2\chi)}\right]\ ,\quad X = 2\sqrt{3c_1}\tan^{-1}\left[\sqrt{\f{\Upsilon_1+2\chi}{2c_1-\Upsilon_1}}\right] \ .
\end{equation}
} We evaluate the field-space metric~\cref{eq:field_space_metric_EF} on~\cref{eq:pot_minimum}, and then require that its determinant and trace be positive. This reproduces the constraints of~\cref{eq:no_tachyon_cond_1,eq:no_tachyon_cond_2}. 

\section{Spectral calibration}\label{app:PGT_comparison}

\paragraph*{General EC/Poincar\'e gravity} In this appendix we perform the most sophisticated possible calibration for the \PSALTer{} implementation. In this procedure we compare the software output with the results of~\cite{Karananas:2014pxa}, where the most general parity-violating theory up to quadratic order in curvature and torsion (see also~\cite{Diakonov:2011fs,Baekler:2011jt}) was already studied in the SPO formalism. In~\cite{Karananas:2014pxa}, however, different Lagrangian coupling coefficients were used relative to those introduced in~\cref{sec:EC_gravity}. Specifically, the~$\Holst{}$ and the~$\ECR{}^2$ operators were not included, the first being related to~$\epsilon^{\m\n\rho\s} \ECT{_{\m\n\lambda}}\ECT{_{\rho\s}^{\lambda}}$ (see also~\cref{foot:holst}), and the latter to the squares of the Ricci and Riemann tensors via the Gauss--Bonnet identity. To facilitate the comparison, we now utilize the parametrization of~\cite{Karananas:2014pxa}. Accordingly, the action in~\cref{eq:EC_parity_indefinite_general} is extended and reparametrized as
\begin{align}
	S=\int & \mathrm{d}^4x \,e\bigg[
		-\lambda\mathscr{R}
	+\frac{1}{6}\left(2r_1+r_2\right)\tensor{\mathscr{R}}{_{\a\b\g\delta}}\tensor{\mathscr{R}}{^{\a\b\g\delta}}
	+\frac{2}{3}\left(r_1-r_2\right)\tensor{\mathscr{R}}{_{\a\b\g\delta}}\tensor{\mathscr{R}}{^{\a\g\b\delta}}
	\nonumber\\
	&
	+\frac{1}{6}\left(2r_1+r_2-6r_3\right)\tensor{\mathscr{R}}{_{\a\b\g\delta}}\tensor{\mathscr{R}}{^{\g\delta\a\b}}
	+\left(r_4+r_5\right)\tensor{\mathscr{R}}{_{\a\b}}\tensor{\mathscr{R}}{^{\a\b}}
	+\left(r_4-r_5\right)\tensor{\mathscr{R}}{_{\a\b}}\tensor{\mathscr{R}}{^{\b\a}}
	\nonumber\\
	&+\frac{1}{6}\left(r_6-r_8\right)\mathscr{R}\tensor{\tilde{\mathscr{R}}}{}
	-\frac{1}{8}\left(r_7+r_8\right)\tensor{\epsilon}{^{\a\b\m\n}}\tensor{\mathscr{R}}{_{\a\b\rho\s}}\tensor{\mathscr{R}}{_{\m\n}^{\s\rho}}
	+\frac{1}{4}\left(r_7-r_8\right)\tensor{\epsilon}{^{\a\b\m\n}}\tensor{\mathscr{R}}{_{\a\b\rho\s }}\tensor{\mathscr{R}}{^{\rho\s}_{\m\n}}
	\nonumber\\
	&
		+\frac{1}{12}\left(4t_1+t_2+3\lambda\right)\tensor{\mathscr{T}}{_{\a\b\g}}\tensor{\mathscr{T}}{^{\a\b\g}}-\frac{1}{3}\left(t_1-2t_3+3\lambda\right)\tensor{\mathscr{T}}{_{\a}}\tensor{\mathscr{T}}{^{\a}}
	\nonumber\\
	&-\frac{1}{6}\left(2t_1-t_2+3\lambda\right)\tensor{\mathscr{T}}{_{\a\b\g}}\tensor{\mathscr{T}}{^{\b\g\a}}-\frac{1}{12}\left(t_4+4t_5\right)\tensor{\epsilon}{^{\a\b\m\n}}\tensor{\mathscr{T}}{_{\rho\a\b}}\tensor{\mathscr{T}}{^{\rho}_{\m\n}}\nonumber\\
	&+\frac{1}{3}\left(t_4-2t_5\right)\tensor{\epsilon}{^{\a\b\m\n}}\tensor{\mathscr{T}}{_{\a\b\rho}}\tensor{\mathscr{T}}{_{\m\n}^{\rho}}
	\bigg] \ , \label{PGTVersion}
\end{align}
where~$\lambda, r_1,\ldots,r_8$ and~$t_1,\ldots,t_5$ are constants \lstinline!kLambda!, \lstinline!kR1!, through to \lstinline!kT5!.

\paragraph*{Results of the calibration} The quadratic part of~\cref{PGTVersion} is inevitably a very long expression:
\begin{lstlisting}
In[#]:= ParticleSpectrum[((kT1 + kT2) * SpinConnection[-a, -b, -c] * SpinConnection[a, b, c])/3 + ((kT4 - 2 * kT5) * epsilonG[-b, -c, -d, -i] * SpinConnection[-a, d, i] * SpinConnection[a, b, c])/3 + ((kT1 - 2 * kT2) * SpinConnection[a, b, c] * SpinConnection[-b, -a, -c])/3 - (2 * (kT4 - 2 * kT5) * (*omitted 4264 characters for brevity*) b]] * CD[i][TetradPerturbation[d, -b]])/3 - (2 * (kR6 - kR8) * epsilonG[-c, -d, -i, -j] * CD[-b][SpinConnection[a, -a, b]] * CD[j][SpinConnection[c, d, i]])/3, TheoryName -> "GeneralParityViolatingPGT", MaxLaurentDepth -> 1, AspectRatio -> Portrait, ShowPropagator -> False];
\end{lstlisting}
The output is shown in~\cref{ParticleSpectrographGeneralParityViolatingPGT}. Apart from polynomial factors in~$k^2$ whose couplings are numerical,\footnote{These factors are artefacts of Moore--Penrose gauge fixing, and do not imply the presence of massive poles.} each determinant is quadratic in~$k^2$ with couplings that depend on the coefficients in~\cref{PGTVersion}. The roots of these quadratic equations are the masses of the two non-graviton particles in each spin sector.  Importantly, the mass expressions are identical to the ones presented in~\cite{Karananas:2014pxa}, and so are the no-tachyon conditions~\cite{Blagojevic:2018dpz} that follow by requiring that these be (real and) positive. This is the first non-trivial sanity-check that the code passes successfully.  It should be noted that there are in fact two differences, attributed to choices of convention, between the matrix elements in~\cref{ParticleSpectrographGeneralParityViolatingPGT} and those in~\cite{Karananas:2014pxa}; nevertheless, neither affects the physics and the coefficient matrices are perfectly consistent with each other. The first difference is that all the off-diagonal~(parity-violating) blocks differ by a factor of~$i$. This is because~\PSALTer{} assumes the convention whereby the parity-violating spin-projection operators are symmetric and imaginary, whereas~\cite{Karananas:2014pxa} takes the same operators to be real, but skew-symmetric. As we showed in~\cref{SpinProjection}, the requirement for physicality is actually that these blocks have a \emph{skew-Hermitian} structure, and so either of these conventions is valid. The second difference is that the degeneracy of the spin-one matrix in~\cref{ParticleSpectrographGeneralParityViolatingPGT} is visible in the form of \emph{two} repeated rows and \emph{one} row of zeros~(and likewise for columns). In~\cite{Karananas:2014pxa}, on the other hand, there are \emph{three} repeated rows and columns~(similar matrices appear in~\cite{Lin:2019ugq}. The actual difference in this case is due to the direct decomposition of the negative parity spin-one modes in~\cref{FieldKinematicsTetradPerturbation}. These modes are linear combinations of the modes used in~\cite{Karananas:2014pxa,Lin:2019ugq}, which are obtained after first breaking the tetrad perturbation into symmetric and antisymmetric parts. The second non-trivial cross-check for the validity of the results obtained by \PSALTer{} is provided by deriving the no-ghost conditions --- with our simplified method of~\cref{MassiveSpectrum},  this can be done almost trivially by inspection of the matrices in~Fig.~\ref{ParticleSpectrographGeneralParityViolatingPGT}. We obtain
\begin{align}
	J&=0:\quad r_2< 0,~2r_2(r_1-r_3+2r_4) +r_6^2 < 0; \\
	J&=1:\quad r_1+r_4+r_5<0,~(r_1+r_4+r_5)(2r_3+r_5) +r_7^2<0; \\ 
	J&=2:\quad r_1< 0,~r_1 (2r_1-2r_3+r_4) + r_8^2<0; 
\end{align}
which are identical to the findings  of~\cite{Karananas:2014pxa,Blagojevic:2018dpz}. As well known, the above constraints for the spin-one and spin-two sectors are contradicting each other~\cite{Blagojevic:2018dpz,Karananas:2016ltn}.

\begin{figure}[htbp]
	\includegraphics[width=\linewidth]{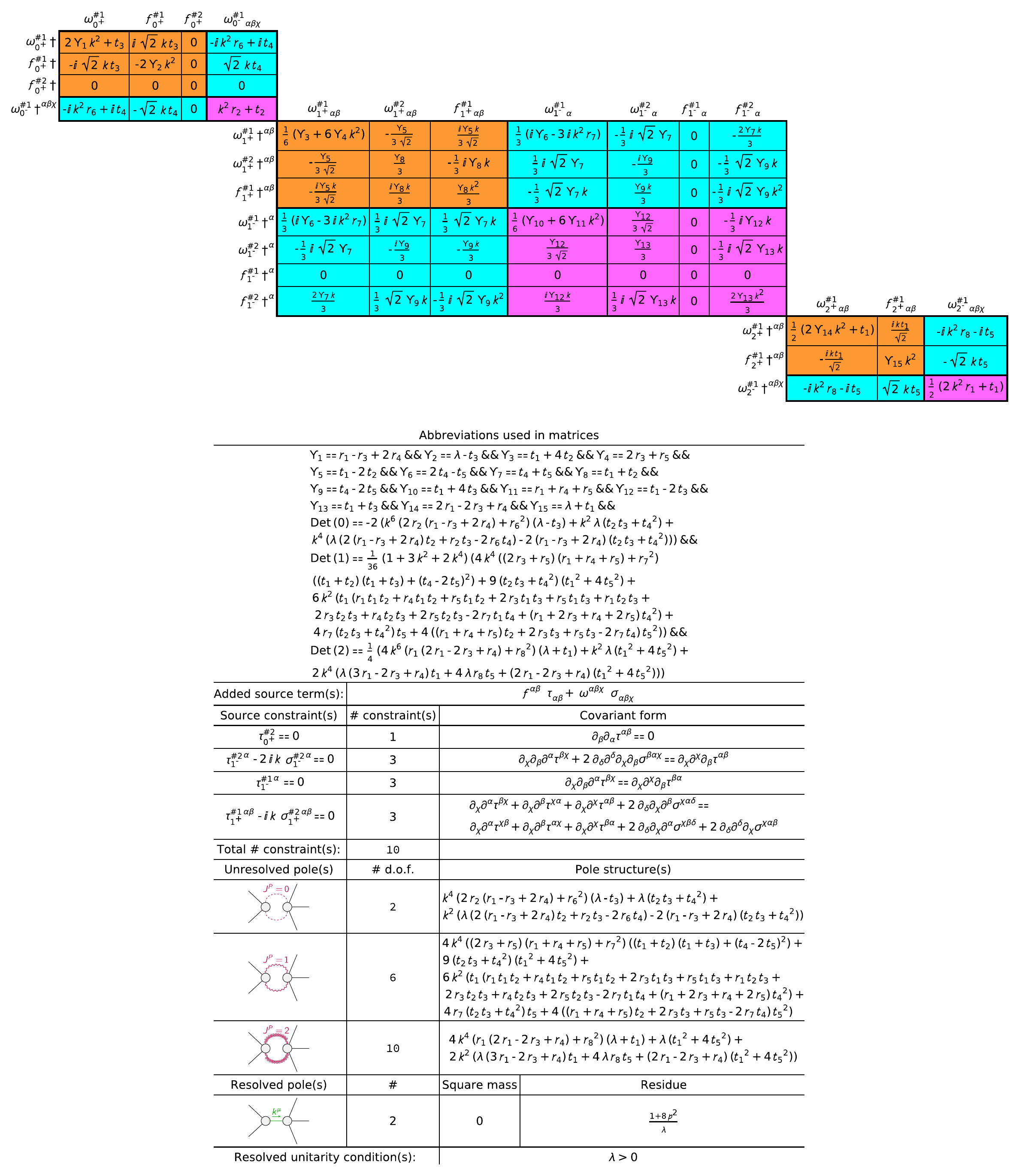}
	\caption{Partial particle spectrum of the most general parity-violating PGT. Due to the square masses of the new species being irrational functions of the Lagrangian coupling coefficients, \PSALTer{} does not yet attempt to evaluate the massive no-ghost criteria. The elements of the wave operator matrices are fully consistent with those in~\cite{Karananas:2014pxa}. The determinants are quadratic in~$k^2$, and the (generally massive) poles defined by their roots are also consistent with the mass formulae in~\cite{Karananas:2014pxa}. We note the appearance of ten gauge generators: precisely this number is expected due to the Poincar\'e gauge symmetry. All quantities are defined in~\cref{FieldKinematicsTetradPerturbation,FieldKinematicsSpinConnection}.}
\label{ParticleSpectrographGeneralParityViolatingPGT}
\end{figure}

\section{Sources and installation}\label{Install}

\paragraph*{Obtaining the package} In this appendix we provide the updated structure of the \PSALTer{} source files. As before, the \PSALTer{} package should only be installed after the \xAct{} suite of packages has been installed. For information about \xAct{}, see \href{http://www.xact.es/}{\texttt{xact.es}}. The actual \PSALTer{} package is available at the \GitHub{} repository \href{https://github.com/wevbarker/PSALTer}{\texttt{github.com/wevbarker/PSALTer}}, along with installation instructions for various operating systems, including \Windows{} and \Mac{}. Here we demonstrate a \Linux{} installation.\footnote{The syntax highlighting for \Bash{} differs from that used for the \WolframLanguage{} in~\cref{SymbolicImplementation}.} One can use \Bash{} to download \PSALTer{} into the home directory as follows:
\begin{lstlisting}[language=Special]
[user@system ~]$\dollar$ git clone https://github.com/wevbarker/PSALTer
\end{lstlisting}
\paragraph*{Structure of the package} The package contains~$\SI{1e4}{}$ source lines of code distributed in a modular design across plaintext \WolframLanguage{} files with \texttt{.m} or \texttt{.wl} extensions (there are also some graphics files). There are no binaries, and the software does not need to be compiled. The latest directory tree, which has been heavily restructured as compared to the initial release in~\cite{Barker:2024juc}, is as follows:
\begin{lstlisting}[breaklines=true,style=ascii-tree,language=Special]
[user@system ~]$\dollar$ tree PSALTer
PSALTer
├── LICENSE.md
├── README.md
└── xAct
    └── PSALTer
        ├── Kernel
        │   └── init.wl
        ├── Logos
        │   ├── ASCIILogo.txt
        │   ├── convert_logos.sh
        │   ├── FieldKinematics.pdf
        │   ├── FieldKinematics.png
        │   ├── GitHubLogo.pdf
        │   ├── GitHubLogo.png
        │   ├── GitHubLogo.svgz
        │   ├── GitLabLogo.pdf
        │   ├── GitLabLogo.png
        │   ├── GitLabLogo.svgz
        │   ├── ParticleSpectrographMassiveGravity.m
        │   ├── ParticleSpectrograph.pdf
        │   └── ParticleSpectrograph.png
        ├── PSALTer.m
        └── Sources
            ├── DefField
            │   ├── AppendToField.m
            │   ├── CombineRules
            │   │   ├── DefAllComponentValues.m
            │   │   ├── DefSummary.m
            │   │   ├── ValidateSO3Irreps
            │   │   │   ├── ValidateInverseField.m
            │   │   │   ├── ValidateInverseMode.m
            │   │   │   ├── ValidateSpatial.m
            │   │   │   ├── ValidateSymmetryField.m
            │   │   │   ├── ValidateSymmetryMode.m
            │   │   │   └── ValidateTraceless.m
            │   │   └── ValidateSO3Irreps.m
            │   ├── CombineRules.m
            │   ├── DefFiducialField.m
            │   ├── DefSO3Irrep
            │   │   ├── DefSymbol.m
            │   │   ├── MakeAutomaticallyNotAntisymmetric
            │   │   │   └── RemoveContraction.m
            │   │   ├── MakeAutomaticallyNotAntisymmetric.m
            │   │   ├── MakeAutomaticallyTraceless.m
            │   │   ├── MakeUniquePartialDual.m
            │   │   └── MakeUniqueQuadratic.m
            │   ├── DefSO3Irrep.m
            │   ├── PreComputeComponents
            │   │   ├── AllIndexConfigurations.m
            │   │   └── AllocateTensorValues.m
            │   ├── PreComputeComponents.m
            │   ├── RegisterFieldRank0.m
            │   ├── RegisterFieldRank1.m
            │   ├── RegisterFieldRank2Antisymmetric.m
            │   ├── RegisterFieldRank2.m
            │   ├── RegisterFieldRank2Symmetric.m
            │   ├── RegisterFieldRank3Antisymmetric12.m
            │   ├── RegisterFieldRank3Antisymmetric13.m
            │   ├── RegisterFieldRank3Antisymmetric23.m
            │   ├── RegisterFieldRank3.m
            │   ├── RegisterFieldRank3Symmetric12.m
            │   ├── RegisterFieldRank3Symmetric13.m
            │   ├── RegisterFieldRank3Symmetric23.m
            │   ├── RegisterFieldRank3TotallyAntisymmetric.m
            │   ├── RegisterFieldRank3TotallySymmetric.m
            │   ├── SummariseField
            │   │   ├── DecompositionTable.m
            │   │   ├── ExpansionTable.m
            │   │   └── FieldMosaic.m
            │   └── SummariseField.m
            ├── DefField.m
            ├── ParticleSpectrum
            │   ├── CombineAssociations
            │   │   ├── CacheContexts.m
            │   │   ├── DefPlaceholderSpins.m
            │   │   ├── GenerateAnsatz
            │   │   │   ├── CatalogueInvariant
            │   │   │   │   └── IsNegativeParitySpinTwo.m
            │   │   │   └── CatalogueInvariant.m
            │   │   ├── GenerateAnsatz.m
            │   │   └── NormaliseRescalings.m
            │   ├── CombineAssociations.m
            │   ├── ConstructLinearAction.m
            │   ├── ConstructMassiveAnalysis
            │   │   ├── MassiveAnalysisOfSector
            │   │   │   ├── IsolatePoles
            │   │   │   │   ├── IrrationalQ.m
            │   │   │   │   ├── PartitionDeterminant
            │   │   │   │   │   └── GaugeArtifactQ.m
            │   │   │   │   ├── PartitionDeterminant.m
            │   │   │   │   └── StripPoly.m
            │   │   │   ├── IsolatePoles.m
            │   │   │   └── PoleToSquareMass.m
            │   │   ├── MassiveAnalysisOfSector.m
            │   │   ├── MassiveGhost.m
            │   │   ├── SimplifyMasses.m
            │   │   ├── SquareMassQ.m
            │   │   └── UnresolvedPoleQ.m
            │   ├── ConstructMassiveAnalysis.m
            │   ├── ConstructMasslessAnalysis
            │   │   ├── ConstructLightcone
            │   │   │   ├── AllIndependentComponents
            │   │   │   │   └── IndependentComponents.m
            │   │   │   ├── AllIndependentComponents.m
            │   │   │   ├── ConstraintComponentToLightcone.m
            │   │   │   ├── MakeConstraintComponentList.m
            │   │   │   ├── MakeFreeSourceVariables
            │   │   │   │   └── DefFreeSourceVariables.m
            │   │   │   ├── MakeFreeSourceVariables.m
            │   │   │   └── RescaleNullVector.m
            │   │   ├── ConstructLightcone.m
            │   │   ├── ConvertLightcone
            │   │   │   ├── ConstrainInLightcone.m
            │   │   │   ├── ExaminePoleOrder
            │   │   │   │   ├── BasicNullResidue.m
            │   │   │   │   ├── ExtractSecularEquation.m
            │   │   │   │   ├── MasslessAnalysisOfTotal
            │   │   │   │   │   └── ExtractPart.m
            │   │   │   │   ├── MasslessAnalysisOfTotal.m
            │   │   │   │   ├── NullResidue
            │   │   │   │   │   ├── AssistedResidue
            │   │   │   │   │   │   ├── NumericalLowestOrder
            │   │   │   │   │   │   │   └── LowestOrder.m
            │   │   │   │   │   │   ├── NumericalLowestOrder.m
            │   │   │   │   │   │   ├── ObtainTermReplacements.m
            │   │   │   │   │   │   ├── ProcessPart.m
            │   │   │   │   │   │   └── ResidueFormula.m
            │   │   │   │   │   └── AssistedResidue.m
            │   │   │   │   ├── NullResidue.m
            │   │   │   │   ├── PerformReduction.m
            │   │   │   │   ├── PrepareReduction.m
            │   │   │   │   ├── ReparameteriseSources
            │   │   │   │   │   ├── ExtractDenominator
            │   │   │   │   │   │   └── DenominatorOfElement.m
            │   │   │   │   │   ├── ExtractDenominator.m
            │   │   │   │   │   └── ExtractReparameterisationMatrix.m
            │   │   │   │   ├── ReparameteriseSources.m
            │   │   │   │   └── TestDependency.m
            │   │   │   ├── ExaminePoleOrder.m
            │   │   │   ├── ExpressInLightcone.m
            │   │   │   ├── FullyCanonicalise.m
            │   │   │   ├── FullyExpandSources.m
            │   │   │   ├── MakeSaturatedMatrix.m
            │   │   │   ├── MatrixFromSymbolic.m
            │   │   │   ├── MatrixToSymbolic.m
            │   │   │   ├── Repartition.m
            │   │   │   └── SourcesToComponents.m
            │   │   └── ConvertLightcone.m
            │   ├── ConstructMasslessAnalysis.m
            │   ├── ConstructSaturatedPropagator
            │   │   ├── ConjectureInverse
            │   │   │   ├── DistributeConjugate.m
            │   │   │   ├── IntermediateRules.m
            │   │   │   ├── MakeSymbolic.m
            │   │   │   ├── ManualPseudoInverse
            │   │   │   │   ├── CarefullyOrthogonalise.m
            │   │   │   │   └── ParameterisedNullVectorQ.m
            │   │   │   ├── ManualPseudoInverse.m
            │   │   │   ├── UnmakeSymbolic
            │   │   │   │   ├── ConsolidateFinalElement.m
            │   │   │   │   ├── ConsolidateUnmakeSymbolic
            │   │   │   │   │   ├── GrabExpression.m
            │   │   │   │   │   └── SimplifyIfSmall.m
            │   │   │   │   ├── ConsolidateUnmakeSymbolic.m
            │   │   │   │   ├── GradualExpandSubTask
            │   │   │   │   │   └── GradualExpand.m
            │   │   │   │   ├── GradualExpandSubTask.m
            │   │   │   │   ├── InitialExpand
            │   │   │   │   │   └── BatchExpanded.m
            │   │   │   │   └── InitialExpand.m
            │   │   │   └── UnmakeSymbolic.m
            │   │   └── ConjectureInverse.m
            │   ├── ConstructSaturatedPropagator.m
            │   ├── ConstructSourceConstraints
            │   │   ├── ConjectureNullSpace
            │   │   │   ├── CleanNullVector.m
            │   │   │   ├── EnsureLinearInCouplings.m
            │   │   │   ├── RemoveReferencesToMomentum
            │   │   │   │   └── CreateList.m
            │   │   │   ├── RemoveReferencesToMomentum.m
            │   │   │   ├── SymbolicNullSpace
            │   │   │   │   ├── CommonNullVector
            │   │   │   │   │   └── IsNullVectorOfSpace.m
            │   │   │   │   ├── CommonNullVector.m
            │   │   │   │   └── MinimalExampleCaseNullSpace.m
            │   │   │   └── SymbolicNullSpace.m
            │   │   ├── ConjectureNullSpace.m
            │   │   ├── NonTrivialDot.m
            │   │   └── ToCovariantForm.m
            │   ├── ConstructSourceConstraints.m
            │   ├── ConstructSpectrograph
            │   │   ├── AbbreviationRow.m
            │   │   ├── AllParticleRows
            │   │   │   ├── ParticleRows
            │   │   │   │   ├── ParticleRow.m
            │   │   │   │   ├── ProtractList.m
            │   │   │   │   └── StripFactors.m
            │   │   │   └── ParticleRows.m
            │   │   ├── AllParticleRows.m
            │   │   ├── CLICallStack
            │   │   │   └── Status.m
            │   │   ├── CLICallStack.m
            │   │   ├── GUICallStack.m
            │   │   ├── ShowIfSmall.m
            │   │   ├── SourceConstraintRows.m
            │   │   ├── TheoryRows.m
            │   │   ├── UnitarityRow.m
            │   │   ├── UnresolvedPoleRows
            │   │   │   └── UnresolvedPoleRow.m
            │   │   ├── UnresolvedPoleRows.m
            │   │   ├── WignerGrid
            │   │   │   ├── Abbreviate
            │   │   │   │   ├── ReplaceRadicals
            │   │   │   │   │   └── CarefulReplace.m
            │   │   │   │   └── ReplaceRadicals.m
            │   │   │   ├── Abbreviate.m
            │   │   │   ├── GetAbbreviationRules
            │   │   │   │   ├── ShortEnoughQ.m
            │   │   │   │   └── StripRadicals.m
            │   │   │   └── GetAbbreviationRules.m
            │   │   └── WignerGrid.m
            │   ├── ConstructSpectrograph.m
            │   ├── ConstructUnitarityConditions
            │   │   └── TimedReduce.m
            │   ├── ConstructUnitarityConditions.m
            │   ├── ConstructWaveOperator
            │   │   ├── ConstructOperator
            │   │   │   └── GetHermitianPart.m
            │   │   ├── ConstructOperator.m
            │   │   └── FourierLagrangian.m
            │   ├── ConstructWaveOperator.m
            │   ├── UpdateTheoryAssociation.m
            │   ├── ValidateLagrangian.m
            │   ├── ValidateMaxLaurentDepth.m
            │   └── ValidateTheoryName.m
            ├── ParticleSpectrum.m
            ├── ReloadPackage
            │   ├── CallStackBegin.m
            │   ├── CallStackEnd.m
            │   ├── Colours.m
            │   ├── DefGeometry.m
            │   ├── Diagnostic.m
            │   ├── MonitorParallel
            │   │   ├── ParallelGrid
            │   │   │   └── RaggedBlock.m
            │   │   ├── ParallelGrid.m
            │   │   └── ReMagnify.m
            │   ├── MonitorParallel.m
            │   ├── NewFramed.m
            │   ├── NewParallelSubmit.m
            │   ├── StackSetDelayed
            │   │   ├── NameOfFunction.m
            │   │   └── StackStrip.m
            │   ├── StackSetDelayed.m
            │   ├── ToNewCanonical.m
            │   └── Vectorize.m
            └── ReloadPackage.m

61 directories, 196 files
\end{lstlisting}
\paragraph*{Installing the package} To make the installation, the sources should simply be copied alongside the other \xAct{} sources. If the installation of \xAct{} is global, one can use:
\begin{lstlisting}[language=Special]
[user@system ~]$\dollar$ cd PSALTer/xAct
[user@system xAct]$\dollar$ sudo cp -r PSALTer /usr/share/Mathematica/Applications/xAct/
\end{lstlisting}
Or, for a local installation of \xAct{}, one may use:
\begin{lstlisting}[language=Special]
[user@system xAct]$\dollar$ cp -r PSALTer ~/.Mathematica/Applications/xAct/
\end{lstlisting}
In the latest version of \PSALTer{}, the additional dependencies \Inkscape{} and \RectanglePacking{} have been removed. It may also be helpful to run \PSALTer{} with a stable internet connection, since some of the functions used may need to be imported from the online Wolfram Function Repository --- this process should happen automatically. All the details provided in this appendix may change with future versions of \PSALTer{}. Up-to-date installation instructions will be maintained at \href{https://github.com/wevbarker/PSALTer}{\texttt{github.com/wevbarker/PSALTer}}.

\addcontentsline{toc}{section}{References}
\bibliography{Manuscript,NotINSPIRE}

\end{document}